\documentclass[11pt]{article}

\usepackage[utf8]{inputenc}
\usepackage{a4wide}
\usepackage{graphicx}
\usepackage{placeins}
\usepackage[hmarginratio={1:1},     
  vmarginratio={1:1},     
  textwidth=16cm,        
  textheight=24cm,
  heightrounded,]{geometry}
  
\usepackage{amsmath,accents}
\usepackage{amsthm}
\usepackage{amssymb}
\usepackage{extarrows}

\usepackage{mathrsfs, mathtools}
\usepackage{dsfont}

\newcommand{\doublearrow}{\mathrel{\substack{\longrightarrow\\[-0.6ex]
                      \longrightarrow}}}

\numberwithin{equation}{section}

\usepackage[]{hyperref}
\usepackage{slashed}
\usepackage[percent]{overpic} 
\newtheoremstyle{styl1}
  {13pt}
  {13pt}
  {}
  {}
  {\itshape}
  {.}
  {.5em}
  {}
\theoremstyle{styl1}

\newtheorem{definition}[equation]{Definition}
\newtheorem{example}[equation]{Example}
\newtheorem{remark}[equation]{Remark}

\newtheorem{lemma}[equation]{Lemma}
\newtheorem{theorem}[equation]{Theorem}
\newtheorem{proposition}[equation]{Proposition}

\usepackage{tikz}
\usetikzlibrary{matrix,arrows,decorations.pathmorphing}
\usepackage{tikz-cd}

\tikzcdset{%
    triple line/.code={\tikzset{%
        double distance = 3pt, 
        double=\pgfkeysvalueof{/tikz/commutative diagrams/background color}}},
    quadruple line/.code={\tikzset{%
        double distance = 5.3pt, 
        double=\pgfkeysvalueof{/tikz/commutative diagrams/background color}}},
    Rrightarrow/.code={\tikzcdset{triple line}\pgfsetarrows{tikzcd implies cap-tikzcd implies}},
    RRightarrow/.code={\tikzcdset{quadruple line}\pgfsetarrows{tikzcd implies cap-tikzcd implies}}
}

\newcommand*{\tarrow}[2][]{\arrow[Rrightarrow, #1]{#2}\arrow[dash, shorten >= 0.5pt, #1]{#2}}

\newcommand{\R}{\mathbb{R}}
\newcommand{\C}{\mathbb{C}}
\newcommand{\Z}{\mathbb{Z}}

\newcommand{\N}{\mathbb{N}}

\newcommand{\Ca}{\mathcal{C}}
\newcommand{\Fa}{\mathcal{F}}
\newcommand{\Ga}{\mathcal{G}}
\newcommand{\Ha}{\mathcal{H}}
\newcommand{\Ka}{\mathcal{K}}
\newcommand{\Aa}{\mathcal{A}}
\newcommand{\Ja}{\mathcal{J}}

\newcommand{\Fscr}{\mathscr{F}}
\newcommand{\Bscr}{\mathscr{B}}
\newcommand{\Vscr}{\mathscr{V}}
\newcommand{\Gscr}{\mathscr{G}}

\newcommand{\sft}{{\mathsf{t}}}
\newcommand{\sfs}{{\mathsf{s}}}
\newcommand{\sfB}{{\mathsf{B}}}
\newcommand{\sfT}{{\mathsf{T}}}

\newcommand{\End}{\mathsf{End}}
\newcommand{\Hom}{\mathsf{Hom}}

\newcommand{\id}{\text{id}}
\newcommand{\Sym}{{\mathsf{Sym}}}
\newcommand{\fvs}{{\mathsf{Vect}_\C}}
\newcommand{\D}{\slashed{D}}
\newcommand{\Tvs}{{\mathsf{2Vect}_\C}}

\newcommand{\Cob}{{\mathsf{Cob}}}
\newcommand{\CobF}{{{\mathsf{Cob}}^{\mathscr{F}}_{d,d-1,d-2}}}

\newcommand{\Hilb}{{\mathsf{Hilb}_\C}}

\newcommand{\e}{\,\mathrm{e}\,}
\newcommand{\ii}{\,\mathrm{i}\,}
\newcommand{\diff}{\mathrm{d}}
\newcommand{\mbf}{\boldsymbol}

\begin{document}

\begin{flushright}
\small
EMPG--17--09
\end{flushright}

\vspace{10mm}

\begin{center}
	\textbf{\LARGE{Extended quantum field theory, index theory\\[4pt] and the parity anomaly}}\\
	\vspace{1cm}
	{\large Lukas Müller} \ \ and \ \ {\large Richard J. Szabo}

\vspace{5mm}

{\em Department of Mathematics\\
Heriot-Watt University\\
Colin Maclaurin Building, Riccarton, Edinburgh EH14 4AS, U.K.}\\
and {\em Maxwell Institute for Mathematical Sciences, Edinburgh, U.K.}\\
and {\em The Higgs Centre for Theoretical Physics, Edinburgh, U.K.}\\
Email: \ {\tt lm78@hw.ac.uk \ , \ R.J.Szabo@hw.ac.uk}

\vspace{2cm}

\end{center}

\begin{abstract}
\noindent
We use techniques from functorial quantum field theory to provide a
geometric description of the parity anomaly in fermionic systems
coupled to background gauge and gravitational fields on
odd-dimensional spacetimes. We give an explicit construction of a
geometric cobordism bicategory which incorporates general background
fields in a stack, and together with the theory of symmetric monoidal
bicategories we use it to provide the concrete forms of invertible extended quantum
field theories which capture anomalies in both the path integral and
Hamiltonian frameworks. Specialising this situation by using the
extension of the Atiyah-Patodi-Singer index theorem to manifolds
with corners due to Loya and Melrose, we obtain a new Hamiltonian
perspective on the parity anomaly. We compute explicitly the 2-cocycle
of the projective representation of the gauge symmetry on the quantum
state space, which is defined in a parity-symmetric way by suitably
augmenting the standard chiral fermionic Fock spaces with Lagrangian
subspaces of zero modes of the Dirac Hamiltonian that naturally appear
in the index theorem. We describe the significance of our
constructions for the bulk-boundary correspondence in a large class of
time-reversal invariant gauge-gravity symmetry-protected topological
phases of quantum matter with gapless charged boundary fermions,
including the standard topological insulator in $3+1$ dimensions.
\end{abstract}

\newpage

\tableofcontents

\newpage

\section{Introduction and summary}

The parity anomaly in field theories of fermions coupled to 
gauge and gravitational backgrounds in odd dimensions was discovered over
30 years ago~\cite{Redlich,Niemi,AlvarezGaume}, and has found renewed
interest recently because of its relevance to certain topological
states of quantum matter~\cite{WittenFermionicPathInt,SeibergWitten}. The
purpose of this paper is to re-examine the parity anomaly from the
perspective of functorial quantum field theory, and in particular to
elucidate its appearence in the Hamiltonian framework (see e.g.~\cite{Chang:1986ri}) which has been
largely unexplored. We begin with some preliminary
physics background to help motivate the problem we study. 

\subsection{Anomalies and symmetry-protected topological phases\label{sec:introSPT}}

In recent years considerable progress has been made in condensed
matter physics towards understanding the distinct possible quantum
phases of matter with an energy gap through their universal
long-wavelength properties, and the ensuing interplay between global
symmetries and topological degrees of freedom. In many
instances the effective low-energy (long-range) continuum theory of a
lattice Hamiltonian model can be formulated as a relativistic field theory, which for a gapped phase can be usually reduced to a topological quantum field theory that describes the ground states and their response to external sources; such a gapped phase is known as a `topological phase' of matter. The classic example of this is the integer quantum Hall state and its effective description as a three-dimensional Chern-Simons gauge theory.

A gapped phase $\Psi$ is `short-range entangled'~\cite{Chen:2010gda},
or `invertible'~\cite{FreedAnomalies,FreedHopkins}, if there exists a
gapped phase $\Psi^{-1}$ such that\footnote{Here `$\otimes$' denotes a
`stacking' operation combining gapped phases together which turns them into a
commutative monoid with identity element the trivial phase. The
short-range entangled phases form the abelian group of units in this
monoid. In the corresponding topological quantum field theory, it is
the local tensor product induced by that on the state space.} $\Psi\otimes\Psi^{-1}$ can be deformed to a trivial product
state by an adiabatic transformation of the Hamiltonian without
closing the bulk energy gap. The macroscopic properties of such gapped
states are described by particular kinds of topological quantum field
theories which are also called `invertible' with respect to the tensor
product of vector spaces~\cite{FreedAnomalies}; they have the property that their Hilbert space of quantum states is one-dimensional and all propagators are invertible, in contrast to most quantum field theories. The correspondence between invertible topological field theories and short-range entangled phases of matter is discussed in~\cite{FreedHopkins}.

Some gapped systems are non-trivial as a consequence of intrinsic
topological order or of protection by
a global symmetry group $G$. A short-range entangled state is
`$G$-symmetry-protected' if it can be deformed to a trivial product
state by a $G$-noninvariant adiabatic
transformation~\cite{Chen:2010gda,Chen:2011pg}. The gapped bulk system
is then characterised by gapless boundary states, such as the chiral
quantum Hall edge states, which exhibit gauge or gravitational
anomalies; conversely, a $d-1$-dimensional system whose ground state topological
order is anomalous can only exist as the boundary of a $d$-dimensional
topological phase. While the boundary quantum field theory on its own
suffers from anomalies, the symmetry-protected boundary states are described by considering the anomalous theory `relative' to the higher-dimensional bulk theory, where it becomes a non-anomalous quantum field theory under the `bulk-boundary correspondence'~\cite{Ryu:2010ah,Ryu:2012he} in which the boundary states undertake anomaly inflow from the bulk field theory~\cite{Callan:1984sa,Faddeev:1986pc}. The standard examples are provided by topological insulators which are protected by fermion number conservation and time-reversal symmetry ($G=U(1)\times\mathbb{Z}_2$)~\cite{Hasan:2010xy,Qi:2011zya}. The correspondence between topological field theories and symmetry-protected topological phases of matter is discussed in e.g.~\cite{Wen:2013oza,Kapustin:2015uma,Gaiotto:2015zta,WittenFermionicPathInt}.

In the situations just described, a field theory with anomaly is 
well-defined as a theory living on the boundary of a quantum field
theory in one higher dimension which is invertible. The modern
perspective on quantum field theories is that they should not be
simply considered on a fixed spacetime manifold, but rather on a
collection of manifolds which gives powerful constraints on their
physical quantities; in
applications to condensed matter systems involving fermions alone the
manifolds in question should be spin$^c$
manifolds~\cite{SeibergWitten}. For example, an anomaly which arises due to the introduction of a gauge-noninvariant regularization may be resolved on any given spacetime manifold, but it may not be possible to do this in a way consistent with the natural cutting and gluing constructions of manifolds. This is, for example, the case for the global anomaly at the boundary of a $3+1$-dimensional topological superconductor~\cite{WittenFermionicPathInt}.

The purpose of the present paper is to cast the topological field
theory formulation of the parity anomaly and the $3+1$-dimensional
topological insulator, which was outlined
in~\cite{WittenFermionicPathInt}, into a more general and rigorous
mathematical framework using the language of functorial quantum field
theory, and to interpret the path integral description of the parity anomaly, together
with its cancellation via the bulk-boundary correspondence, in a
Hamiltonian framework more natural for condensed matter physics
applications. Before outlining precisely what we do, let us first
informally review some of the main mathematical background.

\subsection{Anomalies in functorial quantum field theory\label{sec:AFQFT}}

A natural framework for making the constructions discussed in
Section~\ref{sec:introSPT} mathematically precise is through
functorial field theories. The idea is that a $d$-dimensional quantum
field theory should assign to a $d$-dimensional manifold $M$ a complex
number $Z(M)$, its partition function. Heuristically, this number is
given by a path integral of an action functional over the space of
dynamical field configurations on $M$; thus far there is no mathematically
well-defined theory of such path integration in general. Functorial
quantum field theory is an axiomatic approach to quantum field theory
which formalises the properties expected from path integrals. A
quantum field theory should not only assign complex numbers to
$d$-dimensional manifolds, but also a Hilbert space of quantum states
$Z(\Sigma)$ to every $d-1$-dimensional manifold $\Sigma$. Moreover,
the theory should assign a time-evolution operator (propagator)
$Z(\Sigma\times[t_0,t_1])$ to a cylinder over $\Sigma$, satisfying
$Z(\Sigma\times[t_1,t_2])\circ Z(\Sigma\times[t_0,t_1])=
Z(\Sigma\times[t_0,t_2])$. More generally, we assign an operator
$Z(M):Z(\Sigma_-)\to Z(\Sigma_+)$ to every manifold $M$ with a
decomposition of its boundary $\partial M=\Sigma_-\sqcup\Sigma_+$
satisfying an analogous composition law under gluing.

To make this precise, one generalises Atiyah's definition of topological quantum field theories~\cite{Atiyah1988} and Segal's definition of two-dimensional conformal field theories~\cite{Segal1988} to define a functorial quantum field theory as a symmetric monoidal functor\footnote{This definition also allows for non-unitary or reflection positive, its Euclidean analogue, theories. 
See~\cite{FreedHopkins} for an implementation of reflection positivity, which would be the relevant concept for the theories considered in this paper. However, for simplicity we will not consider reflection positivity in this paper.}
\begin{align}
\label{EQ: FunktorQFT}
\mathcal{Z} \colon {\mathsf{Cob}}^\mathscr{F}_{ d, d-1} \longrightarrow \Hilb \ ,
\end{align}
where ${\mathsf{Cob}}^\mathscr{F}_{ d, d-1}$ is a category modelling
physical spacetimes, and $\Hilb$ is the category of complex Hilbert
spaces and linear maps. Loosely speaking, the category
${\mathsf{Cob}}^\mathscr{F}_{ d, d-1}$ contains compact
$d-1$-dimensional manifolds as objects, $d$-dimensional cobordisms as
morphisms, and a further class of morphisms corresponding to
diffeomorphisms compatible with the background fields $\mathscr{F}$ which represent the physical symmetries of the theory.

If $\mathcal{Z}$ is an invertible field theory, described as a functor
$\mathcal{Z} \colon {\mathsf{Cob}}^\mathscr{F}_{ d, d-1}\rightarrow
\Hilb$, then the partition function $Z$ of a $d-1$-dimensional field
theory with anomaly $\mathcal{Z}$ evaluated on a $d-1$-dimensional
manifold $M$ takes values in the one-dimensional vector space
$\mathcal{Z}(M)$, instead of $\C$. We can pick a non-canonical
isomorphism $\mathcal{Z}(M)\cong \C$ to identify the partition
function with a complex number. Furthermore, the group of symmetries
acts (non-trivially) on $\mathcal{Z}(M)$ describing the breaking of
symmetries in the quantum field theory; see Section~\ref{sec:catQFT} for details.
 
To extend this description to the Hamiltonian formalism incorporating the quantum state spaces of a field theory $A$, we instead seek a functor $\mathcal{A}$ that assigns linear categories to $d-2$-dimensional manifolds $\Sigma$ such that the state space $A(\Sigma)$ is an object of $\mathcal{A}(\Sigma)$. In other words, $\mathcal{A}$ should be an extended quantum field theory, i.e. a symmetric monoidal 2-functor 
\[
\mathcal{A} \colon \mathsf{Cob}^{\mathscr{F}}_{d,d-1,d-2}\longrightarrow \Tvs
\]    
appropriately categorifying \eqref{EQ: FunktorQFT}. There are different possible higher replacements of the category of Hilbert spaces, but for simplicity we restrict ourselves to Kapranov-Voevodsky 2-vector spaces~\cite{KV94} in this paper, ignoring the Hilbert space structure altogether.

In an analogous way to the partition function, we want to be able to
identify the state space of a quantum field theory with anomaly in a
non-canonical way with a vector space, i.e. there should be an equivalence of categories $\mathcal{A}(\Sigma)\cong \mathsf{Vect}_{\mathbb{C}}$. 
We enforce the existence of such an equivalence by requiring that $\mathcal{A}$ is an invertible extended field theory, i.e. it is invertible with respect to the Deligne tensor product. We can then define a quantum field theory with anomaly in a precise manner as a natural symmetric monoidal 2-transformation 
\[
A\colon \mbf1\Longrightarrow \mathsf{tr}\mathcal{A}
\]   
between a trivial field theory $\mbf 1$ and a certain truncation of
$\mathcal{A}$; this definition is detailed in
Section~\ref{sec:anomaliesextended}.  In this formalism one can in
principle compute the 2-cocycles of the projective representation of
the gauge group characterising the anomalous action on the quantum
states~\cite{MonnierHamiltionianAnomalies}. This description of anomalies in
terms of relative field theories~\cite{RelativeQFT} is closely related
to the twisted quantum field theories
of~\cite{Stolz:2011zj,Johnson-Freyd:2017ykw}, the difference being
that the twist they use does not have to come from a full field
theory.

The extended quantum field theories for some physically relevant anomalies are more or less explicitly known. Some noteworthy examples are:
\begin{itemize}
\item
Dai-Freed theories describing chiral anomalies in odd dimensions $d$ have been sketched in~\cite{MonnierHamiltionianAnomalies}. They are an extended version of the field theories constructed in \cite{DaiFreed}.
\item
Wess-Zumino theories describing the anomaly in self-dual field theories have been constructed in~\cite{MonnierHamiltionianAnomalies}.
\item
The theory describing the anomaly of the worldvolume theory of
M5-branes has been constructed as an unextended quantum field theory in~\cite{MonnierMAnomalies}.
\item
The anomaly field theory corresponding to supersymmetric quantum
mechanics is described in~\cite{FreedAnomalies}.  
\end{itemize}
This paper adds a further example to this list by giving a precise
construction of an extended quantum field theory in any even dimension
$d$ which encodes the parity anomaly in odd spacetime dimension. We
shall now give an overview of our constructions and findings.

\subsection{Summary of results}
One of the technical difficulties related to extended functorial field theories
is the construction of the higher cobordism category equipped with additional structure.
For cobordisms with tangential structure there exist an $(\infty,n)$-categorical definition~\cite{Lurie2009a,CalaqueScheimbauer}. A categorical version with arbitrary background fields
taking into account families of manifolds has been defined by Stolz and Teichner~\cite{Stolz:2011zj}.
A bicategory of cobordisms equipped with elements of topological stacks is constructed in \cite{schommer2011classification}.

One of the main technical accomplishments of this paper is the explicit
construction of a geometric cobordism bicategory
$\mathsf{Cob}^{\mathscr{F}}_{d,d-1,d-2}$ which includes arbitrary
background fields in the form of a general stack\footnote{In~\cite{RelativeQFT} the more general situation of simplicial sheaves is considered.} $\mathscr{F}$
(Section~\ref{Appendix Geometric bicategories}); although this is only
a slight generalisation of previous constructions, it is still
technically quite complicated, and its explicit form makes all of our
statements precise. Building on this bicategory, we then use
the theory of symmetric monoidal bicategories
following~\cite{schommer2011classification} and the ideas
of~\cite{MonnierHamiltionianAnomalies} to work out the concrete
form of the anomalous quantum field theories sketched in
Section~\ref{sec:AFQFT}; this is described in
Section~\ref{sec:anomaliesextended} and is the first detailed
description of anomalies in extended quantum field theory using the
framework of symmetric monoidal bicategories, as far as we are
aware. The relation to projective representations of the gauge group
in~\cite{MonnierHamiltionianAnomalies}, and its extension to projective groupoid
representations following~\cite{BoundaryConditionsTQFT}, is explained in
Section~\ref{sec:projreps}.

The central part of this paper is concerned with the construction of a
concrete example of this general formalism describing the parity
anomaly in odd spacetime dimensions. As the parity anomaly is related
to an index in one higher
dimension~\cite{Niemi,AlvarezGaume,WittenFermionicPathInt}, this
suggests that quantum field theories with parity anomaly should take
values in an extended field theory constructed from index theory; this
naturally fits in with the classification of topological insulators
and superconductors using index theory and K-theory,
see~\cite{Ertem:2017lni} for a recent exposition of this. We
build such a theory using the index theory for manifolds with corners
developed in~\cite{LoyaMelrose,Loyaindex}, which extends the well-known
Atiyah-Patodi-Singer index theorem~\cite{APS} to manifolds with
corners of codimension~$2$. Our construction produces an extended
quantum field theory $\mathcal{A}_{\rm parity}^\zeta$ depending on a
complex parameter $\zeta\in\C^\times$ in any even spacetime dimension
$d$; for $\zeta=-1$ this theory describes the parity anomaly in odd
spacetime dimension. The details are
contained in Section~\ref{sec:EQFTparity}. 

To exemplify how our constructions fit into the
usual treatments of the parity anomaly from the path integral
perspective, we first consider in
Section~\ref{Categorical case} the simpler construction of an ordinary
(unextended) quantum field theory $\mathcal{Z}_{\rm parity}^\zeta$ based on a geometric cobordism
category $\mathsf{Cob}^{\mathscr{F}}_{d,d-1}$ and the usual
Atiyah-Patodi-Singer index theorem for even-dimensional manifolds with boundary. We show that the definition of the
partition function $Z_{\rm parity}^\zeta$ as a natural symmetric
monoidal transformation implies that the complex number $Z_{\rm
  parity}^\zeta(M)$ transforms under a gauge transformation $\phi:M\to
M$ by multiplication with a 1-cocycle of the gauge group given by $\zeta$ to a power determined by the index
of the Dirac operator on the corresponding mapping cylinder
$\mathfrak{M}(\phi)$. This is precisely the same gauge anomaly at
$\zeta=-1$ that arises from the spectral flow of edge states under
adiabatic evolution signalling the presence
of the global parity
anomaly~\cite{Redlich,AlvarezGaume,WittenFermionicPathInt}, which is
a result of the sign ambiguity in the definition of the fermion
path integral in odd spacetime dimension. We further illustrate how
the bulk-boundary correspondence in this
case~\cite{WittenFermionicPathInt} is captured by the full quantum
field theory~$\mathcal{Z}^\zeta_{\rm parity}$.

A key feature of the Hamiltonian formalism defined by our construction
of the extended quantum field theory $\mathcal{A}_{\rm parity}^\zeta$
is that the index of a Dirac operator on a manifold with corners
depends on the choice of a unitary self-adjoint chirality-odd
endomorphism of the kernel of the induced Dirac operator on all
corners, whose positive eigenspace defines a Lagrangian subspace of
the kernel with respect to its natural symplectic structure. We assemble
all possible choices into a linear category $\Aa_{\rm
  parity}^\zeta(\Sigma)$ assigned to
$d-2$-dimensional manifolds $\Sigma$ by $\Aa^\zeta_{\rm parity}$. The
index theorem for manifolds with corners splits into a sum of a bulk integral over
the Atiyah-Singer curvature form and boundary contributions depending on
the endomorphisms. We use these boundary
terms to define the theory $\Aa_{\rm parity}^\zeta$ on 1-morphisms,
i.e. on $d-1$-dimensional manifolds $M$; the general idea is to use
categorical limits to treat all possible boundary conditions at the same
time. The index theorem then induces a natural transformation between
linear functors, defining the theory $\Aa_{\rm parity}^\zeta$ on
2-morphisms, i.e. on $d$-dimensional manifolds.

A crucial ingredient
in the construction of the invertible extended field theory $\Aa_{\rm
  parity}^\zeta$ in Section~\ref{Section: Extending A} is a natural linear map
$$
\Phi_{T_0,T_1}(M_0,M_1):\Aa_{\rm parity}^\zeta(M_1)\circ \Aa_{\rm
  parity}^\zeta(M_0)(T_0) \longrightarrow \Aa_{\rm
  parity}^\zeta(M_1\circ M_0)(T_0)
$$
for every pair of 1-morphisms $M_0:\Sigma_0\to\Sigma_1$ and
$M_1:\Sigma_1\to\Sigma_2$ with corresponding endomorphisms $T_i$ on
the corner manifolds $\Sigma_i$; it forms the components of a natural
linear isomorphism $\Phi$ which is associative. A lot of information
about the parity anomaly is contained in this map: The
construction of $\Aa_{\rm parity}^\zeta$ allows us to fix endomorphisms
for concrete calculations and still have a theory which is independent of these
choices. Viewing a field
theory with parity anomaly as a theory $A_{\rm parity}^\zeta$ relative
to $\Aa_{\rm parity}^\zeta$ in the sense explained before, we then get a
vector space of quantum states $A_{\rm parity}^\zeta(\Sigma)$ for
every $d-2$-dimensional manifold $\Sigma$; the group of gauge
transformations $\mathsf{Sym}(\Sigma)$ only acts projectively on this
space. Denoting this projective representation by $\rho$, for any pair
of gauge transformations $\phi_1,\phi_2:\Sigma\to\Sigma$ one finds
$$
\rho(\phi_2)\circ\rho(\phi_1)=\Phi_{T_1,T_2}\big(\mathfrak{M}(\phi_1),
\mathfrak{M}(\phi_2)\big) \ \rho(\phi_2\circ\phi_1) \ ,
$$
where $\mathfrak{M}(\phi_i)$ is the mapping cylinder of $\phi_i$. Using
results of~\cite{RelationEtaInvariants,LeschWojciechowski}, we can
calculate the corresponding 2-cocycle $\alpha_{\phi_1,\phi_2}$
appearing in the conventional Hamiltonian description of anomalies~\cite{Faddeev:1984jp,Faddeev:1985iz,Mickelsson:1983xi} in
terms of the action of gauge transformations on Lagrangian subspaces
of the kernel of the Dirac Hamiltonian on $\Sigma$; the explicit
expression can be found in~\eqref{eq:2cocycleexplicit}. 

This explicit description of the projective representation of the
gauge group due to the parity anomaly is new. The only earlier Hamiltonian description of
the global parity anomaly that we are aware of
is the argument of~\cite{Chang:1986ri} for the case of fermions
coupled to a particular 
background $SU(2)$ gauge field in $2+1$-dimensions. There
the second quantized Fock space is constructed in the usual way from
the polarisation of the first quantized Hilbert space into spaces
spanned by the eigenspinors of the one-particle Dirac Hamiltonian on
$\Sigma$ with positive and negative energies. When the Dirac
Hamiltonian has zero modes, a sign ambiguity arises in the
identification of Fock spaces with a constant space, which is a
result of a spectral flow; whence the fermion Fock space only carries
a representation of a $\mathbb{Z}_2$-central extension of the gauge group. In
our general approach, we are able to give a more in-depth description
which also lends a physical interpretation to the Lagrangian
subspaces occuring in the index theorem: While the standard Atiyah-Patodi-Singer index theorem is
crucial for computing the parity anomaly in the
path integral and its cancellation via the bulk-boundary
correspondence, the extra boundary terms that appear in the index
theorem on manifolds with corners enable the
definition of a second quantized Fock space
of the quantum field theory at $\zeta=-1$ which is compatible with
parity symmetry by suitably extending the
standard polarisation by Lagrangian subspaces of the kernel of the
Dirac Hamiltonian on $\Sigma$. As before, the sign
ambiguities arise from the definition of $A_{\rm parity}^{(-1)}$
as a natural symmetric monoidal 2-transformation, and the gauge anomaly computed
by the spectral flow is now completely encoded in the 2-cocycles
$\alpha_{\phi_1,\phi_2}$, which cancel between the bulk and boundary
theories by a mechanism similar to that of the partition function; for
details see Section~\ref{sec:projparity}. 

Finally, although the present paper deals exclusively with systems of
Dirac fermions and the parity anomaly in odd-dimensional spacetimes,
the method we develop for the concrete construction of our extended
field theory can be used in other contexts to build invertible
extended field theories from invariants of manifolds with corners. For
example, our techniques could be applied to primitive homotopy quantum
field theories, and to
Dai-Freed theories which are related to $\eta$-invariants; such a formalism would be based on
the Dai-Freed theorem~\cite{DaiFreed} rather than the Atiyah-Patodi-Singer index
theorem and would enable constructions with chiral or Majorana fermions and
unoriented manifolds, which are applicable to other symmetry-protected
states of quantum matter such as topological
superconductors~\cite{WittenFermionicPathInt} as well as to anomalies in M-theory~\cite{Witten:2016cio}. Another application involves
repeating our constructions with Dirac operators replaced by signature
operators, which would lead to an extended quantum field theory
describing anomalies in Reshetikhin-Turaev theories based on modular
tensor categories~\cite{turaev2010quantum}; these theories should also have
applications to anomalies in M-theory along the lines considered
in~\cite{Sati}.

\subsection{Outline}

The outline of this paper is as follows.

In Section~\ref{Categorical case}, as a warm-up we construct the theory
$\mathcal{A}_{\rm parity}^\zeta$ as a quantum field theory
$\mathcal{Z}_{\rm parity}^\zeta$, i.e. as a symmetric monoidal
functor, and describe how it captures the parity anomaly at the level
of path integrals. Following~\cite{WittenFermionicPathInt}, we provide
some explicit examples related to the standard topological insulator
in $3+1$-dimensions and other topological phases of matter.

In Section~\ref{sec:generalanomalies} we present the general description of anomalies in the
framework of functorial quantum field theory using symmetric monoidal
bicategories. In particular, in Section~\ref{Appendix Geometric
  bicategories} we introduce the geometric cobordism bicategory
${\mathsf{Cob}}^\mathscr{F}_{ d, d-1,d-2}$ with arbitrary background
fields $\mathscr{F}$ in quite some detail. 

The heart of this paper is Section~\ref{sec:EQFTparity} where we
explicitly build an extended quantum field theory describing the
parity anomaly using index theory on manifolds with corners. In
particular, in Section~\ref{sec:projparity} we
give the first detailed account of the parity anomaly in the
Hamiltonian framework, which further elucidates the physical
meaning of some technical ingredients that go into the index
theorem.

Two appendices at the end of the paper include some additional
technical background. In Appendix~\ref{Appendix corners} we collect
some facts about manifolds with corners and b-geometry which are
essential for this paper. In Appendix~\ref{Appendix bicategories} we
review definitions connected to symmetric monoidal bicategories,
mostly in order to fix notation and conventions. 

\subsection{Notation}

Here we summarise our notation and conventions for the convenience of the reader.

Throughout this paper we use the notation $M^{d,i}$ for manifolds with
corners, where $d\in \N$ is the dimension of $M^{d,i}$ and $i$ is the
codimension of its corners; for closed manifolds we abbreviate
$M^{d,0}$ by $M^d$. 

The smooth sections of a vector bundle $E$ over a manifold $M$ are
denoted by $\Gamma (E)$. The (twisted) Dirac operator on a manifold
$M$ equipped with a spin structure and a principal bundle with connection is denoted by $\D_M$. We denote the corresponding twisted spinor bundle by $S_M$. 

By $\mathscr{F}$ we always denote a stack on manifolds of a fixed dimension.
 
We use calligraphic letters
$\mathscr{C},\mathscr{G},\mathscr{B},\dots$ to denote generic
categories, groupoids and bicategories. Strict bicategories are called 2-categories.  We denote by
$\mathrm{Obj}(\mathscr{C})$ the class of objects of
$\mathscr{C}$. (1-)morphisms and 2-morphisms are denoted by
$\rightarrow$ and $\Rightarrow$, respectively; natural transformations
are 2-morphisms in the 2-category of small categories. Modifications, which are 3-morphisms in the 3-category of bicategories, are represented by $\Rrightarrow $. We denote by
$\mathrm{Hom}_{\mathscr{C}}$ the set of morphisms in a small category
$\mathscr{C}$, and by $\mathsf{Hom}_{\mathscr{B}}$ the category of 1-morphisms in a bicategory
$\mathscr{B}$. We write
$\mathsf{s},\mathsf{t}$ for the maps from (1-)morphisms to their source
and target objects, respectively. The symbol $\circ$ denotes
composition of (1-)morphisms and vertical composition of 2-morphisms,
while $\bullet$ denotes horizontal composition of 2-morphisms.

Functors and 2-functors are denoted by calligraphic letters $\mathcal{F},\mathcal{G},\dots$.

For concrete categories and bicategories we use sans serif letters.
We will frequently encounter the following categories and bicategories:
\begin{itemize}
\item $\fvs$: \ The symmetric monoidal category of finite-dimensional
  complex vector spaces.
\item $\Tvs$: \ The symmetric monoidal bicategory of 2-vector spaces
  (see Example~\ref{Def:2Vect}).
\item $\Hilb$: \ The symmetric monoidal category of complex Hilbert spaces.
\item $\mathsf{Grpd}$: \ The 2-category of small groupoids.
\item $\Cob^{\mathscr{F}}_{d,d-1}$: \ The symmetric monoidal category of
  $d$-dimensional geometric cobordisms with background fields
  $\mathscr{F}$ (see Section~\ref{sec:catQFT}).
\item $\CobF$: \ The symmetric monoidal bicategory of $d$-dimensional
  geometric cobordisms with background fields $\mathscr{F}$ (see
  Section~\ref{Appendix Geometric bicategories}).
\end{itemize}

\section{Quantum field theory and the parity anomaly}\label{Categorical case}

The parity anomaly appears in certain field theories with time-reversal (or space-reflection) symmetry involving fermions coupled to gauge fields and gravity in spacetimes of odd dimension $2n-1$ if, after quantisation, there is no consistent way to make the path integral real. The phase ambiguity appears in a controlled manner and can be understood by regarding the original quantum field theory as living on the boundary of another quantum field theory defined in $d=2n $ dimensions with the same global symmetry in the bulk: We say that the anomalous field theory in $d-1$ dimensions takes values in a non-anomolous quantum field theory in $d$ dimensions, since the phase
ambiguity of the boundary theory is cancelled by the phase of the bulk system. 
As a warmup, in this section we explain a simple functorial perspective that captures the relation between the parity anomaly in $2n-1$ dimensions and (unextended) quantum field theories in $d=2n$ dimensions, based on the index theory of the Dirac operators which feature in field theories with Dirac fermions. This categorical formalism captures the anomaly only at the level of partition functions and ignores the action on the Hilbert space of quantum states; the latter will be incorporated in subsequent sections by extending the underlying source and target categories to bicategories.

\subsection{Geometric cobordisms and quantum field theories}\label{sec:catQFT}

We begin by explaining the formalism of functorial quantum field theories that we shall use in this paper, and the relation between invertible field theories and anomalies. Let us fix a gauge group $G$ with Lie algebra $\mathfrak{g}$, and a finite-dimensional unitary representation $\rho_G \colon G \rightarrow \text{Aut}(V)$ of $G$. 
We define a geometric source category ${\mathsf{Cob}}^\mathscr{F}_{ d, d-1}$ whose objects are
closed smooth $d-1$-dimensional manifolds $M^{d-1}$ equipped with a Riemannian metric $g_{M^{d-1}}$, an orientation, a spin structure and a principal $G$-bundle $\pi_{M^{d-1}} \colon P_{M^{d-1}} \rightarrow M^{d-1}$ with connection $A_{M^{d-1}} \in \Gamma ({T}^\ast P_{M^{d-1}} \otimes \mathfrak{g})$; this  specifies the background field content $\mathscr{F}$ and we call the resulting objects `$\mathscr{F}$-manifolds'. 
We think of an object $M^{d-1}$ as sitting in the germ of $d$-dimensional manifolds of
the form $M^{d-1}\times (-\epsilon,\epsilon)$ with all structures extended as products.

A diffeomorphism of $\mathscr{F}$-manifolds $M^{d-1}_1$ and $M^{d-1}_2$, or an `$\Fscr$-diffeomorphism', consists of an orientation-preserving smooth isometry $\phi \colon M^{d-1}_1 \rightarrow M^{d-1}_2 $ of the underlying manifolds, together with bundle isomorphisms from the spinor bundle and principal $G$-bundle on $M^{d-1}_1$ to the pullbacks along $\phi$ of the corresponding bundles over $M^{d-1}_2$ that preserves the Levi-Civita connection on the spinor bundles, and the connections $A_{M^{d-1}_1}$ and $A_{M^{d-1}_2}$. 

There are then two types of morphisms in ${\mathsf{Cob}}^\mathscr{F}_{ d, d-1}$. The first type of morphisms are given by equivalence classes of compact $d$-dimensional manifolds $M^{d,1}$ endowed with $\mathscr{F}$-fields up to $\mathscr{F}$-diffeomorphisms preserving collars, a decomposition of their boundary $ \partial M^{d,1} = \partial_- M^{d,1} \sqcup \partial_+ M^{d,1}$, and collars $M_\pm^{d,1}$ near the boundary components $\partial_\pm M^{d,1}$ such that the $\mathscr{F}$-fields are of product structure on $M_\pm^{d,1}$. Such a manifold $M^{d,1}$ describes a morphism from $M^{d-1}_1 $ to $M^{d-1}_2 $ if it comes with diffeomorphisms of $\mathscr{F}$-manifolds $\varphi_- \colon M^{d-1}_1 \times [0,\epsilon_1) \rightarrow M_-^{d,1} $ and $\varphi_+ \colon M^{d-1}_2 \times (-\epsilon_2,0] \rightarrow M_+^{d,1} $ for fixed real numbers $\epsilon_i>0$. We refer to these geometric cobordisms as `regular morphisms'.
Composition is defined by gluing along boundaries using the collars
and their trivialisations $\varphi_\pm$ as described in
Appendix~\ref{Section: Gluing} (see Figure~\ref{Fig:Sketch
  Composition}); note that the smooth structure on glued manifolds
depends on the choice of collars~\cite[Section~1.3]{kock}. This
composition is strictly associative. 

\begin{figure}[hbt]
\begin{overpic}[width=15.5cm,
scale=0.3]{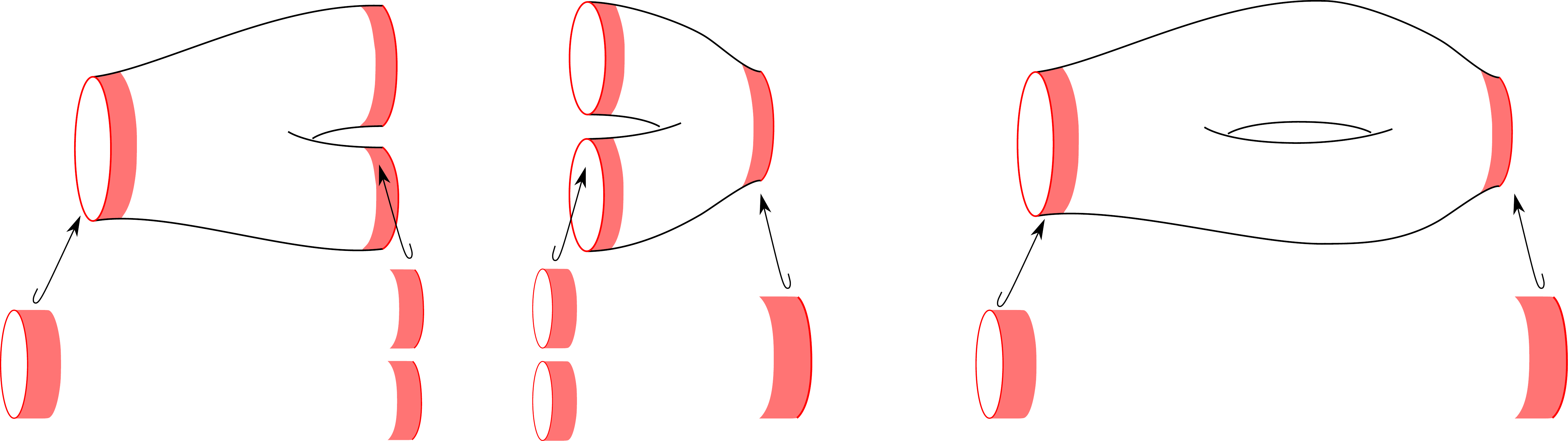}
\put(13,22){\scriptsize{$M^{d,1}_1$}}
\put(42,22){\scriptsize$M^{d,1}_2$}
\put(5,11){\scriptsize$\varphi_{1-}$}
\put(26,15){\scriptsize$\varphi_{1+}$}
\put(33,15){\scriptsize$\varphi_{2-}$}
\put(51,12){\scriptsize$\varphi_{2+}$}
\put(67,11){\scriptsize$\varphi_{1-}$}
\put(99,12){\scriptsize$\varphi_{2+}$}
\put(78,22){\scriptsize$M^{d,1}_2  \circ M^{d,1}_1$}
\put(0,-2){\scriptsize$M^{d-1}_1 \hspace{-0.14cm} \times\hspace{-0.08cm} [0,\epsilon_1)$}
\put(17,-2){\scriptsize$M^{d-1}_2 \hspace{-0.14cm} \times\hspace{-0.08cm} (-\epsilon_2,0]$}
\put(31,-2){\scriptsize$M^{d-1}_2 \hspace{-0.14cm} \times\hspace{-0.08cm} [0,\epsilon_3)$}
\put(45,-2){\scriptsize$M^{d-1}_3 \hspace{-0.14cm} \times\hspace{-0.08cm} (-\epsilon_4,0]$}
\put(60,-2){\scriptsize$M^{d-1}_1 \hspace{-0.14cm} \times\hspace{-0.08cm} [0,\epsilon_1)$}
\put(90,-2){\scriptsize$M^{d-1}_3 \hspace{-0.14cm} \times\hspace{-0.08cm} (-\epsilon_4,0]$}
\end{overpic}
\bigskip
\caption{\small Illustration of two composable regular morphisms in ${\mathsf{Cob}}^\mathscr{F}_{ d, d-1}$ (on the left) and their composition (on the right).}
\label{Fig:Sketch Composition}
\end{figure}

The second type of morphisms are diffeomorphisms of
$\mathscr{F}$-manifolds $\phi \colon M^{d-1}_1 \rightarrow M^{d-1}_2$;
they incorporate symmetries of the background fields including
internal symmetries such as gauge
symmetries. We may regard these morphisms as zero length limits of mapping cylinders, and hence we refer to them as `limit morphisms'.
Composition of limit morphisms is given by concatenation of $\mathscr{F}$-diffeomorphisms. The composition of a limit morphism $\phi \colon  M^{d-1}_1 \rightarrow M^{d-1}_2  $ with a regular morphism $M^{d,1}$ is given by precomposition with the extension of $\phi $ to $\phi' \colon M^{d-1}_1 \times [0,\epsilon_1) \rightarrow M^{d-1}_2 \times [0,\epsilon_1) $; this composition affects only the map $\varphi_-$. The composition of a regular morphism with a limit morphism is defined in a similar way, affecting only $\varphi_+$. 

The disjoint union of manifolds makes ${\mathsf{Cob}}^\mathscr{F}_{ d, d-1}$ into a symmetric monoidal category with monoidal unit $1$ given by the empty manifold $\emptyset$.

\begin{definition}
\label{Definition: Field theory}
A \underline{$d$-dimensional functorial quantum field theory} (or \underline{quantum field theory} for short) with background fields $\mathscr{F}$ is a symmetric monoidal functor
$$
\mathcal{Z} \colon {\mathsf{Cob}}^\mathscr{F}_{ d, d-1} \longrightarrow \Hilb \ ,
$$ 
where $\Hilb$ is the symmetric monoidal category of complex Hilbert spaces and linear maps under tensor product.
\end{definition}
\begin{remark}
It is unclear to us to which extent physical examples of quantum field theories
fit into this framework. See~\cite{Segal} for a discussion of this definition. 
However, it is enough to capture the field theories 
related to anomalies, which suffice for this paper.   
\end{remark}
The simplest example of a quantum field theory is the trivial theory
$\mbf1 \colon {\mathsf{Cob}}^\mathscr{F}_{ d, d-1} \rightarrow \Hilb$
sending every object to the monoidal unit $\C$ in $\Hilb$ and every
morphism to the identity map on $\C$. Given two quantum field theories
$\mathcal{Z}_1 \colon {\mathsf{Cob}}^\mathscr{F}_{ d, d-1} \rightarrow
\Hilb$ and $\mathcal{Z}_2 \colon {\mathsf{Cob}}^\mathscr{F}_{ d, d-1}
\rightarrow \Hilb$, their tensor product $\mathcal{Z}_1 \otimes
\mathcal{Z}_2$ can be defined locally. We ignore technical subtleties
related to the tensor product of infinite-dimensional Hilbert spaces
by tacitly assuming that all vector spaces are finite-dimensional;
see~\cite{MonnierHamiltionianAnomalies} for a discussion of how this
assumption fits in with the infinite-dimensional state spaces
that typically appear in quantum field theory. Using the tensor product one can define the class of quantum field theories relevant for the description of anomalies.
\begin{definition}
A quantum field theory $\mathcal{Z}$ is \underline{invertible} if there exists a quantum field theory $\mathcal{Z}^{-1}$ and a natural symmetric monoidal isomorphism from $\mathcal{Z}\otimes \mathcal{Z}^{-1}$ to $\mbf1$.
\end{definition} 

The modern perspective on anomalous field theories in $d-1$ dimensions
is that they are ``valued'' in invertible quantum field theories in
$d$ dimensions \cite{FreedAnomalies}. We can understand this point of view at the level of the partition function. This requires restriction to the subcategory ${\mathsf{tr}} {\mathsf{Cob}}^\mathscr{F}_{ d, d-1}$ of ${\mathsf{Cob}}^\mathscr{F}_{ d, d-1}$ containing only invertible morphisms, which are precisely the limit morphisms; we call ${\mathsf{tr}} {\mathsf{Cob}}^\mathscr{F}_{ d, d-1}$ the `truncation' of the category ${\mathsf{Cob}}^\mathscr{F}_{ d, d-1}$. We denote by $\mathsf{tr}\mathcal{Z}$ the restriction of the functor $\mathcal{Z}$ to ${\mathsf{tr}} {\mathsf{Cob}}^\mathscr{F}_{ d, d-1}$.
\begin{definition}
\label{Definition path Integral}
A \underline{partition function} of an invertible quantum field theory $\mathcal{Z}$ is a natural symmetric monoidal transformation 
\begin{align*}
Z \colon \mbf1 \Longrightarrow \mathsf{tr}\mathcal{Z} \ .
\end{align*}
\end{definition} 
Unpacking this definition, we get for every closed $d-1$-dimensional
manifold $M^{d-1}$ a linear map $Z(M^{d-1})\colon \C
=\mbf1(M^{d-1}) \rightarrow \mathcal{Z}(M^{d-1}) $, such that the diagram
\[
\begin{tikzcd}
\C \ar[d, swap, "{\text{id}}"] \ar{rrr}{Z(M^{d-1})} & & & \mathcal{Z}(M^{d-1}) \ar{d}{\mathcal{Z}(\phi)} \\
 \C \ar[rrr, swap, "{Z(\phi( M^{d-1}))}"] & & & \mathcal{Z}\big(\phi( M^{d-1})\big)
\end{tikzcd}
\] 
commutes for all limit morphisms $\phi$. Since
$\mathcal{Z}$ is an invertible field theory, $\mathcal{Z}(M^{d-1})$ is
a one-dimensional vector space and so isomorphic to $\C$, though not
necessarily in a canonical way. This translates into an ambiguity in
the definition of the partition function as a complex number, which is
the simplest manifestation of an anomaly. 

\begin{remark}
Following \cite{FreedAnomalies,MonnierHamiltionianAnomalies}, consider
a general stack $\mathscr{F}$ and an arbitrary invertible quantum
field theory $\mathcal{L} \colon
{\mathsf{Cob}}_{d,d-1}^{\mathscr{F}}\rightarrow \mathsf{Hilb}_\C$ (see
Section~\ref{sec:generalanomalies} below for precise definitions).  
Fixing a $d-1$-dimensional closed manifold $M^{d-1}$, for every choice
$\mathsf{f}\in \mathscr{F}(M^{d-1})$ we get a one-dimensional vector
space $\mathcal{L}(M^{d-1},\mathsf{f})$. If we assume that
$\mathscr{F}(M^{d-1})$ is a manifold, then it makes sense to require
that $\mathcal{L}(M^{d-1}, \, \cdot \, )\rightarrow
\mathscr{F}(M^{d-1})$ is a line bundle. In this case the partition
function gives rise to a section of this line bundle. This heuristic
reasoning reproduces the more common geometric picture of anomalies in
terms of the absence of canonical trivialisations of line bundles over the parameter space of the field theory (see e.g. \cite{NashBook}).
\end{remark}

Even when it is
well-defined, the partition function may fail to be invariant under a
limit automorphism $\phi$ by a linear isomorphism $\mathcal{Z}(\phi)$
of $\mathcal{Z}(M^{d-1})$ that can define a non-trivial $\C$-valued
1-cocycle of the group of limit automorphisms of $M^{d-1}$, which is
precisely the case of an anomalous symmetry. In this paper we are
interested in the more general case of anomalous symmetries encoded by
non-trivial $\C$-valued 1-cocycles of the groupoid
$\mathsf{trCob}_{d,d-1}^{\mathscr{F}}$: A partition function $Z:\mbf1 \Rightarrow \mathsf{tr}\mathcal{Z}$ induces,
after picking non-canonical isomorphisms, a groupoid homomorphism
$\chi^Z: \mathsf{trCob}_{d,d-1}^{\mathscr{F}}\to \C \,/\!\! /\, \C^\times$.

From a physical viewpoint it is
natural to require that these quantum field theories are local, which
leads to fully extended field theories. These can be classified in the
case of topological quantum field theories by fully dualizable objects
in the symmetric monoidal target
$(\infty,d)$-category~\cite{Lurie2009a}. If the theory is moreover
invertible then a classification using cobordism spectra and stable homotopy theory is possible~\cite{FreedHopkins}. However, in subsequent
sections we will extend the quantum field theory describing anomalies only
up to codimension $2$; this has the advantage that all concepts are
well-defined and all technical difficulties can be handled rather explicitly. 

\subsection{The quantum field theory $\mathcal{Z}_{\rm parity}^\zeta$\label{sec:parityQGT}}

Let us now construct a $d$-dimensional quantum field theory $\mathcal{Z}_{\rm parity}^\zeta$ encoding the parity anomaly in $d-1$ dimensions. 
The theory $\mathcal{Z}_{\rm parity}^\zeta \colon {\mathsf{Cob}}^\mathscr{F}_{ d, d-1} \rightarrow \Hilb$ assigns the one-dimensional vector space $\C$ to every object $M^{d-1}$:
$$
\mathcal{Z}_{\rm parity}^\zeta(M^{d-1})=\C \ .
$$

The background field content $\mathscr{F}$ together with the representation $\rho_G$ defines a Dirac operator $\slashed{D }_{M^{d,1}}$ on every manifold $M^{d,1}$ corresponding to a regular morphism in ${\mathsf{Cob}}^\mathscr{F}_{ d, d-1}$. Since $d$ is even, the twisted spinor bundle $S_{M^{d,1}}= S_{M^{d,1}}^+\oplus S^-_{M^{d,1}}$ splits into bundles $S_{M^{d,1}}^\pm$ of positive and negative chirality spinors; the Dirac operator is odd with respect to this $\Z_2$-grading. 
On a closed manifold $M^d$, the chiral Dirac operator $\slashed{D
}^+_{M^{d}} \colon H^1(S _{M^{d}}^+)\rightarrow L^2(S _{M^{d}}^-)$ is
a first order elliptic differential operator, where $H^1(S
_{M^{d}}^+)$ is the first Sobolev space of sections of $S _{M^{d}}^+$,
i.e. spinors $\Psi$ whose image $\slashed{D }^+_{M^{d}}\Psi$ is square-integrable, and $L^2(S _{M^{d}}^-)$ is the Hilbert space of square-integrable sections of $S _{M^{d}}^-$; the integration is with respect to the Hermitian structure on $S _{M^{d}}$ induced by the metric and the unitary representation $\rho_G$. Every elliptic operator acting on sections of a vector bundle of finite rank over a closed manifold $M^{d}$ is Fredholm, so that we can define a map $
\mathcal{Z}_{\rm parity}^\zeta(M^{d})\colon \C \rightarrow \C$ by $z\mapsto \zeta^{\text{ind}(\slashed{D }^+_{M^{d}})} \cdot z$,
where the index of a Fredholm operator $D$ is defined by 
$$
\text{ind}(D)= \text{dim}\,\text{ker}(D)-\text{dim}\,\text{coker}(D)
$$ 
and $\zeta \in \C^\times$ is a non-zero complex parameter. We would like to extend this map to manifolds with boundary, but the Dirac operator on a manifold with boundary is never Fredholm. A standard solution to this problem is to attach infinite cylindrical ends to the boundary of $M^{d,1}$.\footnote{This is equivalent to the introduction of
  Atiyah-Patodi-Singer spectral boundary conditions on the spinors~\cite{APS}. We use the method of cylindrical ends here, since it can be generalised to manifolds with corners and gives a natural cancellation of certain terms later on.} We define
\begin{align*}
\hat{M}^{d,1}= M^{d,1} \sqcup_{\partial M^{d,1}} \big(\partial_- M^{d,1} \times (-\infty,0] \sqcup \partial_+ M^{d,1} \times [0,\infty ) \big)  
\ ,
\end{align*} 
where we use the identification of the collars $M_\pm^{d,1}$, which are part of the data of a regular morphism, with open cylinders to glue as discussed in Appendix \ref{Section: Gluing}. We extend all of the background field content $\Fscr$ as products to $\hat{M}^{d,1}$. The structure of the regular morphisms in ${\mathsf{Cob}}^\mathscr{F}_{ d, d-1}$ makes it natural to attach inward and outward pointing cylinders to the incoming and outgoing boundary, respectively, contrary to what is normally done in the index theory literature; this will be crucial for compatibility with composition later on. It is further natural, again in contrast to what is normally done in index theory, to glue in the cylinders along the identification of the collars with cylinders; this means that the gluing could ``twist'' bundles. Alternatively, we could first attach a mapping cylinder for the identification and then an infinite cylinder.

The Dirac operator $\slashed{D }^+_{\hat{M}^{d,1}} \colon H^1(\hat{S}_{\hat M^{d,1}}^+)\rightarrow L^2(\hat{S}_{\hat M^{d,1}}^-)$ is Fredholm if and only if the kernel of the induced Dirac operator on the boundary of $M^{d,1}$ is trivial. If the kernel is non-trivial, then we have to regularize the index in an appropriate way, which corresponds physically to introducing small masses for the massless fermions on $M^{d,1}$. This is done precisely by picking, for every connected component $\partial M^{d,1}_i$ of the boundary, a small number $\alpha_i$ with $0<\alpha_i  < \delta_i$, where $\delta_i$ is the smallest magnitude $|\lambda_i|$ of the non-zero eigenvalues $\lambda_i$ of the induced Dirac operator on $\partial M^{d,1}_i$. Now we can attach weights $\e^{\alpha_i \, s_i}$ to the integration measure on the cylindrical ends, where $s_i$ is the coordinate on the cylinder over $\partial M^{d,1}_i $. Denoting the corresponding weighted Sobolev spaces by $\e^{\alpha\cdot s}H^1(\hat{S}_{\hat M^{d,1}}^+)$ and $\e^{\alpha \cdot s}L^2(\hat{S}_{\hat M^{d,1}}^-)$, we then have\footnote{To be more precise, we have to first attach a mapping cylinder before we can apply~\cite[Theorem 5.60]{MelrosebGeo}.}
\begin{theorem}(\cite[Theorem 5.60]{MelrosebGeo})
 \ The Dirac operator 
$$ 
\slashed{D }^+_{\hat{M}^{d,1}} \colon \e^{\alpha\cdot s}H^1\big(\hat S_{\hat M^{d,1}}^+\big)\longrightarrow \e^{\alpha\cdot s}L^2\big(\hat S_{\hat M^{d,1}}^-\big)
$$ 
is Fredholm and its index is independent of the masses $\alpha_i$.
\end{theorem}

Having at hand a well-defined notion of an index for manifolds with boundaries, we can now define
\begin{align*}
\mathcal{Z}_{\rm parity}^\zeta(M^{d,1})\colon \C  \longrightarrow \C \ , \qquad z \longmapsto \zeta^{\text{ind}(\slashed{D }^+_{\hat{M}^{d,1}})}\cdot z \ . 
\end{align*}
The index can be computed by means of the Atiyah-Patodi-Singer index theorem \cite{APS}, which gives a concrete formula for the index if the attaching of cylindrical ends is taken along identity maps: 
\begin{align}
\text{ind} \big(\slashed{D }^+_{\hat{M}^{d,1}} \big) =
  \int_{M^{d,1}}\,  K_{\rm AS}-\frac{1}{2} \, \Big( \eta \big(
  \slashed{D }_{\partial M^{d,1}}\big) +\dim \ker \big(\slashed{D
  }_{\partial_- M^{d,1}} \big)- \dim \ker \big(\slashed{D }_{\partial_+ M^{d,1}}\big) \Big) \ ,
\label{APS Theorem}
\end{align} 
where the Atiyah-Singer density 
$$
K_{\rm AS}= \text{ch}\big(P_{M^{d,1}} \big)\wedge \widehat{A}\big(TM^{d,1} \big)\big|_d
$$ 
is the homogeneous differential form of top degree in $\Omega^d(M^{d,1})$ occuring in the exterior product of the Chern character of the bundle $P_{M^{d,1}}$ with the $\widehat{A}$-genus of the tangent bundle $TM^{d,1}$. The $\eta $-invariant of the Dirac operator on a closed manifold $M^{d-1}$ of odd dimension calculates the number of positive eigenvalues minus the number of negative eigenvalues of $\slashed{D }_{M^{d-1}}$, and is defined by 
\begin{align*}
\eta \big( \slashed{D }_{M^{d-1}}\big)= \lim_{s\rightarrow 0} \ \sum_{\substack{\lambda \in \text{spec}( \slashed{D }_{M^{d-1}}) \\ \lambda\neq0}} \, \frac{\text{sign}(\lambda)}{|\lambda|^s} \ . 
\end{align*}
The limit here should be understood as the value of the analytic continuation of the meromorphic function $\sum_{\lambda\neq0} \, \tfrac{\text{sign}(\lambda)}{|\lambda|^s}$ at $s=0$; the regularity of this value is proven in \cite{APS}. The sign difference between the dimensions of the kernels in \eqref{APS Theorem} comes from the fact that we attach cylinders with opposite orientation to the incoming and outgoing boundary; this corresponds to a negative sign for the numbers $\alpha_i$ on the outgoing boundary $\partial_+M^{d,1}$ in the version of the Atiyah-Patodi-Singer index theorem given in \cite{MelrosebGeo}. 
The $\eta $-invariant can be reformulated as an integral over the trace of the corresponding heat kernel operator as
\begin{align}
\label{Equation: Reformulation eta invariant}
\eta \big( \slashed{D }_{ M^{d-1}}\big) = \frac{1}{\sqrt{\pi}} \, \int_0^\infty\, t^{-{1}/{2}} \ \text{Tr}\Big(\slashed{D }_{ M^{d-1}} \, \e^{-t\, \slashed{D }{}_{M^{d-1}}^2} \Big) \ \diff t \ .
\end{align}

We assign to a limit morphism $\phi$ the value of $\mathcal{Z}_{\rm parity}^\zeta$ on a corresponding mapping cylinder; in order for $\mathcal{Z}_{\rm parity}^\zeta$ to be well-defined, this construction must then be independent of the length of the mapping cylinder. We prove this as part of
\begin{theorem}\label{A is field theory}
 \ $\mathcal{Z}_{\rm parity}^\zeta:{\mathsf{Cob}}^\mathscr{F}_{ d, d-1}\to\Hilb$ is an invertible quantum field theory.
\end{theorem}
\begin{proof}
The value of $\mathcal{Z}_{\rm parity}^\zeta$ on a mapping cylinder is independent of its length, since the manifolds constructed by attaching cylindrical ends are $\mathscr{F}$-diffeomorphic. This proves that $\mathcal{Z}_{\rm parity}^\zeta$ is well-defined on limit morphisms $\phi$.

If we cut a manifold $M^{d,1}$ along a hypersurface $H$ into two pieces $M_1^{d,1}$ and $M_2^{d,1}$, then from the Atiyah-Patodi-Singer index theorem \eqref{APS Theorem} we get\footnote{For this we need a cylindrical neighbourhood of $H$ on which all of the field content $\Fscr$ is of product form.} 
\begin{align*}
\text{ind}\big(\slashed{D}^+_{\hat M^{d,1}}\big)=\text{ind}\big(\slashed{D}^+_{\hat M_1^{d,1}}\big)+\text{ind}\big(\slashed{D}^+_{\hat M_2^{d,1}} \big) \ ,
\end{align*} 
since the integration is additive and $\eta\big( \slashed{D }_{M^{d-1}}\big) = - \eta\big( \slashed{D }_{-M^{d-1}}\big)$, where $-M^{d-1}$ is the manifold $M^{d-1}$ with the opposite orientation. The contributions from the boundary along which the cutting takes place cancel in (\ref{APS Theorem}) because of the opposite signs of the dimensions of the kernel of the boundary Dirac operator for incoming and outgoing boundaries.

For regular morphisms $M^{d,1}\colon M^{d-1}_1 \rightarrow M^{d-1}_2$
and $N^{d,1}\colon M^{d-1}_2 \rightarrow M^{d-1}_3$, we can cut the
manifold $N^{d,1}\circ M^{d,1}$ into $M^{d,1}$ and $N^{d,1}$, with
half of the collar around $M^{d-1}_2$ removed and mapping cylinders
corresponding to the identification attached. This uses the
description of the gluing process in terms of mapping cylinders (see
Appendix \ref{Section: Gluing}), but as mentioned earlier the index of
such pieces is the same as the index corresponding to a manifold where
the attachment is twisted by the identification of the collars with
cylinders. This implies 
\begin{align*}
\mathcal{Z}_{\rm parity}^\zeta\big(N^{d,1}\circ M^{d,1}\big)= \mathcal{Z}_{\rm parity}^\zeta\big(M^{d,1}\big)\cdot \mathcal{Z}_{\rm parity}^\zeta\big(N^{d,1}\big) \ .
\end{align*}   
This proves that $\mathcal{Z}_{\rm parity}^\zeta$ is a functor, which
is furthermore symmetric monoidal since all our constructions are
multiplicative under disjoint unions. The inverse functor is $\big(\mathcal{Z}_{\rm parity}^\zeta\big)^{-1}=\mathcal{Z}_{\rm parity}^{\zeta^{-1}}$.
\end{proof}

\begin{remark}
It may seem unnatural for $\mathcal{Z}_{\rm parity}^\zeta$ to
assign the one-dimensional vector space $\C $ to every closed $d-1$-dimensional manifold $M^{d-1}$. 
Rather one would expect a complex line generated by all boundary conditions via an inverse limit construction as for example in~\cite{FreedQuinn}. Assigning $\C$ to every closed manifold is 
only possible due to the presence of canonical APS-boundary conditions related to the $L^2$-condition on the non-compact manifolds~$\hat{M}{}^{d,1}$.
\end{remark}

\subsection{Partition functions and symmetry-protected topological phases}

We turn our attention now to the partition function for a quantum field theory describing the parity anomaly.
According to Definition~\ref{Definition path Integral}, it is a
natural symmetric monoidal transformation
$Z_{\rm parity}^\zeta \colon \mbf1 \Rightarrow \mathsf{tr}\mathcal{Z}_{\rm parity}^\zeta$. This yields, for every closed $d-1$-dimensional manifold $M^{d-1}$, a linear map $Z_{\rm parity}^\zeta(M^{d-1})\colon \C \rightarrow \mathcal{Z}_{\rm parity}^\zeta(M^{d-1})=\C $.
 A linear map $Z_{\rm
  parity}^\zeta(M^{d-1}) \colon \C \rightarrow \C$ can be canonically
identified with a complex number $Z_{\rm parity}^\zeta(M^{d-1})\in \C$. Now there is no ambiguity in the
definition of the partition function as a complex number. The essence
of the parity anomaly, like most anomalies associated with the
breaking of a classical symmetry in quantum field theory, is the
lack of invariance of $Z_{\rm parity}^\zeta$ under limit morphisms $\phi$: The naturality of the partition function
implies that it transforms under gauge transformations $\phi $ by
multiplication with a 1-cocycle $\mathcal{Z}_{\rm parity}^\zeta(\phi)\in\C^\times$; note that
in the present context `gauge transformations' also refer to
isometries and isomorphisms of the spinor bundle $S_{M^{d-1}}$. Since
$\mathcal{Z}_{\rm parity}^\zeta$ depends only on topological data,
this multiplication is given by
\begin{eqnarray}
\mathcal{Z}_{\rm parity}^\zeta(\phi) = \zeta^{{\rm ind}(\slashed{D}{}^+_{\mathfrak{M}(M^{d-1},\phi)})}
\label{eq:Zparityphi}\end{eqnarray}
where $\mathfrak{M}(M^{d-1}, \phi)$ is the corresponding mapping torus constructed by identifying the boundary components of $M^{d-1}\times [0,1]$ using $\phi$. 

\begin{example} 
\label{Example:TopologicalInsolature}
We shall now illustrate how the functorial formalism of this section
connects with the more conventional treatments of the parity anomaly
in the physics literature, following~\cite{WittenFermionicPathInt} (see also~\cite{SeibergWitten});
indeed, what mathematicians call `invertible quantum field theories'
are known as `short-range entangled topological phases' to
physicists. A partition function with parity anomaly can be defined by
fixing its value on a representative for every gauge equivalence class of field configurations and applying \eqref{eq:Zparityphi}. 
Now consider the partition function with parity anomaly defined by 
\begin{align*}
Z_{\rm parity}^{(-1)}\big(M^{d-1}\big)=\big|\text{det}\big(\slashed{D }_{M^{d-1}}\big)\big|  
\end{align*} 
for an arbitrary chosen background $(A_{M^{d-1}},g_{M^{d-1}})$ in
every gauge equivalence class, where the definition of the determinant requires a suitable regularization.
Formally, this is the absolute value of the contribution to the path
integral measure from a massless Dirac fermion in $d-1$ dimensions coupled to a background $(A_{M^{d-1}},g_{M^{d-1}})$.
 There is an ambiguity in defining the phase of $Z_{\rm
   parity}^{(-1)}\big(M^{d-1}\big)$. Time-reversal (or space-reflection)
 symmetry forces $Z_{\rm parity}^{(-1)}\big(M^{d-1}\big)$ to be real. Here we
 chose the phase to make the partition function positive at the fixed representative.

From a physical perspective, having set the phase of the partition function at a fixed background $(A_{M^{d-1}},g_{M^{d-1}})$ we can calculate the phase at a gauge equivalent configuration $\phi(A_{M^{d-1}},g_{M^{d-1}})$ by following the path
\begin{equation}
(1-t)\, (A_{M^{d-1}},g_{M^{d-1}})+t\, \phi(A_{M^{d-1}},g_{M^{d-1}}) \ , \quad t\in[0,1]
\label{eq:spectralflow}\end{equation}
in the configuration space of the field theory, and changing the sign
every time an eigenvalue of the Dirac operator crosses through
zero. It is well-known that this spectral flow can be calculated by
the index of the Dirac operator on the corresponding mapping
cylinder. This physical intuition is formalised by the definition
above for $\zeta=-1$: The phase ambiguity is determined by requiring
the partition function to define a natural symmetric
monoidal transformation.

We can preserve gauge invariance by using Pauli-Villars regularization~\cite{WittenFermionicPathInt}, leading to the gauge invariant partition function  
\begin{align*}
Z_{\rm parity}\big(M^{d-1}\big)=\big|\text{det}\big(\slashed{D }_{M^{d-1}}\big)\big| \ (-1)^{\eta (\slashed{D }_{M^{d-1}})/2} \ .
\end{align*}
The global parity anomaly is due to the fact that the fermion path integral is in general not a real number, whereas classical orientation-reversal (or `parity') symmetry, which acts by complex conjugation on path integrals, would imply that the path integral is real.

The formula \eqref{eq:Zparityphi} for the anomalous phase $\mathcal{Z}^{(-1)}_{\rm parity}(\phi)$ now immediately suggests a way to cancel the parity anomaly: We combine bulk and boundary degrees of freedom by introducing for the bulk fields the action 
\begin{align}
\label{EQ: Action topological isulator bulk fields}
 S_{\rm bulk}\big(M^{d,1}\big) = \ii \pi \, \int_{M^{d,1}}\, K_{\rm AS} \ ,
\end{align} 
where $M^{d,1}$ is a regular morphism from $\emptyset$ to $M^{d-1}$, i.e. $\partial M^{d,1}=M^{d-1}$. Then after integrating out the boundary fermion fields, the contribution to the path integral measure for the combined system is given by
\begin{align*}
Z_{\rm comb}\big(M^{d,1}\big) & = \big|\text{det}\big(\slashed{D }_{M^{d-1}}\big)\big| \ (-1)^{ \eta (\slashed{D }_{M^{d-1}})/2} \ \e^{-S_{\rm bulk}(M^{d,1})} \\[4pt] & = 
\big|\text{det}(\slashed{D }_{M^{d-1}})\big| \, \exp \Big(\,\frac{\ii\pi}2\, \eta \big(\slashed{D }_{M^{d-1}}\big) - \ii\pi\, \int_{M^{d,1}}\, K_{\rm AS}\, \Big) \\[4pt] & = \big|\text{det}\big(\slashed{D }_{M^{d-1}}\big)\big| \ (-1)^{\text{ind}(\slashed{D }^+_{\hat M^{d,1}})} \ ,
\end{align*}
where in the last line we used the Atiyah-Patodi-Singer index formula \eqref{APS Theorem}. This expression is real. Thus the combined bulk-boundary system is invariant under orientation-reversal and gauge transformations, since now its path integral is real, due to `anomaly inflow' from the bulk to the boundary. In particular, the non-anomalous partition function of the combined system
\begin{eqnarray*}
Z_{\rm comb}\big(M^{d,1}\big) = \mathcal{Z}^{(-1)}_{\rm parity}\big(M^{d,1}\big) \ \big|\text{det}\big(\slashed{D }_{M^{d-1}}\big)\big|
\end{eqnarray*}
is defined in the full $d$-dimensional quantum field theory $\mathcal{Z}^{(-1)}_{\rm parity}$, rather than just the truncation $\mathsf{tr}\mathcal{Z}^{(-1)}_{\rm parity}$ in which the original partition function $Z^{(-1)}_{\rm parity}$ lives. Looking at this from a different perspective, we see that the existence of an effective long wavelength action \eqref{EQ: Action topological isulator bulk fields} for the bulk gauge and gravitational fields implies the existence of gapless charged boundary fermions with an anomaly cancelling the anomaly of the bulk quantum field theory under orientation-reversing transformations. 

This example provides a simple model for the general feature of some topological states of matter: Symmetry-protected topological phases in $d$ dimensions are related to global anomalies in $d-1$ dimensions. In the simplest case $d=2$, the quantum mechanical time-reversal anomaly on the $0+1$-dimensional boundary is encoded by the $1+1$-dimensional symmetry-protected topological phase in the bulk whose topological response action \eqref{EQ: Action topological isulator bulk fields} evaluates to $\ii\pi\,\Phi$, where $\Phi$ is the magnetic flux of the background gauge field through $M^{2,1}$. This sets the two-dimensional $\theta$-angle equal to $\pi$, and the action reduces to the Wilson loop of the gauge field $A_{M^1}$ on $\partial M^{2,1}=M^1$.

For the $d=4$ example of the time-reversal (or space-reflection)
invariant $3+1$-dimensional fermionic topological insulator with
$2+1$-dimensional boundary~\cite{WittenFermionicPathInt}, the integral
of the Atiyah-Singer index density $K_{\rm AS}$ in four dimensions
yields the sum of the instanton number $I$ of the background gauge
field and a gravitational contribution related to the signature
$\sigma$ of the four-manifold $M^{4,1}$~\cite{NashBook}. For the cancellation of the parity anomaly we had to introduce the term $\ii\pi\, I$ in the action, which is the anticipated statement that the $\theta$-angle parameterising the axionic response action is equal to $\pi$ inside a topological insulator. The bulk-boundary correspondence discussed above then resembles the well-known situation from three-dimensional Chern-Simons theory, to which the bulk theory reduces on $\partial M^{4,1}=M^3$~\cite{Niemi,AlvarezGaume}.

The present formalism generalises this perspective to systematically
construct quantum field theories with global parity symmetry that
characterise gapless charged fermionic boundary states of certain
symmetry-protected topological phases of matter in all higher even
dimensions $d\geq6$. Indeed, the anomaly of a quantum field theory in
$d=2n$ dimensions involving an action that integrates the
Atiyah-Singer index density $K_{\rm AS}$ reduces on the boundary $\partial M^{d,1}=M^{d-1}$ to coupled combinations of gauge and gravitational Chern-Simons type terms. The bulk action \eqref{EQ: Action topological isulator bulk fields} will now also involve couplings between gauge and gravitational degrees of freedom through intricate combinations of Chern and Pontryagin classes, such that the
bulk symmetry-protected topological phase completely captures the
parity anomaly of the boundary theory. Some examples of such mixed
gauge-gravity phases can be found e.g. in~\cite{Wang:2014pma}.
\end{example}

\section{Anomalies and projective representations\label{sec:generalanomalies}}

When a quantum field theory has a global symmetry which is
non-anomalous, the symmetry group acts on the Hilbert space of quantum
states. When the global symmetry is anomalous, the group instead acts
projectively on the state space, or equivalently a non-trivial central
extension acts linearly. Such central extensions correspond to group
2-cocycles which also specify the class of a gerbe on the classifying
space of the symmetry group. To see this effect in our framework, it
is necessary to extend the functorial quantum field theories defined
in Section~\ref{sec:catQFT} in order to capture the action of the
anomaly on quantum states. In this section we develop a general
framework of extended field theories which will encode anomalies in
this way,
following~\cite{FreedAnomalies,MonnierHamiltionianAnomalies,BoundaryConditionsTQFT}. In
this formalism the same group 2-cocycle characterising the projective
representation in the $d-1$-dimensional boundary field theory also
specifies an invariant of the bulk $d$-dimensional quantum field
theory in which the anomaly is encoded, so that such cocycles also
describe invariants of certain topological phases.

\subsection{Invariant background fields}

Physical fields should be local, i.e. they form a sheaf on
the category of manifolds, and can have (higher) internal symmetries such as gauge
symmetries. We thus incorporate all data of fields such as bundles with connections, spin structures and metrics into a stack
$
\mathscr{F}\colon \mathsf{Man}_d^{\mathrm{op}}\rightarrow \mathsf{Grpd}
$
on the category $\mathsf{Man}_d$ of $d$-dimensional manifolds with
corners and local diffeomorphisms; we regard $\mathsf{Man}_d$ as a
2-category with only trivial 2-morphisms, and $\mathsf{Grpd}$ denotes
the 2-category of small groupoids, functors and natural
isomorphisms. One should think of elements of $\mathscr{F}(M)$ as the
collection of classical background fields on $M$, which in particular
satisfies the sheaf condition, i.e. for every open cover $\{U_a\}$ of
a manifold $M$, the diagram
$$
\Fscr(M)\longrightarrow \prod_a\, \Fscr(U_a)\doublearrow
\prod_{a,b}\, \Fscr(U_a\cap U_b)
$$
is a weak/homotopy equalizer diagram in $\mathsf{Grpd}$. Using stacks we avoid problems associated to the fact that
pullbacks of certain background fields are only functorial up to
canonical isomorphism. We can include some geometrical structures such as metrics by considering the corresponding set as a groupoid with only identity morphisms. With this in mind it would be more general to work with $\infty$-stacks, but for our purpose stacks are enough. 

\begin{remark}
We implicitly pick, for every surjective submersion $\pi \colon Y\rightarrow M$, weak adjoint inverses to the canonical map $\mathscr{F}(M) \rightarrow \mathsf{Desc}_\mathscr{F}(Y)$ where $\mathsf{Desc}_\mathscr{F}(Y)$ is the category of descent data associated to $\pi$. For every refinement 
\[
\begin{tikzcd}
Y_1 \ar[dd,swap,"{f}"] \ar[dr,"{\pi_1}"]  & \\
 & M \\
Y_2 \ar[ru,swap, "{\pi_2}"] &
\end{tikzcd}
\] 
we get a natural functor $f^\ast \colon \mathsf{Desc}_\mathscr{F}(Y_2)\rightarrow \mathsf{Desc}_\mathscr{F}(Y_1)$ for which we pick a weak adjoint inverse. The adjointness condition is essential for ensuring naturality of constructions using descent properties.   
\end{remark}

To generalize the constructions of Appendix \ref{Section: Gluing} we need the following notion. 

\begin{definition}
\label{Def:Invariant element stack}
Let $\mathscr{F}$ be a stack, $M$ a manifold, and $\mathscr{S}(M)$ a groupoid which consists of a collection of open subsets of $M$ including $M$ and a collection of diffeomorphisms as morphisms. 
\begin{itemize}
\item[(a)] An \underline{invariant structure} with respect to $\mathscr{S}(M)$ for an element $\mathsf{f} \in \mathscr{F}(M)$ is a natural 2-transformation
\[
\begin{tikzcd}
 &\; \ar[dd, Leftarrow, shorten >= 20 ,shorten <= 15, "\, \mathsf{f}"] & \\
\mathscr{S}(M)^{\mathrm{op}} \ar[rr, bend left, "{\mathscr{F}}"] \ar[rr, swap, bend right, "\mathsf{1}"]& & \mathsf{Grpd} \\
 & \; &
\end{tikzcd}
\]
from the constant 2-functor sending every object to the groupoid $\mathsf{1}$ with one object and one morphism, such that the natural transformation is induced by $f$ on objects (see the following remark for an explaination). Here we regard $\mathscr{S}(M)$ as a 2-category with trivial 2-morphisms.\footnote{This is the same as a higher fixed point for the groupoid action corresponding to $\mathscr{F}$, as discussed for example in \cite{HesseSchweigertValentino}.} We call an element $\mathsf{f}\in \mathscr{F}(M)$ together with the choice of an invariant structure an \underline{invariant element}.

\item[(b)] A morphism $\mit\Theta\colon \mathsf{f}\rightarrow \mathsf{f}'$ between invariant elements $\mathsf{f},\mathsf{f}' \in \mathscr{F}(M)$ is \underline{invariant} under $\mathscr{S}(M)$ if it induces a modification between the natural 2-transformations corresponding to $\mathsf{f}$ and~$\mathsf{f}'$.
\end{itemize}    
\end{definition}

\begin{remark}
Let us spell out in detail what we mean by saying that $\mathsf{f}\in\mathscr{F}(M)$ induces a natural 2-transformation on objects. A map $\mathsf{f}_U \colon \mathsf{1} \rightarrow \mathscr{F}(U)$ is an element of $\mathscr{F}(U)$. We set $\mathsf{f}_U = \mathsf{f}|_U$ for all $U \in \mathrm{Obj}\big(\mathscr{S}(M)\big)$. To equip this with the structure of a natural 2-transformation we have to fix natural transformations (see Definition~\ref{Definition transformation Bicategory})
 \[
\begin{tikzcd}
\Hom_{\mathscr{S}(M)^{\rm op}}(U_1,U_2) \arrow{r}{\mathsf{1} }\arrow[d,swap,"\mathscr{F}(\, \cdot\, )"] & \Hom_{\mathsf{Grpd}}(\mathsf{1},\mathsf{1})\arrow{d}{\, \mathsf{f}_{U_2\ast}}\\
\Hom_{\mathsf{Grpd}}\big(\mathscr{F}(U_1),\mathscr{F}(U_2)\big)\arrow[ru,Rightarrow, "\mathsf{f}_{U_1 U_2}"] \arrow[r,swap, "\mathsf{f}_{U_1}^\ast"] & \Hom_{\mathsf{Grpd}}\big(\mathsf{1},\mathscr{F}(U_2)\big)
\end{tikzcd}
\]  
This is the same thing as morphisms $\mathsf{f}_{U_1 U_2}(t)\colon t^\ast \, \mathsf{f}_{U_1} \rightarrow \mathsf{f}_{U_2}$ for every morphism $t\colon U_2 \rightarrow U_1$ of $\mathscr{S}(M)$ which have to satisfy the coherence conditions \eqref{Equation1: Definition Transformation} and \eqref{Equation2: Definition Transformation}:
\begin{align}
\label{EQ1:Structure Invariant}
\mathsf{f}_{U_2 U_3}(t_2) \circ t_2^\ast\, \mathsf{f}_{U_1 U_2}(t_1) &= \mathsf{f}_{U_1U_3}(t_1\circ t_2) \circ \Phi_{\mathscr{F}(U_3)\mathscr{F}(U_2)\mathscr{F}(U_1)}(t_1\times t_2) \ , \\[4pt]
\label{EQ2:Structure Invariant}
\mathsf{f}_{U U}(\id_U)\circ \Phi_{\mathscr{F}(U)}(\id_\star) &= \id\,_{\mathsf{f}_U} 
\end{align}
for morphisms $t_2 \colon U_3 \rightarrow U_2$ and $t_1\colon U_2
\rightarrow U_1$. For a sheaf considered as a stack, the maps $\mathsf{f}_{U_1U_2}(t)$ must be identity maps and we reproduce, for example, the definition of an invariant function. 
\end{remark}

\begin{example}
Let $\mathscr{F}=\mathsf{Bun}_G$ be the stack of principal $G$-bundles, and let $M$ be a manifold equipped with an action $\rho\colon \Gamma \rightarrow \text{Diff}(M)$ of a group $\Gamma$ by diffeomorphisms of $M$. We can encode the action into a groupoid $\mathscr{S}(M)$ as in Definition \ref{Def:Invariant element stack} with one object $M$ and morphisms $\{\rho(\gamma)\mid \gamma \in \Gamma \}$. A $G$-bundle $P$ which is invariant under $\mathscr{S}(M)$ comes with gauge transformations $\mit\Theta_\gamma \colon \rho(\gamma)^\ast P \rightarrow P$ satisfying \eqref{EQ1:Structure Invariant} and \eqref{EQ2:Structure Invariant}. This is just a $\Gamma$-equivariant $G$-bundle.
An invariant morphism between two $\Gamma$-equivariant $G$-bundles is then a $\Gamma$-equivariant gauge transformation.    
\end{example}

\begin{definition}
Let $\Sigma $ be a $d-1$-dimensional manifold (with boundary). For every (not necessarily open) interval $I\subset \R$, we say that an element $\mathsf{f} \in \mathscr{F}(\Sigma\times I)$ is \underline{constant} along $I$ if it is invariant under translations in the direction along $I$, i.e. invariant with respect to the groupoid with open subsets of $\Sigma\times I$ as objects and translations along $I$ as morphisms.
\end{definition}

\begin{remark}
We employ a similar definition for manifolds of the form $Y\times I_1\times \cdots \times I_n$.
\end{remark}

\subsection{Geometric cobordism bicategories}\label{Appendix Geometric bicategories}

Inspired by \cite{schommer2011classification} and the sketch of~\cite[Appendix~A]{MonnierHamiltionianAnomalies}, we introduce a bicategory of manifolds equipped with geometric fields. For the definition of a Dirac operator, a metric on the underlying manifold is crucial, whence we cannot assume that the field content is topological. This leads to technical problems in defining 2-morphisms. We make the assumption that the field content is constant near gluing boundaries and use a specific choice of collars to get around these problems.   

We define a bicategory $\mathsf{Cob}^{\mathscr{F}}_{d,d-1,d-2}$ with objects given by quadruples 
\[
\big(M^{d-2},\mathsf{f}^{d-2}, \epsilon_1, \epsilon_2\big)
\]
consisting of a closed $d-2$-dimensional manifold $M^{d-2}$ with $n$ connected components $M^{d-2}_i$, $n$-tuples $\epsilon_{1},\epsilon_{2}\in \R_{>0}^n$ and an element $\mathsf{f}^{d-2}\in \mathscr{F}\big(M^{d-2}\times(-\epsilon_1,\epsilon_1)\times (-\epsilon_2,\epsilon_2)\big)$ which is constant along $(-\epsilon_1,\epsilon_1)\times (-\epsilon_2,\epsilon_2)$. Here we introduced the notation
\begin{align*}
M^{d-2}\times(-\epsilon_1,\epsilon_1)\times (-\epsilon_2,\epsilon_2)= \bigsqcup_{i=1}^n \, M_i^{d-2}\times (-\epsilon_{1,i},\epsilon_{1,i})\times(-\epsilon_{2,i},\epsilon_{2,i}) \ ,
\end{align*} 
which we will continue to use throughout this section.

There are two different kinds of 1-morphisms in $\mathsf{Cob}^{\mathscr{F}}_{d,d-1,d-2}$:
\begin{itemize}
\item[(1a)]
Regular 1-morphisms 
\[
M^{d-1,1}\colon \big(M^{d-2}_-,\mathsf{f}^{d-2}_-, \epsilon_{-1},\epsilon_{-2}\big) \longrightarrow \big(M^{d-2}_+,\mathsf{f}^{d-2}_+,\epsilon_{+1},\epsilon_{+2}\big)
\]
consist of 7-tuples
\[
\big(M^{d-1,1},\varphi^{d-1}_-,\varphi^{d-1}_+,\mathsf{f}^{d-1},{\mit\Theta}^{d-1}_-,{\mit\Theta}^{d-1}_+, \epsilon\big) \ ,
\]
where $M^{d-1,1}$ is a $d - 1$-dimensional manifold with boundary and $n$ connected components together with a decomposition of a collar of its boundary into $N^{d-1}_-$ and $N^{d-1}_+$, 
$\varphi^{d-1}_- \colon M^{d-2}_-\times [0,\epsilon_{-1})\rightarrow N^{d-1}_-$  and $\varphi^{d-1}_+ \colon M^{d-2}_+\times (-\epsilon_{+2},0]\rightarrow N^{d-1}_+$ are diffeomorphisms, $\epsilon \in \R_{>0}^n$, $\mathsf{f}^{d-1}\in \mathscr{F}\big( M^{d-1,1}\times (-\epsilon,\epsilon)\big)$ is constant along $(-\epsilon,\epsilon)$ and\footnote{For this statement to make sense we require that $\epsilon$ is compatible with $\epsilon_{\pm 2}$ on the boundary.} ${\mit\Theta}^{d-1}_\pm \colon \mathsf{f}^{d-2}_\pm\rightarrow \varphi_\pm^\ast \mathsf{f}^{d-1}$ are constant morphisms. Here we use ${\mit\Theta}^{d-1}_\pm$ to implicitly define the structure of a constant object on $N^{d-1}_\pm$. 

\item[(1b)]
Limit 1-morphisms consist of diffeomorphisms $\phi \colon M_-^{d-2}\rightarrow M_+^{d-2}$ together with a morphism ${\mit\Theta} \colon \mathsf{f}^{d-2}_- \rightarrow \phi^\ast \, \mathsf{f}^{d-2}_+$ which is constant along $(-\epsilon_1, \epsilon_1)\times (-\epsilon_2, \epsilon_2)$.\footnote{For this to make sense we require that all $n$-tuples $\epsilon$ are equal.} We refer to diffeomorphisms of this form as `$\mathscr{F}$-diffeomorphisms'.
\end{itemize}

For the composition of regular 1-morphisms we glue the underlying manifolds using their collars, and define the composed field content using covers $U_1\times (-\epsilon,\epsilon)$, $U_2\times (-\epsilon,\epsilon)$ and $U_3\times (-\epsilon,\epsilon)$ constructed from the cover in \eqref{Definition Cover Gluing}, and the descent property of the stack $\mathscr{F}$. Note that we can use the interpretation in terms of mapping cylinders also in this general situation.
Composition of limit 1-morphisms is given by composition of $\Fscr$-diffeomorphisms. The composition of a limit 1-morphism with a regular 1-morphism is given by changing the identification of the collars and ${\mit\Theta}_\pm^{d-1} $ using the limit 1-morphism.

There are also two different kinds of 2-morphisms in $\CobF$:
\begin{itemize}
\item[(2a)]
Regular 2-morphisms (see Figure~\ref{Fig: 2-Morphism}) 
\begin{eqnarray*}
&& M^{d,2} \colon \big(M^{d-1,1}_1,\varphi^{d-1}_{1-},\varphi^{d-1}_{1+},\mathsf{f}^{d-1}_1,{\mit\Theta}^{d-1}_{1-},{\mit\Theta}^{d-1}_{1+}, \epsilon_1\big) \\ && \qquad \qquad \qquad \qquad \qquad \qquad \qquad \qquad \Longrightarrow \big(M^{d-1,1}_2,\varphi^{d-1}_{2-},\varphi^{d-1}_{2+},\mathsf{f}^{d-1}_2,{\mit\Theta}^{d-1}_{2-},{\mit\Theta}^{d-1}_{2+}, \epsilon_2\big) \ ,
\end{eqnarray*}
with $M^{d-1,1}_{i} \colon\big(M^{d-2}_-,\mathsf{f}^{d-2}_-, \epsilon_{-1},\epsilon_{-2}\big) \rightarrow \big(M^{d-2}_+,\mathsf{f}^{d-2}_+, \epsilon_{+1},\epsilon_{+2}\big)$ regular 1-morphisms for $i=1,2$, consist of equivalence classes of 6-tuples 
\[
\big(M^{d,2},\mathsf{f}^d,\varphi^d_-,\varphi^d_+,{\mit\Theta}^d_-,{\mit\Theta}^d_+\big) \ ,
\]
where $M^{d,2}$ is a $d$-dimensional $\langle 2 \rangle $-manifold (see Appendix \ref{Appendix Manifolds with corners..}) whose corners are equipped with a decomposition of a collar of the $0$-boundary into $N^{d}_-$ and $N^{d}_+$ such that the closure of $N^{d}_\pm$ contains the $1$-boundary, $\mathsf{f}^{d}$ is an element of $\mathscr{F}(M^{d,2})$, $\varphi^d_- \colon M^{d-1,1}_1 \times [0,\epsilon_{-1}) \rightarrow N^d_-$ and $\varphi^d_+ \colon M^{d-1,1}_2 \times (-\epsilon_{+2},0] \rightarrow N^d_+$ are diffeomorphisms, and ${\mit\Theta}^d_- \colon \mathsf{f}^{d-1}_1 \rightarrow \varphi^{d\, \ast}_- \, \mathsf{f}^d$ and ${\mit\Theta}^d_+ \colon \mathsf{f}^{d-1}_2 \rightarrow \varphi^{d\, \ast}_+ \, \mathsf{f}^d$ are constant morphisms. All of these structures have to be compatible, in the sense that the diagram~\cite{schommer2011classification}   
\begin{equation}
\begin{footnotesize}
\label{EQ: Compatibility condition 2morphisms}
\begin{tikzcd}
& & M^{d-1,1}_2 \times (-\epsilon_2,0] \ar[dd,"\varphi_+^d"]& & \\ 
& &  & & \\
M^{d-2}_-\times [0,\epsilon_{-1}) \times (-\epsilon_{-2},0] \ar[rr, "{\imath_-}"]\ar[rruu, "{\varphi_{2-}^{d-1}\times \id}"] \ar[rrdd,"{\varphi_{1-}^{d-1}\times -\id}",swap] & & M^{d,2} & &M^{d-2}_+\times  [0,\epsilon_{+1}) \times (-\epsilon_{+2},0] \ar[ll,swap, "{\imath_+}"] \ar[lldd,"{\varphi_{1+}^{d-1}\times -\id}"] \ar[lluu,"{\varphi_{2+}^{d-1}\times \id}",swap] \\ 
& & & & \\
& & M^{d-1,1}_1 \times [0, \epsilon_1) \ar[uu,swap,"\varphi_-^d"] & &
\end{tikzcd}
\end{footnotesize}
\end{equation}
commutes, where $\imath_\pm$ are inclusions. We change the sign of the coordinates corresponding to both intervals in the lower embedding. This induces a diagram of functors in groupoids and we require that all morphisms $\mit\Theta$ are compatible with this diagram.  Note that the collars of the 1-morphisms induce collars for the 1-boundaries which agree by \eqref{EQ: Compatibility condition 2morphisms}. Two such 6-tuples are equivalent if they are $\mathscr{F}$-diffeomorphic relative to half of the collars.

\begin{figure}
\footnotesize
\begin{center}
\begin{overpic}[width=12cm,
scale=1]{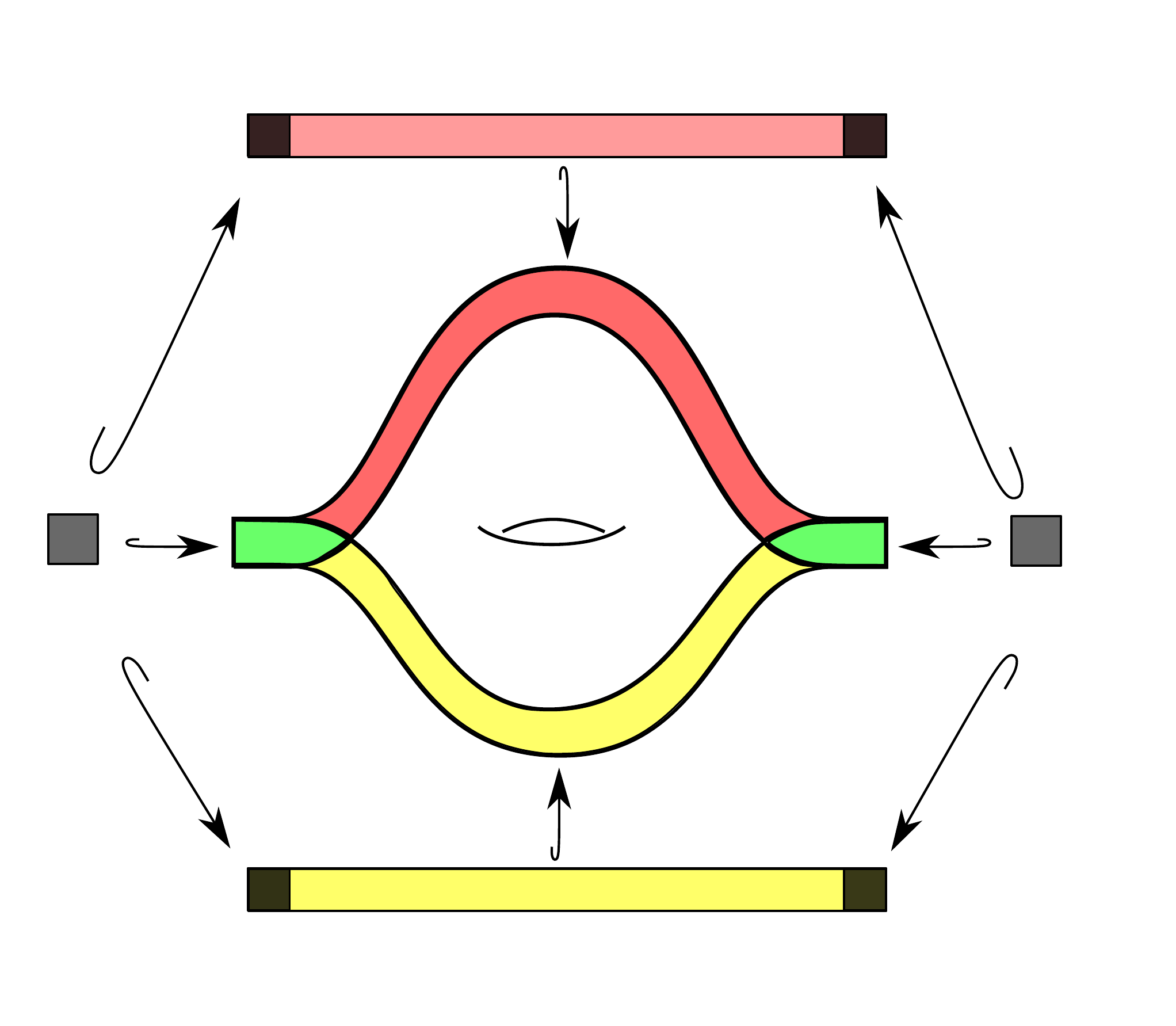}
\put(40,5){$M^{d-1,1}_1 \times [0,\epsilon_1)$}
\put(40,80){$M^{d-1,1}_2 \times (-\epsilon_2,0]$}
\put(-10,35){$M^{d-2}_-\!\times \![0,\epsilon_{-1})\!\times \!(-\epsilon_{-2},0]$}
\put(75,35){$M^{d-2}_+\!\times \![0,\epsilon_{+1})\!\times \!(-\epsilon_{+2},0]$}
\put(46,50){$M^{d,2}$}
\put(49,16){$\varphi^d_-$}
\put(50,69){$\varphi^d_+$}
\put(17,24){$\varphi^{d-1}_{1-}\! \times \! -\id$}
\put(17,56){$\varphi^{d-1}_{2-}\! \times \! \id$}
\put(64,24){$\varphi^{d-1}_{1+}\! \times \! -\id$}
\put(67,56){$\varphi^{d-1}_{2+}\! \times \! \id$}
\end{overpic}
\end{center}
\caption{\small Illustration of a regular 2-morphism in $\CobF$.}
\normalsize
\label{Fig: 2-Morphism}
\end{figure}

\item[(2b)]
Limit 2-morphisms consist of pairs $(\phi, {\mit\Theta})$, where $\phi: M^{d-1,1}_1 \rightarrow M^{d-1,1}_2$ is a diffeomorphism relative to collars together with a morphism ${\mit\Theta} \colon \mathsf{f}^{d-1}_1 \rightarrow \phi^\ast\, \mathsf{f}^{d-1}_2$. There are no non-trivial 2-morphisms between limit 1-morphisms. 
\end{itemize}

We define horizontal and vertical composition of 2-morphisms as follows:
\begin{itemize}
\item[(Ha)]
Horizontal composition of regular 2-morphisms is given by gluing along 1-boundaries.

\item[(Hb)]
Horizontal composition of limit 2-morphisms is defined by ``gluing together'' diffeomorphisms and the descent condition for morphisms in the stack $\mathscr{F}$. This uses the open cover defined in \eqref{Definition Cover Gluing}.

\item[(Hc)]
Horizontal composition of a limit 2-morphism with a regular 2-morphism is defined by the attachment of a mapping cylinder to the 1-boundary. 

\item[(Va)]
Vertical composition of limit 2-morphisms is given by composition of diffeomorphisms, pullback and composition of morphisms in the stack $\mathscr{F}$. 

\item[(Vb)]
Vertical composition of regular 2-morphisms is a little bit more complicated. Simple gluing of $M^{d,2}_1$ and $M^{d,2}_2$ along a common 1-morphism does not give a 2-morphism again, since the resulting 1-boundaries are ``too long". In the context of topological field theories a solution to this problem \cite{schommer2011classification} consists in picking once and for all a diffeomorphism $[0,2] \rightarrow [0,1]$. We are unable to use this trick here, since the stack we consider in this paper contains a metric. Instead, we will use collars to circumvent this problem. 
Given two regular 2-morphisms $M^{d,2}_1 \colon M^{d-1,1}_1 \Rightarrow M^{d-1,1}_2$ and $M^{d,2}_2 \colon M^{d-1,1}_2 \Rightarrow M^{d-1,1}_3$, we define 
\begin{align*}
&\tilde{M}^{d,2}_1 = M^{d,2}_1 \setminus \varphi^d_{1+}\big(M^{d-1,1}_2\times(-\tfrac{\epsilon}{2},0] \big) \ ,\\[4pt]
&\tilde{M}^{d,2}_2 = M^{d,2}_2 \setminus \varphi^d_{2-}\big(M^{d-1,1}_2\times [0,\tfrac{\epsilon}{2}) \big) \ .
\end{align*} 
We define the vertical composition $M^{d,2}_{2}\circ M^{d,2}_{1}$ to be the manifold resulting from gluing $\tilde{M}^{d,2}_1$ and $\tilde{M}^{d,2}_2$ along $M^{d-1,1}_2$. We have to equip this manifold with appropriate collars: 
Write $\tilde{N}^d_1= N^d_1 \cap \varphi^d_{1+}\big(M^{d-1,1}_2\times(-\tfrac{\epsilon}{2},0]\big)$, where $N^d_1$ is the incoming collar of $M^{d,2}_{1}$. We set  
\begin{align*}
C = \big( \varphi^d_{2-} \circ (\id \times \, \cdot \, +\epsilon) \circ (\varphi^d_{1+})^{-1}\big)\big(\tilde{N}^d_1\big) \ .
\end{align*}
We can glue $C$ to the remainder of the collar of $M^{d-1,1}_1$ to get
a new collar; this is only possible because we assumed that the
corresponding elements of $\mathscr{F}$ are constant along the
collars. It is possible to define a new collar for $M^{d-1,1}_3$ in
the same way.

\item[(Vc)]
Vertical composition of a limit 2-morphism with a regular 2-morphism
is defined by changing the identification of collars as in Section~\ref{sec:catQFT}.
\end{itemize}

This completes the definition of the geometric cobordism bicategory $\CobF$.
The disjoint union of manifolds makes $\CobF$ into a symmetric monoidal bicategory.

\subsection{Anomalies and extended quantum field theories\label{sec:anomaliesextended}}

We will now give a general description of anomalies in the framework of functorial quantum field theory. The point of view we take in this paper is that anomalies in $d-1$ dimensions can be described by invertible extended field theories in $d$ dimensions \cite{FreedAnomalies, BoundaryConditionsTQFT, MonnierHamiltionianAnomalies}. This is naturally formulated in the language of symmetric monoidal bicategories (or $(\infty,d)$-categories, see \cite{BoundaryConditionsTQFT}). The most important concepts for the following treatment are summarized in Appendix \ref{Appendix bicategories}; a detailed introduction can be found in~\cite[Chapter~2]{schommer2011classification}.

\begin{definition}
A \underline{$d$-dimensional extended functorial quantum field theory} with background fields $\mathscr{F}$ (or \underline{extended quantum field theory} for short) is a symmetric monoidal 2-functor 
$$
\mathcal{A} \colon \mathsf{Cob}^{\mathscr{F}}_{d,d-1,d-2}\longrightarrow \Tvs
$$
from a geometric cobordism bicategory to the 2-category of 2-vector spaces.
\end{definition} 
\begin{remark}
The definition of the symmetric monoidal bicategory $\Tvs$ is given in Example~\ref{Def:2Vect}. Although the 2-category of 2-Hilbert spaces would be more natural for some physical applications, we choose here to work with $\Tvs$ since it reduces some of the technical complexity while still capturing all essential features. 

Furthermore, the KV 2-vector spaces used in this paper are categorifications of finite dimensional vector spaces. They are enough for our purpose. However, more complicated extended quantum field theories require more elaborate target 2-categories corresponding to infinite dimensional 2-vector spaces.
\end{remark}
The Deligne product $\boxtimes $ induces a tensor product for suitable extended field theories. 

\begin{definition}
An extended quantum field theory $\mathcal{A}$ is
\underline{invertible} if there exists an extended quantum field
theory $\mathcal{A}^{-1}$ and a natural symmetric monoidal 2-isomorphism from $\mathcal{A} \boxtimes \mathcal{A}^{-1}$ to the trivial theory $\mbf1\colon \mathsf{Cob}^{\mathscr{F}}_{d,d-1,d-2}\rightarrow \Tvs $ sending every object to the monoidal unit $\fvs$ of $\Tvs$, every 1-morphism to the identity functor on $\fvs$, and every 2-morphism to the identity natural transformation.
\end{definition}

Defining the truncation of $\mathsf{Cob}^{\mathscr{F}}_{d,d-1,d-2}$ to be the sub-bicategory $\mathsf{trCob}^{\mathscr{F}}_{d,d-1,d-2}$ containing only invertible 2-morphisms, and $\mathsf{tr}\mathcal{A}$ the restriction of the 2-functor $\Aa$ to $\mathsf{trCob}^{\mathscr{F}}_{d,d-1,d-2}$, we can now give a precise definition of a quantum field theory with anomaly.

\begin{definition}
\label{Definition Anommalous field theory}
An \underline{anomalous quantum field theory} with anomaly described
by an invertible extended quantum field theory $\mathcal{A}\colon
\mathsf{Cob}^{\mathscr{F}}_{d,d-1,d-2}\rightarrow \Tvs$ is a natural symmetric monoidal 2-transformation 
\begin{align*}
A \colon \mbf1 \Longrightarrow \mathsf{tr}\mathcal{A} \ .
\end{align*} 
We call $\mathcal{A}$ the \underline{anomaly quantum field theory} describing the anomaly of $A$.
\end{definition} 
This definition is a special case of the relative quantum field theories of~\cite{RelativeQFT}. These anomaly field theories are often topological field theories. For some applications, such as to two-dimensional rational conformal field theories or to six-dimensional $(2,0)$ superconformal field theories, it is necessary to consider also non-invertible quantum field theories to capture the feature that the partition function is valued in a vector space of dimension $>1$~\cite{MonnierHamiltionianAnomalies}. 
 
We can recover the anomalous partition function of Definition \ref{Definition path Integral} by restricting $\mathcal{A}$ to a functor 
$$
\mathcal{Z}:= \mathcal{A}\big|_\emptyset \colon \mathsf{Cob}^{\mathscr{F}}_{d,d-1}\cong \End_{\mathsf{Cob}^{\mathscr{F}}_{d,d-1,d-2}} (\emptyset ) \longrightarrow \End_\Tvs (\fvs) \cong \fvs
$$ 
and $A$ to a natural transformation $Z:=A\big|_\emptyset \colon \mbf1 \Rightarrow \mathsf{tr} \mathcal{A}\big|_\emptyset $.

Unpacking Definition \ref{Def:2sym transformation} we get for every
closed $d-2$-dimensional manifold $M^{d-2}$ a $\C$-linear functor
$A(M^{d-2}) \colon \fvs=\mbf1(M^{d-2}) \rightarrow
\mathcal{A}(M^{d-2})$, which can be (non-canonically) identified with
a complex vector space in $\fvs$, and for all pairs $(M^{d-2}_-,M^{d-2}_+)$ a natural transformation 
\[
\begin{tikzcd}
\Hom_{\mathsf{Cob}^{\mathscr{F}}_{d,d-1,d-2}} \big(M^{d-2}_- , M^{d-2}_+\big) \ar[dd,swap,"\mathcal{A}"] \ar{rr}{\mbf1}& &\Hom_\Tvs (\fvs,\fvs)\ar{dd}{A(M^{d-2}_+)_\ast} \\ 
 & & \\
 \Hom_\Tvs \big(\mathcal{A}( M^{d-2}_-) ,\mathcal{A} (M^{d-2}_+)\big) \ar[uurr, Rightarrow, "A",shorten <= 2em, shorten >= 2em] \ar[rr,swap,"A(M^{d-2}_-)^\ast"] & & \Hom_\Tvs \big(\fvs, \mathcal{A}(M^{d-2}_+)\big)  
\end{tikzcd}
\]
This consists of a natural transformation
\[ 
A\big(M^{d-1,1}\big): \mathcal{A}\big(M^{d-1,1}\big)\circ A\big(M^{d-2}_-\big) \Longrightarrow A\big(M^{d-2}_+\big) 
\] 
for every 1-morphism $M^{d-1,1} \colon M^{d-2}_- \rightarrow
M^{d-2}_+$. The definition further includes a modification $\mit\Pi_A$
consisting of natural isomorphisms 
\[
{\mit\Pi}_A\big(M^{d-2}_-,M^{d-2}_+\big)\colon \chi_\mathcal{A}\circ A\big(M^{d-2}_-\big)\boxtimes A\big(M^{d-2}_+\big)\Longrightarrow A\big(M^{d-2}_- \sqcup M^{d-2}_+\big) \circ \lambda_{\mathsf{Cob}^{\mathscr{F}}_{d,d-1,d-2}}
\]
and a natural isomorphism
\[
M_A^{-1}\colon A(\emptyset)\Longrightarrow \iota_\mathcal{A} \ .
\]  
All of these structures have to satisfy appropriate compatibility conditions, which we summarize in
\begin{proposition}
\label{Lemma Compatibility conditions}
For every anomalous quantum field theory $A$ with anomaly $\mathcal{A}$, there are identities
\begin{align}
\label{EQ1: Lemma Compatibility conditions}
A\big(M^{d-1,1}_2\big) \circ A\big(M^{d-1,1}_1\big) &= A\big(M^{d-1,1}_2 \circ M^{d-1,1}_1\big) \circ \Phi_{\mathcal{A}}\big(\Aa(M^{d-1,1}_2) \circ \Aa(M^{d-1,1}_1)\big) \ , \\[4pt]
\label{EQ2: Lemma Compatibility conditions}
A\big(\id_{M^{d-2}}\big)\circ\Phi_\Aa\big(\Aa(\id_{M^{d-2}}) \big)&= \id_{A(M^{d-2})} \ , \\[4pt]
\label{EQ3: Lemma Compatibility conditions}
A\big(M^{d-1,1}_1\big) &= A\big(M^{d-1,1}_2\big)\circ \big(\mathcal{A}(f)\bullet \id_{A(M^{d-2}_-)}\big) \ , 
\end{align}
for some 2-isomorphism $f\colon M^{d-1,1}_1 \Rightarrow M^{d-1,1}_2$, together with the following commutative diagrams wherein we suppress obvious structure 2-morphisms and identity 2-morphisms:
\begin{equation}
\label{EQ3': Lemma Compatibility conditions}
\begin{footnotesize}
\begin{tikzcd}
\mathcal{A}(M^{d-1,1}_1 \sqcup M^{d-1,1}_2)\bullet \chi_\mathcal{A} \bullet A(M^{d-2}_{1-})\boxtimes A(M^{d-2}_{2-}) \ar[rrr, Rightarrow,"{{\mit\Pi}_A(M^{d-2}_{1-}, M^{d-2}_{2-})}"]
 \ar[dd,swap,Rightarrow, "A{(M^{d-1,1}_1 \sqcup M^{d-1,1}_2)}"] 
& &  &\mathcal{A}(M^{d-1,1}_1 \sqcup M^{d-1,1}_2)\bullet A(M^{d-2}_{1-}\sqcup M^{d-2}_{2-})
\ar[dd,Rightarrow, "{A(M^{d-1,1}_1) \boxtimes A(M^{d-1,1}_2)}"] \\ 
 & & &
\\
\chi_\mathcal{A} \bullet A(M^{d-2}_{1+})\boxtimes A(M^{d-2}_{2+})
\ar[rrr,swap,Rightarrow, "{{\mit\Pi}_A(M^{d-2}_{1+},M^{d-2}_{2+})}"] 
& & & A(M^{d-2}_{1+}\sqcup M^{d-2}_{2+})
\end{tikzcd}
\end{footnotesize}
\end{equation}

\begin{equation}
\label{EQ4: Lemma Compatibility conditions}
\begin{scriptsize}
\begin{tikzcd}
 \mathcal{A}(\alpha_{\CobF})\!\bullet\! A((M^{d-2}_1\sqcup M^{d-2}_2)\sqcup M^{d-2}_3)
 \ar[r, Rightarrow,"{A(\alpha_{\CobF})}"]& 
 A((M^{d-2}_1\sqcup M^{d-2}_2)\sqcup M^{d-2}_3)
  \\
   & \\
\mathcal{A}(\alpha_{\CobF})\!\bullet\! \chi_\mathcal{A} \!\bullet\! ( A( M^{d-2}_1\sqcup M^{d-2}_2)\boxtimes A( M^{d-2}_3) )
\ar[uu,Rightarrow,"{{\mit\Pi}_A(M^{d-2}_1\sqcup M^{d-2}_2, M^{d-2}_3 )} "] &  \chi_\mathcal{A} \!\bullet\! (A(M^{d-2}_1)\boxtimes A(M^{d-2}_2 \sqcup M^{d-2}_3 ) )\ar[uu,swap,Rightarrow,"{{\mit\Pi}_A( M^{d-2}_1,M^{d-2}_2 \sqcup M^{d-2}_3) }"] \\
   & \\
\mathcal{A}(\alpha_{\CobF})\!\bullet\! \chi_\mathcal{A} \!\bullet\!  ((\chi_\mathcal{A}\!\bullet\! A(M^{d-2}_1)\boxtimes A(M^{d-2}_2))\boxtimes A(M^{d-2}_3)) \ar[uu,Rightarrow,"{{\mit\Pi}_A(M^{d-2}_1,M^{d-2}_2)} "] \ar[r,swap,Rightarrow, "{{\mit\Omega}_\mathcal{A} }"] & \chi_\mathcal{A} \!\bullet\!  ( A(M^{d-2}_1)\boxtimes (\chi_\mathcal{A}\!\bullet\! A( M^{d-2}_2)\boxtimes A(M^{d-2}_3) )) \ar[uu,swap,Rightarrow, "{{\mit\Pi}_A( M^{d-2}_2,M^{d-2}_3) }"]
\end{tikzcd}
\end{scriptsize}
\end{equation}

\begin{equation}
\label{EQ5: Lemma Compatibility conditions}
\begin{tikzcd}
 \mathcal{A}\big(\lambda_\CobF\big)\bullet A(\emptyset \sqcup M^{d-2}) \arrow[rr,Rightarrow,"A(\lambda_\CobF)"] & & A\big(M^{d-2}\big)\ar[d, Rightarrow,"{{\mit\Gamma}_\mathcal{A}^{-1} }"] \\
\mathcal{A}\big(\lambda_\CobF\big)\bullet \chi_\mathcal{A} \bullet A(\emptyset )\boxtimes A\big(M^{d-2}\big) \ar[u,Rightarrow, "{{\mit\Pi}_A(\emptyset,M^{d-2}) }"]\arrow[rr,swap, Rightarrow,"M^{-1}_A\boxtimes \id"] & &  \iota_\mathcal{A} \boxtimes A\big(M^{d-2}\big)
\end{tikzcd}
\end{equation}

\begin{equation}
\label{EQ6: Lemma Compatibility conditions}
\begin{tikzcd}
A\big(M^{d-2} \sqcup \emptyset\big) \arrow[rr,Rightarrow,"A(\rho_\CobF)"] & & \mathcal{A}\big(\rho_\CobF\big)\bullet A\big(M^{d-2}\big)\ar[d,Rightarrow,"{{\mit\Delta}_\mathcal{A} }"] \\
\chi_\mathcal{A} \bullet A\big(M^{d-2}\big) \boxtimes A(\emptyset )\ar[u,Rightarrow,"{{\mit\Pi}_A(M^{d-2},\emptyset)}"]\arrow[rr,swap, Rightarrow,"\id \boxtimes M^{-1}_A"] & & A\big(M^{d-2}\big)\boxtimes \iota_\mathcal{A}  
\end{tikzcd}
\end{equation}
and
\begin{equation}
\label{EQ7: Lemma Compatibility conditions}
\begin{small}
\begin{tikzcd}
\mathcal{A}( \beta_{\CobF})\bullet\chi_\mathcal{A}\bullet A(M^{d-2}_1)\boxtimes A(M^{d-2}_2) \ar[d,swap,Rightarrow,"{{\mit\Upsilon}^{-1}_\mathcal{A} }"] \ar[rrr,Rightarrow,"{{\mit\Pi}_A(M^{d-2}_1, M^{d-2}_2)}"] & & & \mathcal{A}( \beta_{\CobF})\bullet A(M^{d-2}_1\sqcup M^{d-2}_2)\ar[d,Rightarrow, "{A(\beta_{\CobF})}"]\\
\chi_\mathcal{A}\bullet  A(M^{d-2}_2)\boxtimes A(M^{d-2}_1) \ar[rrr,swap,Rightarrow, "{{\mit\Pi}_A(M^{d-2}_2,M^{d-2}_1)}"] & & & A(M^{d-2}_2\sqcup M^{d-2}_1)
\end{tikzcd}
\end{small}
\end{equation}
\end{proposition}

\begin{proof}
Writing out the coherence diagrams \eqref{Equation1: Definition
  Transformation} and \eqref{Equation2: Definition Transformation} for
$A$ implies \eqref{EQ1: Lemma Compatibility conditions} and
\eqref{EQ2: Lemma Compatibility conditions}. The identity \eqref{EQ3:
  Lemma Compatibility conditions} is the naturality condition for the
natural symmetric monoidal 2-transformation $A$. The diagram \eqref{EQ3': Lemma Compatibility conditions} follows from the diagram \eqref{EQ: Modification} for the modification ${\mit\Pi}_A$.
The diagrams \eqref{EQ4: Lemma Compatibility conditions}--\eqref{EQ7: Lemma Compatibility conditions} follow from writing out the coherence conditions \eqref{EQ:1 Definition s.m. transformation}--\eqref{EQ:4 Definition s.m. transformation} for $A$.
\end{proof}

\begin{remark}
These conditions should be understood as a projective (or twisted) version of the
definition of a symmetric monoidal functor. For this reason we have
drawn the diagrams \eqref{EQ4: Lemma Compatibility
  conditions}--\eqref{EQ7: Lemma Compatibility conditions} in close
analogy to the diagrams in the definition of a braided monoidal
functor. 
\end{remark}

An anomalous quantum field theory with trivial anomaly $A\colon \mbf1
\Rightarrow \mbf1$ is a $d-1$-dimensional quantum field theory in the sense of Definition~\ref{Definition: Field theory}: We can canonically identify the functor $A(M^{d-2})\colon \fvs\rightarrow \fvs$ with the vector space $A(M^{d-2})(\C )$ and the natural transformation $A(M^{d-1,1})\colon \text{id}_{\fvs}\circ A(M^{d-2}_-)\Rightarrow A(M^{d-2}_+)$ with a linear map $A(M^{d-1,1}) \colon  A(M^{d-2}_-)(\C)\rightarrow A(M^{d-2}_+)(\C)$. The compatibility conditions summarised by Proposition~\ref{Lemma Compatibility conditions} then imply that the vector spaces and linear maps defined in this way form a quantum field theory.

\subsection{Projective anomaly actions\label{sec:projreps}}

Following \cite{MonnierHamiltionianAnomalies,BoundaryConditionsTQFT}
we describe how the extended quantum field theory encodes the
projective action on the state space of an anomalous field theory $A$ with anomaly
$\Aa$. We fix an object $M^{d-2}\in \mathrm{Obj}\big(\CobF \big)$. The limit 1-automorphisms of $M^{d-2}$ form the group of physical symmetries $\Sym(M^{d-2})$. Every $\phi \in \Sym(M^{d-2})$ gives rise to a $\C$-linear functor $\mathcal{A}(\phi) \colon \mathcal{A}(M^{d-2})\rightarrow \mathcal{A}(M^{d-2}) $. Choosing a non-canonical equivalence $\chi \colon \mathcal{A}(M^{d-2})\rightarrow \fvs$ identifies $\mathcal{A}(\phi)$ as a functor which takes the tensor product with a one-dimensional vector space $L_{\chi,\phi}$. The structure of $\mathcal{A}$ defines an isomorphism 
\begin{align*}
\alpha_{\chi, \phi_1 , \phi_2} \colon L_{\chi,\phi_1 } \otimes L_{\chi,\phi_2}  \longrightarrow L_{\chi,\phi_2 \circ \phi_1} \ .
\end{align*}
If we furthermore pick an isomorphism $\varphi \colon L_{\chi, \phi} \rightarrow \C$ for every $\phi \in \Sym(M^{d-2})$ we get a family of linear isomorphisms
\begin{align}\label{eq:2cocycledef}
\alpha_{\chi,\varphi, \phi_1 , \phi_2} \colon \C \longrightarrow \C \ .
\end{align}  
We will show later on using abstract arguments that this is a 2-cocycle for the group $\Sym(M^{d-2})$ whose group cohomology class is independent of the chosen equivalence $\chi$ and isomorphisms $\varphi$; a concrete proof can be found in \cite{MonnierHamiltionianAnomalies}. 

This cocycle describes the projective action of $\Sym(M^{d-2})$ on the space of quantum states of the theory $A(M^{d-2})$ as follows:
Let $A: \mbf1 \Rightarrow \mathsf{tr}\mathcal{A}$ be an anomalous quantum field theory with anomaly $\mathcal{A}$.
We use the equivalence $\chi$ chosen above to identify $A(M^{d-2})\colon \fvs \rightarrow \mathcal{A}(M^{d-2})$ with a vector space $A_\chi(M^{d-2})$.  
From $A$ we get a natural transformation $A(\phi): \mathcal{A}(\phi)\circ A(M^{d-2})\Rightarrow A(M^{d-2})$, which by horizontal composition with the identity natural transformation of $\chi$ induces a linear map
\begin{align*}
A_\chi(\phi) \colon A_\chi\big(M^{d-2}\big)\otimes L_{\chi, \phi} \longrightarrow A_\chi \big(M^{d-2}\big) \ .
\end{align*} 
By precomposing with the isomorphisms $\varphi^{-1} $ we get a projective representation 
\[
\rho_\varphi \colon \Sym\big(M^{d-2}\big) \longrightarrow \mathrm{End}_\C \big(A_\chi(M^{d-2})\big) \ .
\]
We cannot say anything about the structure of this projective representation in general, but we can describe the failure of the composition law explicitly in terms of $\mathcal{A}$:
\begin{align*}
\rho_\varphi (\phi_2) \circ \rho_\varphi (\phi_1)= \alpha_{\chi ,\varphi, \phi_1, \phi_2} \ \rho_\varphi (\phi_2\circ \phi_1) 
\end{align*} 

We say that the quantum field theory $A$ is anomaly-free on $M^{d-2}$ if there is a choice of $\chi$ and $\varphi$ such that the corresponding projective representation $\rho_\varphi$ reduces to an honest representation. This is only possible if the corresponding cohomology class of $\alpha_{\chi ,\varphi, \phi_1, \phi_2}$ is trivial. 

More generally, we can build a projective representation of the groupoid $\Sym \CobF$ of symmetries having the same objects as $\CobF$ and all limit 1-morphisms as morphisms. For this, we first need to recall the notion of a groupoid 2-cocycle.
\begin{definition}
A \underline{2-cocycle} of a groupoid $\mathscr{G}$ with values in $\fvs$ is a 2-functor $\alpha \colon \mathscr{G}\rightarrow \mathsf{BLine}_\C$, where we consider $\mathscr{G}$ as a 2-category with trivial 2-morphisms and $\mathsf{BLine}_\C$ is the 2-category with one object, and the symmetric monoidal category $\mathsf{Line}_\C$ of complex lines and linear isomorphisms as endomorphisms.  
\end{definition}
 
\begin{remark}\label{rem:2cocycle}
Let us spell out explicitly some details of this definition:
\begin{itemize}
\item[(a)]
We can pick an equivalence between $\mathsf{Line}_\C$ and
$\mathsf{B}\C^\times = \C \,/\!\! /\, \C^\times $ by choosing for every complex line $L$ an isomorphism $\chi \colon L \rightarrow \C$; the inverse of this equivalence is the embedding of  $\C \,/\!\! /\, \C ^\times $ into $\mathsf{Line}_\C$. This induces a 2-functor $\alpha:\Gscr \rightarrow \mathsf{B}^2\C^\times$. Writing out Definition \ref{Definition Morphism Bicategory} we get for every pair $(g,g')\in \mathrm{Hom}_\mathscr{G}(G_1,G_2)\times \mathrm{Hom}_\Gscr(G_2,G_3)$ a non-zero complex number $\alpha_{g,g'}$ such that 
\begin{align*}
\alpha_{g_3\circ g_2, g_1} \, \alpha_{g_3,g_2} = \alpha_{g_3, g_2\circ g_1} \,  \alpha_{g_2,g_1} \ , 
\end{align*}
for all composable morphisms $g_1,g_2,g_3$, and 
\begin{align*}
\alpha_{\id_{\sft(g)},g}=\alpha_{\id_{\sft(g)},\id_{\sft(g)}}=\alpha_{g,\id_{\sfs(g)}} \ . 
\end{align*}
Note that the 2-morphism $\alpha_1 \colon \alpha(\id)\Rightarrow \id$ is completely fixed by the coherence condition \eqref{EQ2: Definition 2Functor} and takes the value $\alpha_{\id,\id}^{-1}$.

\item[(b)]
The data contained in a natural 2-transformation $\sigma :\alpha\Rightarrow\alpha'$ between two 2-cocycles is given by a collection $\sigma_g \in \C^\times$ for all morphisms $g$ in $\mathscr{G}$ such that
\begin{align*}
\sigma_{g_2\circ g_1} \ \alpha'_{g_2,g_1}= \alpha_{g_2,g_1} \ \sigma_{g_1} \, \sigma_{g_2} 
\end{align*}
for all composable morphisms $g_1,g_2$. This is the coherence condition \eqref{Equation1: Definition Transformation} which also implies \eqref{Equation2: Definition Transformation}.
We see that natural 2-transformations restrict to the usual coboundaries on endomorphisms of an object.
This immediately implies that the 2-cocycles \eqref{eq:2cocycledef} are well-defined up to coboundaries. To see this we pick two different 2-equivalences $\chi_1,\chi_2 \colon \mathsf{BLine}_\C \rightarrow \mathsf{B}^2\C^\times $ which both have the embedding $i\colon \sfB^2\C^\times\rightarrow \mathsf{BLine}_\C$ as inverse. We then get a chain of natural 2-transformations
\begin{align*}
\chi_1 \Longrightarrow \chi_2 \circ i\circ \chi_1 \Longrightarrow \chi_2
\end{align*}   
which implies that the 2-cocycles are independent of the choice of $\chi$ up to coboundary. 

\item[(c)]
The data contained in a modification $\theta:\sigma\Rrightarrow\sigma'$ between two natural 2-transformations is an assignment of an element $\theta_G \in \C^\times $ to every $G\in \mathrm{Obj}(\mathscr{G})$ such that
\begin{align*}
\theta_{\sft(g)} \ \sigma_g = \sigma'_g \ \theta_{\sfs(g)} \ ,
\end{align*}
which is the condition \eqref{EQ: Modification}.
\end{itemize}
\end{remark} 

Having defined 2-cocycles for groupoids we can now define projective representations (see e.g.~\cite[Section 2.3.1]{TwistedDWandGerbs}).

\begin{definition}
A \underline{projective representation} $\rho $ of a groupoid $\mathscr{G}$ twisted by a 2-cocycle $\alpha\colon \mathscr{G} \rightarrow \sfB^2\C^\times$ consists of the following data:
\begin{itemize}
\item[(a)]
A complex vector space $V_G$ for all $G\in \mathrm{Obj}(\mathscr{G})$.

\item[(b)]
A linear map $\rho (g)\colon V_{\sfs(g)}\rightarrow V_{\sft(g)}$ for each morphism $g$ of $\Gscr$ such that 
\begin{align*}
\rho(g_2) \circ \rho(g_1) = \alpha_{g_2,g_1} \ \rho (g_2 \circ g_1)
\end{align*}
for all composable morphisms $g_1,g_2$.
\end{itemize}
\end{definition}

In this definition we work with cocycles valued in $\sfB^2\C^\times$. A similar but slightly more complicated definition using cocycles with target $\mathsf{BLine}_\C$ can be deduced from

\begin{proposition}\label{prop:projrep}
A projective groupoid representation with 2-cocycle $\alpha\colon \mathscr{G}\rightarrow \sfB^2\C^\times \subset \Tvs$ is the same as a natural 2-transformation $\mbf1 \Rightarrow \alpha$, where $\alpha$ is considered as a 2-functor to $\Tvs$.\footnote{This is the same thing as a higher fixed point for the representation $\alpha $ of $\mathscr{G}$.}  
\end{proposition}

\begin{proof}
This follows immediately from Definition \ref{Definition transformation Bicategory}. 
\end{proof}

\begin{remark}
We can use Proposition~\ref{prop:projrep} to define intertwiners between projective representations as modifications between the corresponding natural 2-transformations. 
\end{remark}

To apply this general formalism to the anomalous field theories at
hand, we introduce the Picard 2-groupoid $\mathsf{Pic}_2(\mathscr{B})$
of a monoidal 2-category $\mathscr{B}$ consisting of the objects of
$\mathscr{B}$ which are invertible with respect to the monoidal
product, and invertible 1-morphisms and 2-morphisms; there is a
canonical embedding $\mathsf{Pic}_2(\mathscr{B})\to\mathscr{B}$. An extended quantum field theory $\Aa$ is invertible if and only if it factors uniquely through $\Aa:\CobF\to\mathsf{Pic}_2(\Tvs)\hookrightarrow\Tvs$. 
 
We can pick an equivalence of 2-categories $\mathsf{Pic}_2(\Tvs)\rightarrow \mathsf{BLine}_\C$ by choosing a non-canonical equivalence between every invertible 2-vector space and $\fvs$; an inverse to this equivalence is given by the embedding $\mathsf{BLine}_\C\rightarrow \mathsf{Pic}_2(\Tvs)$.
The invertibililty of the anomaly quantum field theory $\mathcal{A}$ and this equivalence induces a 2-cocycle of the symmetry groupoid with values in $\fvs$: 
\begin{align*}
\alpha^\Aa:\Sym \CobF \longrightarrow \mathsf{BLine}_\C \ .
\end{align*}
The same argument as that used in Remark~\ref{rem:2cocycle}(b) shows that this cocycle is independent of the choice of equivalence $\mathsf{Pic}_2(\Tvs)\rightarrow \mathsf{BLine}_\C$ up to coboundary. Combining these facts with Proposition~\ref{prop:projrep} we can then infer

\begin{proposition}
Every anomalous quantum field theory $A:\mbf1\Rightarrow\mathsf{tr}\Aa$ induces a projective
representation of the symmetry groupoid $\Sym \CobF$. The
2-cocycle $\alpha^\Aa$ corresponding to this representation is unique up to coboundary.
\end{proposition}

We have seen in Proposition~\ref{prop:projrep} that natural 2-transformations $\mbf1\Rightarrow \alpha$ are the same things as projective representations of groupoids, so it should come as no surprise that these cocycles appear in the description of anomalies. The interesting prospect is that we can extend these cocycles to invertible extended field theories. This allows us to calculate quantities related to anomalies using the machinery of extended quantum field theories. Furthermore, we can couple such a theory to a bulk theory cancelling the anomaly as in Example~\ref{Example:TopologicalInsolature}. It is not clear that every anomaly admits such an extension, but all anomalies should give a projective representation of the symmetry groupoid.

\section{Extended quantum field theory and the parity anomaly\label{sec:EQFTparity}}

In this section we extend the quantum field theory $\mathcal{Z}_{\rm
  parity}^\zeta$ constructed in Section~\ref{sec:parityQGT} to an invertible
extended quantum field theory $\Aa_{\rm parity}^\zeta$ describing the parity anomaly of an anomalous field theory $A_{\rm parity}^\zeta:\mbf1 \Rightarrow\mathsf{tr}\mathcal{A}_{\rm parity}^\zeta$. This leads us naturally to the index theorem for manifolds with corners. We use the version of \cite{LoyaMelrose} involving b-geometry; a non-technical introduction to the topic can be found in \cite{Loyaindex}. An alternative approach can be found in~\cite{BunkeIndex}. The most important concepts from b-geometry for the ensuing formalism are summarised in Appendix~\ref{Appendix b-geometry}; a detailed introduction can be found in~\cite{MelrosebGeo}, see also~\cite{Sati} for a more physics oriented introduction. 

We fix the background field content introduced in Section~\ref{sec:catQFT} into the stack
\begin{align*}
\mathscr{F}= \mathsf{Bun}_G^\nabla \, \times \, \mathsf{Met} \, \times \, \mathsf{Spin} \, \times \, \mathsf{Or}
\end{align*}
for the geometric cobordism bicategory $\CobF$ constructed in
Section~\ref{Appendix Geometric bicategories}, together with a
finite-dimensional unitary representation $\rho_G$ of the gauge group
$G$ which specifies the matter field content. For technical reasons we
restrict ourselves to the sub-bicategory of $\CobF$ containing only
objects with vanishing index for the Dirac operator on every connected
component; this imposes conditions on the topology of each manifold $M^{d-2}$. 
This condition is a requirement for the existence of a well-defined index
theory on manifolds with corners~\cite{LoyaMelrose}. Similar restrictions also
appear in~\cite{BunkeIndex}. 
We further require that all structures on the collars are of product from. By a slight abuse of notation, we continue to call this bicategory $\CobF$.

\subsection{Index theory}
  
We have seen in Section \ref{Categorical case} that it is helpful to attach mapping cylinders to a manifold encoding the data of the identification of the boundary components with lower-dimensional objects. In the extended case we also need mapping boxes at the corners. 
Let $Y_i$, $i=1,2,3, 4$ be four closed manifolds equipped with $\mathscr{F}$-fields of product form $\mathsf{f}_i \in \mathscr{F}\big(Y_i\times (-\epsilon_1, \epsilon_1)^2\big)$, and a diagram of $\mathscr{F}$-diffeomorphisms $\varphi_{ij}$:
\begin{equation*}
\begin{tikzcd}
Y_1 \ar[d,swap,"\varphi_{13}"] \ar[r,"\varphi_{12}"] & Y_2\ar[d,"\varphi_{24}"] \\
Y_3 \ar[r,swap,"\varphi_{34}"] & Y_4 
\end{tikzcd}
\end{equation*}
Then the mapping box $\mathfrak{M}(Y, \varphi)$ of length
$\epsilon$ corresponding to this data is constructed by gluing
$Y_1\times \big[0, \tfrac{3}{4}\, \epsilon\big)^2$, $Y_2\times
\big(\tfrac{1}{4}\, \epsilon,\epsilon\big] \times \big[0,
\tfrac{3}{4}\, \epsilon\big)$, $Y_3 \times \big[0, \tfrac{3}{4}\,
\epsilon\big)\times  \big(\tfrac{1}{4}\, \epsilon,\epsilon\big]$ and
$Y_4\times \big( \tfrac{1}{4}\, \epsilon, \epsilon\big]^2$ along
$\varphi_{ij}$. Using descent we can construct an element $\mathsf{f}
\in \mathscr{F}\big(\mathfrak{M}(Y, \varphi)\big)$.
   
Given a regular 2-morphism $M^{d,2}$ from $M_1^{d-1,1}\colon M_- ^{d-2}\rightarrow M_+^{d-2}$ to $M_2^{d-1,1}\colon M_- ^{d-2}\rightarrow M_+^{d-2}$ in $\mathsf{Cob}_{d,d-1,d-2}^\mathscr{F}$, by definition it comes with collars $N^d_- \cong M_1^{d-1,1}\times [0,\epsilon_1)$ and $N^d_+ \cong M_2^{d-1,1}\times (-\epsilon_1,0]$.
We first attach mapping cylinders of a fixed length $\epsilon \in \R_{>0}$ to $M^{d-1,1}_1$, $M^{d-1,1}_2$ and the $0$-boundary. In a second step we attach mapping boxes of length $\epsilon$ to the corners of $M^{d,2}$. We denote this new manifold by $M'{}^{d,2}$ (see Figure \ref{Fig: Extended Indextheorem}). For this to be well-defined we need compatibility of all collars involved.
The new manifold has four distinct boundaries which we denote by
$M'_1{}^{d-1,1}$, $M'_2{}^{d-1,1}$, $C\big(M_-^{d-2}\big)= M_-^{d-2}\times \big[-\epsilon-\tfrac{1}{2}\, \epsilon_1, \epsilon+\tfrac{1}{2}\, \epsilon_1\big]$ and $C\big(M_+^{d-2}\big)= M_+^{d-2}\times \big[-\epsilon-\tfrac{1}{2}\, \epsilon_1, \epsilon+\tfrac{1}{2}\, \epsilon_1\big]$. 

\begin{figure}
\footnotesize
\begin{center}
\begin{overpic}[width=6cm,
scale=1]{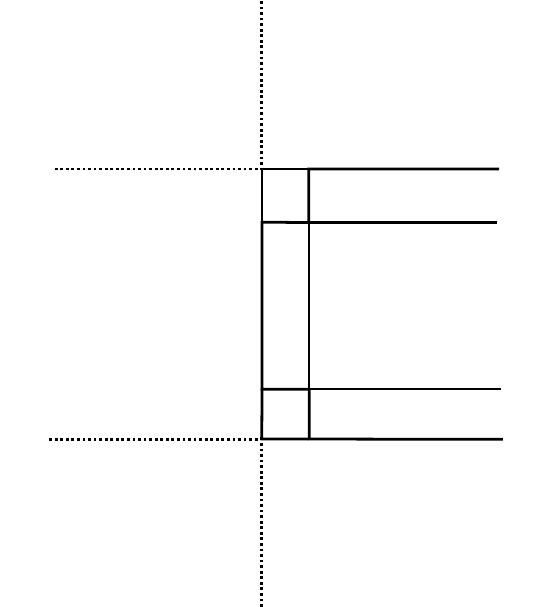}
\put(65,50){$M^{d,2}$}
\put(75,23){$\nwarrow$}
\put(80,20){$M_1^{d-1,1}$}
\put(75,74){$\swarrow$}
\put(80,77){$M_2^{d-1,1}$}
\put(-5,10){$M_-^{d-2}\times (-\infty,0]^2$}
\put(-5,50){$C(M_-^{d-2})\times (-\infty,0]^2$}
\put(-15,87){$M_-^{d-2}\times (-\infty,0]\times [0,\infty)$}
\put(55,10){$M'_1{}^{d-1,1}\times (-\infty,0]$}
\put(55,87){$M'_2{}^{d-1,1}\times [0, \infty)$}
\end{overpic}
\end{center}
\caption{\small Illustration of the construction of $\hat{M}'{}^{d,2}$ near $M_-^{d-2}$.}
\label{Fig: Extended Indextheorem}
\end{figure} 
 
We can now attach cylindrical ends to $M'{}^{d,2}$. For this, we first define
\begin{align*}
\hat{M}'{}^{d,2\,\circ}= M'{}^{d,2}\ \sqcup_{\partial M'{}^{d,2}}\ & \big( M'_1{}^{d-1,1} \times (-\infty, 0] \sqcup  M'_2{}^{d-1,1} \times [0, \infty ) \\
& \qquad \sqcup C(M_-^{d-2})\times (-\infty,0] \sqcup C(M_+^{d-2})\times [0,\infty) \big) \ ,
\end{align*} 
where we use the collars to glue the manifolds and extend all fields as products. Then $\hat{M}'{}^{d,2\,\circ}$ is a non-compact manifold with corners. Further gluing (see Figure~\ref{Fig: Extended Indextheorem}) produces 
\begin{align*}
\hat{M}'{}^{d,2} = \hat{M}'{}^{d,2\,\circ}\ &\sqcup_{\partial \hat{M}'{}^{d,2\,\circ}}\ \big( M_-^{d-2} \times (-\infty,0]^2 \sqcup M_-^{d-2} \times (-\infty,0]\times [0,\infty) \\ & \qquad \qquad \qquad \sqcup  M_+ ^{d-2} \times [0 , \infty)^2 \sqcup  M_+ ^{d-2} \times (-\infty, 0] \times [0, \infty) \big) 
\end{align*}
with all structures extended as products. As in the case of manifolds with boundaries, the Dirac operator $\slashed{D }_{\hat{M}'{}^{d,2}}$ is not Fredholm in general, and one can prove analogously that $\slashed{D }_{\hat{M}'{}^{d,2}}$ is Fredholm if and only if the induced Dirac operators on the corners and boundaries are invertible~\cite{LoyaMelrose}. 

When the kernel of the corner Dirac operator is non-trivial, we have
to add a mass perturbation~\cite{Loyaindex}. The induced twisted spinor bundle over $Y= M_+ ^{d-2} \sqcup -M_- ^{d-2}$ decomposes into spinors of positive and negative chirality. We pick a unitary self-adjoint isomorphism $T_i\colon \ker(\slashed{D }_{Y_i}) \rightarrow \ker(\slashed{D }_{Y_i})$, for every connected component $Y_i$ of the corner $Y$, which is odd with respect to the $\Z_2$-grading of the spinor bundle; this is possible since the index of $\slashed{D}_Y$ is $0$ by assumption. We define
\[
 T_\pm= \bigoplus_{i=1}^n\, T_{\pm,i} \qquad \text{ and } \qquad T=T_-\oplus T_+ \ ,
\] 
where $T_{\pm,i}:\ker\big(\slashed{D }_{M^{d-2}_{\pm,i}}\big)
\rightarrow \ker\big(\slashed{D }_{M^{d-2}_{\pm ,i}}\big)$. Now the
operator $\slashed{D }_Y - T$ is invertible. This suggests extending
$T$ to an operator $\hat{T}$ on $\hat{M}'{}^{d,2}$ such that the
massive Dirac operator $\slashed{D }_{\hat{M}'{}^{d,2}}-\hat{T}$ is Fredholm on weighted Sobolev spaces. A concrete construction of $\hat{T}$ can be found in \cite[Section 2.3]{LoyaMelrose},\footnote{$\hat{T}$ corresponds to $-S$ constructed in \cite[Section 2.3]{LoyaMelrose}, where we choose the same operators for the two remaining corners.} from which it is clear that $\hat{T}$ is independent of the length $\epsilon$ of the attached mapping cylinders and boxes. When we choose for every boundary component a small mass $\alpha_i$ as in Section~\ref{Categorical case}, then 
$$ 
\slashed{D }_{\hat{M}'{}^{d,2}}^+ -\hat{T}^+ \colon \e^{\alpha\cdot s}H^1\big(\hat S_{\hat M'{}^{d,2}}^+\big)\longrightarrow \e^{\alpha\cdot s}L^2\big(\hat S_{\hat M'{}^{d,2}}^-\big)
$$ 
is a Fredholm operator on weighted Sobolev spaces~\cite[Theorem
2.6]{LoyaMelrose}.
We restrict ourselves to a description of the corresponding index theorem on manifolds which are of the form $M'{}^{d,2}$ for a regular 2-morphism $M^{d,2}$ in $\mathsf{Cob}^{\mathscr{F}}_{d,d-1,d-2}$; the more general version can be found in~\cite[Theorem~6.13]{LoyaMelrose}.  

To define the $\eta$-invariant on a manifold $M^{d-1,1}$ with boundary
we proceed as in Section~\ref{Categorical case} and define
$\hat{M}^{d-1,1}$ by attaching cylindrical ends to $M^{d-1,1}$. In
general, the Dirac operator $\slashed{D}_{\hat{M}^{d-1,1}}$ has a
continuous spectrum, so we have to use the expression \eqref{Equation: Reformulation eta invariant} to define the $\eta$-invariant as an integral
\begin{align*}
{}^{\mathrm{b}}\eta \big(\slashed{D }_{\hat{M}^{d-1,1}}\big)= \frac{1}{\sqrt{\pi}} \, \int_0^\infty \, t^{-{1}/{2}} \ {}^{\mathrm{b}}\mathrm{Tr}\Big( \slashed{D} _{\hat{M}^{d-1,1}}\, \e^{-t\,\slashed{D }_{\hat{M}^{d-1,1}}^2} \Big) \ \diff t \ , 
\end{align*}
where we have to replace the usual trace by the b-geometric trace (see
Appendix \ref{Appendix b-geometry}) because its argument is not a
trace-class operator on $\hat{M}^{d-1,1}$ in general. There are other
ways of defining $\eta $-invariants for manifolds with boundaries
using appropriate boundary conditions \cite{DaiFreed,
  LeschWojciechowski, MuellerEtaInvariant}. The
$^{\mathrm{b}}\eta$-invariant agrees with the canonical boundary
conditions on spinors at the infinite ends of the cylinders induced by scattering Lagrangian subspaces, which we describe below.
 
There is a further contribution to the index theorem coming from the
corners. We define for $M'_i{}^{d-1,1}$, $i=1,2$ the scattering
Lagrangian subspace
\begin{align*}
\Lambda_{C_i} = \Big\lbrace \, \lim_{s_- \rightarrow -\infty}\, \Psi(y_-,s_-)\oplus  \lim_{s_+ \rightarrow \infty}\, \Psi(y_+,s_+) \ \Big| \ \Psi\in C^\infty\big(\hat{S}_{\hat{M}'_i{}^{d-1,1}}\big) \cap \ker\big(\slashed{D } _{\hat{M}'_i{}^{d-1,1}}\big) \ \text{bounded} \, \Big\rbrace 
\end{align*}
where $s_-\in(-\infty,0]$, $s_+\in[0,\infty)$ and $y_\pm\in M_\pm^{d-2}$.
The set $\Lambda_{C_i}\subset \ker \big(\slashed{D }_Y\big)$ is a Lagrangian subspace of $\ker \big(\slashed{D }_{\partial M^{d-1,1}_i}\big)$ with respect to the symplectic form $\omega (\, \cdot\, , \, \cdot\, )= (\Gamma \, \cdot\, ,\, \cdot\, )_{L^2}$~\cite{MuellerEtaInvariant}, where $\Gamma$ is the chirality operator
on the spinor bundle over the corners, i.e. $\Gamma$ acts by scalar multiplication with $\pm\ii$ on $S_Y^\pm$.
We define an odd unitary self-adjoint isomorphism $C_i$ of $\ker
\big(\slashed{D }_{Y}\big)$, called the scattering matrix of
$\slashed{D } _{\hat{M}'_i{}^{d-1,1}}$, by setting $C_i=\id$ on
$\Lambda_{C_i}$ and $C_i=-\id$ on $\Lambda_{C_i}^\perp$. We denote by
$\Lambda_{T}\subset \ker \big(\slashed{D }_Y\big)$ the $+1$-eigenspace
of $T$; there is a one-to-one correspondence between Lagrangian
subspaces of $\ker \big(\slashed{D }_Y\big)$ and boundary conditions
$T\in \mathrm{End}_\C \big(\ker (\slashed{D }_Y)
\big)$. Following~\cite{LeschWojciechowski,Bunke}, we introduce the
`exterior angle' between Lagrangian subspaces by the spectral formula
\begin{align}
\label{EQ: Definition of m( , )}
\mu(\Lambda_{T}, \Lambda_{C_i})=-\frac{1}{\pi} \ \sum_{\substack{\e^{\ii\theta}\in \mathrm{spec}(-T^-\,C_i^+)\\ -\pi<\theta <\pi}} \, \theta \ ,
\end{align} 
where the grading is with respect to the $\Z_2$-grading of the twisted spinor bundle over the corners. 

With this notation, we can now formulate the index theorem for
manifolds of the form $M'{}^{d,2}$ as

\begin{theorem}
\label{Theorem extended index theorem}
If $M^{d,2}$ is a regular 2-morphism in $\CobF$, then
\begin{align}
\label{Index theorem corners}
\begin{split}
\mathrm{ind} \big(\slashed{D }_{\hat{M}'{}^{d,2}}^+ -\hat{T}^+
\big) = &\int_{M'{}^{d,2}} \, K_{\rm AS} -\frac{1}{2}\, \Big( \,
^{\mathrm{b}}\eta \big(\slashed{D}_{\hat{M}'_1{}^{d-1,1}}
\big)+{}^{\mathrm{b}}\eta \big(\slashed{D }_{\hat{M}'_2{}^{d-1,1}}
\big) \\ & +\, \dim \ker \big(
\slashed{D}_{\hat{M}'_1{}^{d-1,1}} \big) -\dim \ker
\big(\slashed{D}_{\hat{M}'_2{}^{d-1,1}} \big) \\ & +\,
\dim(\Lambda_{T} \cap \Lambda_{C_1}) -\dim(\Lambda_T \cap
\Lambda_{C_2} ) + \mu(\Lambda_{T}, \Lambda_{C_1}) + \mu(\Lambda_{ T},
\Lambda_{C_2}) \, \Big) \ .
\end{split}
\end{align}
\end{theorem}

\begin{remark}
The extra corner contributions in the last line of \eqref{Index theorem corners} to the usual
(b-geometric) Atiyah-Patodi-Singer formula \eqref{APS Theorem} can be
understood as follows. Let $M^{d-1,1}$ be a regular
1-morphism in the bicategory $\CobF$. Then for every $T\in
\mathrm{End}_\C\big(\ker (\slashed{D }_{\partial M^{d-1,1}}) \big)$ as
above we can relate the spectral data of the massive Dirac operator on
$\hat M^{d-1,1}$ to their massless counterparts as
\begin{align*}
{}^{\rm b}\eta\big(\slashed{D}_{\hat M^{d-1,1}} - \hat T \big)&= {}^{\rm b}\eta\big(\slashed{D}_{\hat M^{d-1,1}} \big)+\mu(\Lambda_{T},
  \Lambda_{C}) \ , \\[4pt]
\dim \ker \big(\slashed{D}_{\hat M^{d-1,1}} - \hat T \big) &= \dim
                                                             \ker
                                                             \big(\slashed{D}_{\hat
                                                             M^{d-1,1}}\big)+\dim(\Lambda_T\cap\Lambda_C)
                                                             \ ,
\end{align*}
where $\Lambda_C$ is the scattering Lagrangian subspace for
$M^{d-1,1}$. 
\end{remark}

\begin{remark}\label{Remark: Relation eta Invariants}
We describe the relation between $^{\rm b}\eta $-invariants and
$\eta $-invariants with boundary conditions~\cite{RelationEtaInvariants}. We denote by $\Pi_+$ the
projection onto the space spanned by the positive eigenspinors of
$\slashed{D }_{\partial M^{d-1,1}}$, and by $\Pi_T$ the projection onto the positive eigenspace $\Lambda_T$. This
allows us to define a Dirac operator $\D_T$, which coincides with $\slashed{D
}_{\partial M^{d-1,1}}$, on the domain
\begin{align*}
\big\{
\Psi\in H^1\big(\hat S_{\hat M^{d-1,1}}\big) \ \big| \ (\Pi_+ + \Pi_T)\Psi\big|_{\partial M^{d-1,1}}=0
\big\} \ .
\end{align*}
The operator $\D_T$ is self-adjoint and elliptic for all $T$. It is shown in \cite[Theorem~1.2]{RelationEtaInvariants} that
\begin{align}\label{EQ: b-eta= eta with boundary condition}
\eta(\D_T)= {}^{\rm b} \eta \big(\slashed{D}_{\hat M^{d-1,1}} \big)+\mu(\Lambda_{T},
  \Lambda_{C}) \ ,
\end{align}  
so that we can combine the ${}^{\rm
  b}\eta$-invariant and the exterior angle $\mu$ in \eqref{Index
  theorem corners} into an $\eta$-invariant for a Dirac operator with suitable boundary conditions induced by the Lagrangian subspace $\Lambda_T$.    
\end{remark}

\begin{proof}[Proof of Theorem \ref{Theorem extended index theorem}]
From the general index theorem for manifolds with corners \cite[Theorem 6.13]{LoyaMelrose} we get 
\begin{align*}
&\text{ind} \big(\slashed{D }_{\hat{M}'{}^{d,2}}^+ -\hat{T}^+
  \big) \\ & \quad = \int_{M'{}^{d,2}}\, K_{\rm AS} -\frac{1}{2}\, \Big( \, {}^{\rm
          b}\eta \big(\slashed{D}_{\hat{M}'_1{}^{d-1,1}} \big)+{}^{\rm
          b}\eta \big(\slashed{D }_{\hat{M}'_2{}^{d-1,1}} \big) +
          {}^{\rm b}\eta \big(\slashed{D}_{\hat C(M^{d-2}_-)} \big)+
          {}^{\rm b}\eta \big(\slashed{D}_{\hat C(M^{d-2}_+)} \big)\\ 
& \qquad + \dim \ker \big(\slashed{D}_{\hat{M}'_1{}^{d-1,1}}
  \big)+\dim(\Lambda_{T} \cap \Lambda_{C_1})-\dim \ker
  \big(\slashed{D}_{\hat{M}'_2{}^{d-1,1}} \big)-\dim(\Lambda_T \cap
  \Lambda_{C_2}) \\
& \qquad +\dim \ker \big(\slashed{D}_{\hat{C}(M^{d-2}_-)} \big)
  +\dim(\Lambda_{-\ii\Gamma\, T_-} \cap \Lambda_{C_-}) -\dim \ker
  \big(\slashed{D}_{\hat{C}(M^{d-2}_+)} \big) -\dim(\Lambda_{-\ii\Gamma\,
  T_+} \cap \Lambda_{C_+}) \\
& \qquad +\mu(\Lambda_{T}, \Lambda_{C_1})+ \mu(\Lambda_{T}, \Lambda_{C_2})
  +\mu(\Lambda_{-\ii\Gamma\, T_+}, \Lambda_{C_+}) +
  \mu(\Lambda_{-\ii\Gamma\, T_-}, \Lambda_{C_-}) \, \Big)
\end{align*}
where $\Lambda_{C_\pm}$ are the scattering Lagrangian subspaces for
$C(M^{d-2}_\pm)$, respectively. We can calculate the contributions
from the boundaries $C(M^{d-2}_\pm)$ explicitly and show that they all
vanish. 

Attaching infinite cylindrical ends to $C(M^{d-2}_\pm)$ leads to the manifolds $M_\pm^{d-2}\times (-\infty,\infty)$.
The Dirac operator on the manifold $M_\pm^{d-2}\times (-\infty, \infty)$ of odd dimension $d-1$ is given by \eqref{Equation: Dirac operator on a product}:
\begin{align}
\label{EQ: Proof index theorem}
\slashed{D}_{{M_\pm^{d-2}} \times (-\infty, \infty)}= \sigma_t\,
  \big(\slashed{D}_{M_\pm^{d-2}} + \partial_t \big) \ , 
\end{align}
where $t\in(-\infty,\infty)$. 
We are interested in the dimension of the space of harmonic
spinors $\Psi(y_\pm,t)$. By elliptic regularity there exists a basis of smooth
sections. Multiplying \eqref{EQ: Proof index theorem} with
$\sigma_t^{-1}$, we get
\begin{align*}
\big(\slashed{D}_{M_\pm^{d-2}}+ \partial_t \big)\Psi(y_\pm,t)=0 \ .
\end{align*} 
Using separation of variables $\Psi(y_\pm,t)=\psi(y_\pm)\, \alpha(t)$, this
equation reduces to a pair of equations
\begin{align}
\label{EQ: Dirac operator on cylinder Proof index theorem}
\slashed{D}_{M_\pm^{d-2}} \, \psi(y_\pm)= \lambda \, \psi(y_\pm)  \qquad
  \mbox{and} \qquad
\frac{\diff \alpha(t)}{\diff t}=-\lambda \, \alpha(t) \ , 
\end{align}
for an arbitrary constant $\lambda $ which must be real since $\slashed{D}_{M_\pm^{d-2}} $ is an elliptic operator. 
The second equation has solution (up to a constant) $\alpha(t)= \e^{-\lambda\, t}$, and
we finally see that there are no non-zero square-integrable spinors
$\Psi(y_\pm,t)$ with eigenvalue $0$.
Hence the contributions from the terms $\dim \ker \big(\slashed{D}_{\hat{C}(M_\pm^{d-2})}\big)$ are $0$.

A solution of \eqref{EQ: Dirac operator on cylinder Proof index
  theorem} is bounded if and only if $\lambda=0$, and so the
scattering Lagrangian subspace takes the form 
\begin{align*}
\Lambda_{C_\pm}= \Delta\big(\ker (\slashed{D}_{M_\pm^{d-2}}) \big)=
  \big\{ \psi \oplus \psi \ \big| \ \psi \in \ker
  \big(\slashed{D}_{M_\pm^{d-2}} \big) \big\} \subset \ker
  \big(\slashed{D}_{M_\pm^{d-2}} \big) \oplus \ker
  \big(\slashed{D}_{-{M_\pm^{d-2}}} \big) \ .
\end{align*}
This implies that 
\begin{align*}
\dim(\Lambda_{-\ii\Gamma \, T_\pm} \cap \Lambda_{C_\pm})= 0 \ ,
\end{align*}
since the chirality operator $\Gamma$ on the outgoing and ingoing
boundaries differs by a sign while $T_\pm$ is the same over both boundaries.

Finally, by Remark~\ref{Remark: Relation eta Invariants}, ${}^{\rm
  b}\eta \big(\slashed{D}_{\hat{C}(M^{d-2}_\pm)}\big)+
\mu(\Lambda_{-\ii\Gamma \, T_\pm}, \Lambda_{C_{\pm}})$ is the
$\eta$-invariant on a cylinder with identical boundary conditions at
both ends, which vanishes by~\cite[Theorem~2.1]{LeschWojciechowski}. 
\end{proof}

For later use we derive here a formula for the index of a 2-morphism
under cutting. For this, we first have to study the behaviour of the
various quantities in the index formula \eqref{Index theorem corners} under orientation-reversal.

\begin{lemma}\label{Lemma: Properties terms in the index theorem}
Let $M^{d-1,1}$ be a regular 1-morphism in $\CobF$ with fixed boundary condition
$T \in \mathrm{End}_\C \big(\ker(\slashed{D}_{\partial M^{d-1,1}}) \big)$ as
above. If we reverse the orientation of $M^{d-1,1}$, then $T$ still
defines a suitable boundary condition of $\slashed{D}_{\partial (-M^{d-1,1})}$ and 
\begin{align*}
\dim \ker\big(\slashed{D}_{\hat M^{d-1,1}} \big) &= \dim
                                                     \ker\big(\slashed{D}_{-\hat
                                                     M^{d-1,1}} \big)
                                                     \ \qquad
                                                     \mbox{and} \qquad
{}^{\rm b}\eta \big(\slashed{D}_{\hat{M}^{d-1,1}} \big)=-{}^{\rm
                                                     b}\eta
                                                     \big(\slashed{D}_{-\hat{M}^{d-1,1}}
                                                     \big) \ , \\[4pt] 
\dim(\Lambda_T \cap \Lambda_C)&=\dim(\Lambda_T \cap \Lambda_{-C})
                                \qquad \mbox{and} \qquad
\mu(\Lambda_{T}, \Lambda_{C})=-\mu(\Lambda_{T}, \Lambda_{-C}) \ ,
\end{align*}
where $\Lambda_{-C}$ is the scattering Lagrangian subspace for $-M^{d-1,1}$. 
\end{lemma}

\begin{proof}
There is an equality $\slashed{D}_{\hat M^{d-1,1}}=-\slashed{D}_{-\hat M^{d-1,1}}$
of operators acting on sections of the underlying twisted
spinor bundle $\hat S_{\hat M^{d-1,1}}$, which implies the first two
equations. The Lagrangian subspaces $\Lambda_{C}$ and $\Lambda_{T}$
are independent of the orientation, which implies the third equation. 

We can interpret the exterior angle $\mu(\Lambda_{T}, \Lambda_{C})$ as
the $\eta$-invariant of a cylinder with boundary conditions induced by
$\Lambda_{T}$ and $\Lambda_{C}$~\cite{LeschWojciechowski}. Reversing
the orientation of this cylinder corresponds to $\mu(\Lambda_{T},
\Lambda_{-C})$. The last equation then follows from the fact that the
$\eta$-invariant changes sign under orientation-reversal.
\end{proof}

\begin{proposition}\label{Corollary: Index theorem for composition of regular 2-morphisms}
The index is additive under vertical composition of regular 2-morphisms in
$\CobF$ if we choose identical boundary conditions on the corners. 
\end{proposition}

\begin{proof}
The contributions from the gluing boundary cancel each other by Lemma \ref{Lemma: Properties terms in the index theorem}.  We still have to show that 
\begin{align*}
\int_{(M_2^{d,2}\circ M_1^{d,2})'}\, K_{\rm AS}=\int_{M_1'{}^{d,2}}\,
  K_{\rm AS}+\int_{M_2'{}^{d,2}}\, K_{\rm AS} \ .
\end{align*}
This is not completely obvious, since the vertical composition also
involves deleting half of the collars of the gluing boundary. However,
from
\eqref{Index theorem corners} and the construction of
$\hat{M}'{}^{d,2}$ it is clear that $\int_{M'{}^{d,2}}\, K_{\rm AS}$
is independent of the length of the collars. Using the description of
gluing in terms of mapping cylinders we can cover $(M_2^{d,2}\circ
M_1^{d,2})'$ by $\tilde{M}_1'{}^{d,2}$ and $\tilde{M}_2'{}^{d,2}$,
where $\tilde{M}_i'{}^{d,2}$ is the manifold $M_i'{}^{d,2}$ with $\tfrac{3}{4}$ of the collar corresponding to the gluing boundary removed. 
\end{proof}

\subsection{The extended quantum field theory $\mathcal{A}_{\rm parity}^\zeta$} \label{Section: Extending A}

We shall now proceed to extend the quantum field theory
$\mathcal{Z}_{\rm parity}^\zeta$ to an anomaly quantum field theory
$\Aa_{\rm parity}^\zeta:\CobF\to\Tvs$ describing the parity anomaly in $d-1$ dimensions. Following Section~\ref{Categorical case},
we would like to define something like $\mathcal{A}_{\rm
  parity}^\zeta(M^{d,2}) = \zeta^{\text{ind} (\slashed{D
  }_{\hat{M}'{}^{d,2}}^+ -\hat{T}^+ )}$ for a fixed
$\zeta\in\C^\times$ and every regular 2-morphism
$M^{d,2}$ of $\CobF$.
The problem with this definition is that the index may depend on our
choice of $T \in \text{End}_\C\big(\ker (\slashed{D}_Y) \big)$. The
resolution is to include the data about the choice of $T$ into our
extended quantum field theory. 

We do this by combining, for each object $M^{d-2}$ of $\CobF$, all possible boundary conditions $T$ into a
category $\mathcal{A}_{\rm parity}^\zeta(M^{d-2})$ in the following way: Let $\mathsf{T}(M^{d-2})$ be the category with one
object for every odd self-adjoint unitary $T\in \text{End}_\C\big(\ker (\slashed{D}_{M^{d-2}}) \big) $ and one
morphism between every pair of objects. The
category $\mathcal{A}_{\rm parity}^\zeta(M^{d-2})$ is then defined to be the finite completion 
of the $\C$-linearisation $\C\mathsf{T}(M^{d-2}) = \{\C\,T_i \mid T_i\in\mathsf{T}(M^{d-2})\}$ of the category $\mathsf{T}(M^{d-2})$. A concrete model is given, for example, by the functor category 
\begin{equation}
\mathcal{A}_{\rm parity}^\zeta(M^{d-2})\cong \big[\mathsf{T}(M^{d-2}), \fvs \big]
\label{eq:functorcat}\end{equation}
of pre-cosheaves of complex vector spaces on $\mathsf{T}(M^{d-2})$. 

We can construct a $\C$-linear functor $\mathcal{A}_{\rm parity}^\zeta (M^{d-1,1}):\Aa_{\rm parity}^\zeta(M_-^{d-2})\to\Aa_{\rm parity}^\zeta(M_+^{d-2})$ for
a regular 1-morphism $ M^{d-2}_- \times
[0,\epsilon_-)\overset{\varphi_-}{\hookrightarrow} M^{d-1,1}
\overset{\varphi_+}{\hookleftarrow} M^{d-2}_+\times (-\epsilon_+,0]$
in $\CobF$ by using the corresponding boundary and corner contributions to the index formula \eqref{Index theorem corners} to build a complex line encoding all possible boundary conditions on $M_+^{d-2}$. For generators $T_\pm$ of $\mathcal{A}_{\rm parity}^\zeta(M^{d-2}_\pm)$, we define 
\begin{align*}
I_{M^{d-1,1}}(T_-,T_{+})=&-\tfrac{1}{2} \, \big( \, {}^{\rm
                                               b}\eta(\slashed{D}_{\hat{M}'{}^{d-1,1}}
                                               )-\dim \ker
                                               (\slashed{D}_{\hat{M}'{}^{d-1,1}}
                                               ) \\ & \qquad \qquad - \dim(\Lambda_{T_-\oplus T_{+}} \cap \Lambda_C)+ \mu(\Lambda_{T_-\oplus T_{+}}, \Lambda_{C}) \, \big) \ .
\end{align*}
We will sometimes drop the subscript $M^{d-1,1}$ when it is obvious from the context.
We fix a generator $T_-$ of $\mathcal{A}_{\rm parity}^\zeta (M^{d-2}_-)$.
Then we get a morphism in $\mathcal{A}_{\rm parity}^\zeta (M^{d-2}_+)$ from generator $T_{+,i}$ to generator $T_{+,j}$ in $\mathsf{T}(M^{d-2}_+)$ by multiplying the special morphism between $T_{+,i}$ and $T_{+,j}$ with 
\begin{align*}
\zeta^{I_{M^{d-1,1}}(T_-,T_{+,j})-I_{M^{d-1,1}}(T_-,T_{+,i})} \ .
\end{align*} 
These morphisms fit into a family of diagrams $\Ja\big(M^{d-1,1}\big)(T_-)\colon \mathsf{T}(M^{d-2}_+)\rightarrow \mathcal{A}_{\rm parity}^\zeta (M^{d-2}_+)$. We finally define the complex line
\begin{align*}
\mathcal{A}_{\rm parity}^\zeta \big(M^{d-1,1}\big)(T_-)= \lim_{\longleftarrow} \, \Ja\big(M^{d-1,1}\big)(T_-) \ .
\end{align*}
To be precise, we have to invoke the axiom of choice to pick a representative for the limit. We define $\mathcal{A}_{\rm parity}^\zeta \big(M^{d-1,1}\big)(f)$ for a morphism $f\colon T_-\rightarrow T'_-$ in $\sfT(M^{d-2}_-)$ to be the unique map induced by the matrix of compatible morphisms 
\begin{align}
\label{EQ: Definition A on morphisms natural transformation}
\big(\zeta^{I_{M^{d-1,1}}(T'_-,T_{+,j})-I_{M^{d-1,1}}(T_-,T_{+,i})} \big)
\end{align}
via the functoriality of limits.
We can summarise our definition of $\Aa_{\rm parity}^\zeta$ on regular 1-morphisms by the commutative diagram
\begin{equation}\label{Definition on 1-Morphisms}
\begin{tikzcd}
 & & & &\big[ \sfT(M^{d-2}_+), \mathcal{A}_{\rm parity}^\zeta(M^{d-2}_+)\big] \ar[dd, "{\lim\limits_{\longleftarrow}}"] \\
 & & & & \\
 \ \ \ \ \ \mathcal{A}_{\rm parity}^\zeta(M^{d-2}_-)\ar[rrrruu, "{\Ja_{\rm cont}(M^{d-1,1})}"] \ar[rrrr,swap, "{\mathcal{A}_{\rm parity}^\zeta(M^{d-1,1})}"]& & & & \mathcal{A}_{\rm parity}^\zeta(M^{d-2}_+) \\
 \sfT(M^{d-2}_-)\ar[rrrruuu, controls={+(-4,4) and +(0,0)}, "{\Ja(M^{d-1,1})}",in=180] \ar[u, hook] & & & &
\end{tikzcd}
\end{equation}
where $\Ja_{\rm cont}(M^{d-1,1})$ is the continuous extension of $\Ja(M^{d-1,1})$, which is unique up to a canonical natural isomorphism.  

For a regular 2-morphism $M^{d,2}:M_1^{d-1,1}\Rightarrow M_2^{d-1,1}$, we define a family of morphisms in $\mathcal{A}_{\rm parity}^\zeta(M^{d-2}_+)$ by
\begin{align}
\xi\big(M^{d,2}\big)_{T_-,T_{+,i}} = \zeta^{\text{ind}
  (\slashed{D}^+_{\hat{M}'{}^{d,2}}-\widehat{T_-\oplus T_{+,i}}{}^+)}\cdot
  \text{id}_{T_{+,i}} \colon T_{+,i} \longrightarrow T_{+,i} \ .
\label{eq:familyofmorphisms}\end{align}
\begin{proposition} The family of morphisms
  \eqref{eq:familyofmorphisms} induces a natural transformation 
\[
\mathcal{A}_{\rm parity}^\zeta\big(M^{d,2}\big) = \text{id}_{\lim\limits_{\longleftarrow}}\bullet \xi_{\rm
  cont}\big(M^{d,2}\big) \colon \mathcal{A}_{\rm
  parity}^\zeta\big(M^{d-1,1}_1\big) \Longrightarrow \mathcal{A}_{\rm
  parity}^\zeta\big(M^{d-1,1}_2\big) 
\]
as indicated in the diagram
\begin{equation}\label{Definition on 2-Morphisms}
\begin{small}
\begin{tikzcd}
 & & & & & & & &\big[ \sfT(M^{d-2}_+), \mathcal{A}_{\rm parity}^\zeta(M^{d-2}_+)\big] \ar[dddd, "{\lim\limits_{\longleftarrow}}"] \\
 & & & & & & & &\\
  & & & & & & & &\\
   & & & & & & & &\\
 \sfT(M^{d-2}_-) \ar[rrrr, hook] \ar[rrrrrrrruuuu, bend left = 30,
 "\Ja(M_1^{d-1,1})"{name=A1}]\ar[rrrrrrrruuuu, bend left = 12, swap, "\Ja(M_2^{d-1,1})"{name=A2}] & & & & \mathcal{A}_{\rm parity}^\zeta(M^{d-2}_-)\ar[rrrruuuu, bend left = 15, ""{name=B1}]\ar[rrrruuuu, bend right = 15, ""{name=B2}] \ar[rrrr, "{\mathcal{A}_{\rm parity}^\zeta(M_1^{d-1,1})}", ""{name=C1}] \ar[rrrr, "{\mathcal{A}_{\rm parity}^\zeta(M_2^{d-1,1})}", bend right, swap, , ""{name=C2}]& & & & \mathcal{A}_{\rm parity}^\zeta(M^{d-2}_+) 
\arrow[Rightarrow, from=A1, to=A2, shorten <= 1em, shorten >=1em, "{\xi(M^{d,2})}", pos=0.5]
\arrow[Rightarrow, from=B1, to=B2, shorten <= 1em, "{\xi_{\rm cont}(M^{d,2})}"]
\arrow[Rightarrow, from=C1, to=C2, shorten <= 0.5em, shorten >= 0.5em, "{\mathcal{A}_{\rm parity}^\zeta(M^{d,2})} "]
\end{tikzcd}
\end{small}
\end{equation}
where $\xi_{\rm cont}(M^{d,2})$ is the unique extension of
$\xi(M^{d,2})$ to a natural transformation between the continuous
extensions of $\Ja(M_1^{d-1,1})$ and $\Ja(M_2^{d-1,1})$.
\end{proposition}

\begin{proof}
The only thing we need to show is that $\xi(M^{d,2})$ is actually a natural transformation. 
First of all, we have to prove that we really do get a diagram from the definition \eqref{eq:familyofmorphisms}. For this, we have to prove that the diagram
\begin{equation}\label{Diagram definition A on morphisms}
\begin{tikzcd}
\vdots \ar{d} & & & & \vdots \ar{d} \\
T_{+,i} \ar{d} \ar{d} \ar[dd,swap , bend right=30,"{\zeta^{I(M^{d-1,1}_1,T_-,T_-,T_{+,i},T_{+,j})}}"]  \ar[rrrr,"{\zeta^{\text{ind} (\slashed{D}^+_{\hat{M}'{}^{d,2}}-\widehat{T_-\oplus T_{+,i}}{}^+)}}"] & & & & T_{+,i} \ar{d} \ar[dd,bend left=30,"{\zeta^{I(M^{d-1,1}_2,T_-, T_-,T_{+,i},T_{+,j})}}"] 
\\
\vdots \ar{d} & & & & \vdots \ar{d} \\
T_{+,j} \ar{d} \ar[rrrr,swap,"{\zeta^{\text{ind} (\slashed{D}^+_{\hat{M}'{}^{d,2}}-\widehat{T_-\oplus T_{+,j}}{}^+)}}"]& & & & T_{+,j} \arrow{d} \\
\vdots  & & & &  \vdots 
\end{tikzcd}\
\end{equation}
commutes, where we defined
\begin{align*}
I(M^{d-1,1},T_{-},T'_{-},T_{+,i},T_{+,j})=I_{M^{d-1,1}}(T'_{-},T_{+,j})-I_{M^{d-1,1}}(T_{-},T_{+,i}) \ .
\end{align*}
This is an immediate consequence of the index theorem \eqref{Index
  theorem corners}; note that this only works because the incoming and outgoing boundaries contribute with different signs. Thus the diagram \eqref{Diagram definition A on morphisms} induces a morphism $\mathcal{A}_{\rm parity}^\zeta(M^{d,2})_{T_-} \colon \mathcal{A}_{\rm parity}^\zeta(M^{d-1,1}_1)(T_-)\rightarrow \mathcal{A}_{\rm parity}^\zeta(M^{d-1,1}_2)(T_-)$. To show that these morphisms fit into a natural transformation it is enough to show that the diagram
\begin{equation*}
\begin{tikzcd}
T_{+,i} \ar[d, swap, "{\zeta^{I(M^{d-1,1}_1,T_-,T'_-,T_{+,i},T_{+,i})}}"] \ar[rrrr,"{\zeta^{\text{ind} (\slashed{D}^+_{\hat{M}'{}^{d,2}}-\widehat{T_-\oplus T_{+,i}}{}^+)}}"] & & & & T_{+,i} \ar[d,"{\zeta^{I(M^{d-1,1}_2,T_-,T'_-,T_{+,i},T_{+,i})}}"]
\\
T_{+,i}  \ar[rrrr,swap,"{\zeta^{\text{ind} (\slashed{D}^+_{\hat{M}'{}^{d,2}}-\widehat{T'_-\oplus T_{+,i}}{}^+)}}"] & & & & T_{+,i} 
\end{tikzcd}
\end{equation*}
commutes. This follows again immediately from the index theorem. 
\end{proof}

We define the theory $\mathcal{A}_{\rm parity}^\zeta$ on limit 1-morphisms as the functor corresponding to a mapping cylinder. From the definition of $\mathcal{A}_{\rm parity}^\zeta(M^{d-1,1})$ it is clear that this is independent of the length of the mapping cylinder, since only the behaviour at infinity is important. On limit 2-morphisms we define the theory to be the value of $\mathcal{A}_{\rm parity}^\zeta$ on a mapping cylinder of length $\epsilon$. 
This completes the definition of $\mathcal{A}_{\rm parity}^\zeta$.

Before demonstrating that $\mathcal{A}_{\rm parity}^\zeta$ is an extended quantum field theory, we explicitly calculate the functors corresponding to limit 1-morphisms.

\begin{proposition}
\label{Proposition: A on limit 1-morphisms}
Let $\phi \colon M^{d-2}_- \rightarrow M^{d-2}_+$ be a limit 1-morphism in $\CobF$ with mapping cylinder $\mathfrak{M}(\phi)$, and let $T_\pm$ be fixed objects in $\sfT(M^{d-2}_\pm)$. Then 
\begin{align}
\label{EQ: A on limit 1-morphisms}
I_{\mathfrak{M}(\phi)}(T_-,T_+) = -\tfrac{1}{2} \, \big( \dim (\Lambda_{T_-} \cap \Lambda_{\phi^\ast T_+}) + \mu(\Lambda_{T_-} ,\Lambda_{\phi^\ast T_+}) \big) \ .
\end{align}
\end{proposition}

\begin{proof}
By Remark \ref{Remark: Relation eta Invariants} we get a term in $I_{\mathfrak{M}(\phi)}(T_-,T_+)$ corresponding to the $\eta$-invariant with boundary conditions induced by the Lagrangian subspaces $\Lambda_{T_\pm}$.
There is a diffeomorphism induced by $\phi^{-1}$ and $\text{id}$ from the mapping cylinder of $\phi$ with length $1$ to the cylinder $M^{d-2}_-\times [0,1]$. The boundary conditions change to new boundary conditions induced by $\Lambda_{T_-}$ and $\Lambda_{\phi^\ast T_+}$. 
The $\eta$-invariant for this situation was calculated in~\cite[Theorem 2.1]{LeschWojciechowski}, from which we get
\begin{align*}
\eta (\D_{T_-,\phi^\ast T_+})= \mu(\Lambda_{T_-},\Lambda_{\phi^\ast T_+}) \ .
\end{align*}
We can extend the diffeomorphism induced by $\phi^{-1}$ and $\id$ above to manifolds with cylindrical ends attached. The expression (\ref{EQ: A on limit 1-morphisms}) then follows from similar arguments to those used in the proof of Theorem~\ref{Theorem extended index theorem}. 
\end{proof}

\begin{theorem}\label{Theorem: A is extended field theory}
 \ $\mathcal{A}_{\rm parity}^\zeta \colon \mathsf{Cob}_{d,d-1,d-2}^\mathscr{F} \rightarrow \Tvs$ is an invertible extended quantum field theory.
\end{theorem}

\begin{proof}
We construct a family of natural isomorphisms $\Phi_{M^{d-2}} \colon \id \Rightarrow \mathcal{A}_{\rm parity}^\zeta(\id_{M^{d-2}})$. For this, it is enough to construct a natural isomorphism $\Phi'_{M^{d-2}}\colon \Ca \Rightarrow \Ja({M^{d-2}\times [0,1]})$, where $\Ca$ is the functor sending an object $T$ of $\sfT(M^{d-2})$ to the constant diagram on $T$. The collection of special morphisms from $T_i$ to $T$ multiplied with $\zeta^{I_{M^{d-2}\times [0,1]}(T,T_i)}$ induces such an isomorphism:
\[
\begin{tikzcd}
T_i \ar[d] \ar[dd, bend right=40,"
{\zeta^{I(T,T_j)-I(T,T_i)}}", swap]  & & \\
\vdots \ar[d] & & T\ar[llu,"\zeta^{I(T,T_i)}",swap] \ar[lld, "\zeta^{I(T,T_j)}"] \\
T_j  & & 
\end{tikzcd}
\]
The naturality with respect to morphisms in $\sfT(M^{d-2})$ follows immediately from the commuting diagram
\[
\begin{tikzcd}
T_1 \ar[d,swap,"1"] \ar[rrr,"\zeta^{I(T_1,T_i)}"] & & &  T_i \ar[d,"{\zeta^{I(T_2,T_i)-I(T_1,T_{i})}}"] \\
 T_2 \ar[rrr,swap,"\zeta^{I(T_2,T_i)}"] & & & T_i
\end{tikzcd}
\] 

For the composition of regular 1-morphisms we have to construct
natural $\C$-linear isomorphisms 
\begin{align}\label{eq:Phinatural}
\Phi_{M^{d-1,1}_1,M^{d-1,1}_2} \colon \mathcal{A}_{\rm parity}^\zeta\big(M^{d-1,1}_2\big) \circ \mathcal{A}_{\rm parity}^\zeta\big(M^{d-1,1}_1\big)\Longrightarrow \mathcal{A}_{\rm parity}^\zeta\big(M^{d-1,1}_2\circ M^{d-1,1}_1\big) \ .
\end{align}
Using \eqref{Definition on 1-Morphisms} we get the diagram
\[
\begin{tikzcd}
 &  &\big[ \sfT(M^{d-2}_{1+}), \mathcal{A}_{\rm
   parity}^\zeta(M^{d-2}_{1+})\big] \ar[rr, "{\lim\limits_{\longleftarrow}
   \, \circ \, \Ja^\star }"] \ar[dd, "{\lim\limits_{\longleftarrow}}"]  &  &\big[ \sfT(M^{d-2}_{2+}), \mathcal{A}_{\rm parity}^\zeta(M^{d-2}_{2+})\big] 
 \ar[dd, "{\lim\limits_{\longleftarrow}}"] \\
 &  & &  &  \\
 \mathcal{A}_{\rm parity}^\zeta(M^{d-2}_{1-})\ar[rrrruu, controls={+(1,6) and +(0,2)}, "{\Ja(M^{d-1,1}_2\circ M^{d-1,1}_1)}" description,in=90, name=U] \ar[rruu, "{\Ja(M^{d-1,1}_1)}" description] \ar[rr,swap, "{\mathcal{A}_{\rm parity}^\zeta(M^{d-1,1}_1)}"]&  & \mathcal{A}_{\rm parity}^\zeta(M^{d-2}_{1+}) \ar[rruu, "{\Ja(M^{d-1,1}_2)}" description] \ar[rr,swap, "{\mathcal{A}_{\rm parity}^\zeta(M^{d-1,1}_2)}"]&  & \mathcal{A}_{\rm parity}^\zeta(M^{d-2}_{2+}) 
\end{tikzcd}
\]
with $M^{d-2}_{1+}=M^{d-2}_{2-}$. The lower part of this diagram is commutative up to a canonical
natural isomorphism coming from the universal property of the limit,
since by definition $\Ja$ is continuous; this isomorphism depends on
the concrete realisation of the limit that we pick. Our goal is to now
define a natural isomorphism 
$$
\Phi' \colon
\lim\limits_{\longleftarrow} \, \circ \, \Ja^\star \circ
\Ja\big(M^{d-1,1}_1\big) \Longrightarrow  \Ja\big(M^{d-1,1}_2\circ M^{d-1,1}_1\big)
$$ 
which then induces the natural isomorphism \eqref{eq:Phinatural}. 

We begin by evaluating $\Ja^\star \circ \Ja(M^{d-1,1}_1)$ on a fixed
object $T_-$ of  $\sfT(M^{d-2}_{1-})$ to get the family of diagrams
\begin{align*}
\Ja^\star \circ \Ja\big(M^{d-1,1}_1\big)(T_-) \colon
                                               \sfT(M^{d-2}_{1+})
                                               \longrightarrow
                                               \big[\sfT(M^{d-2}_{2+}),
                                               \mathcal{A}_{\rm
                                               parity}^\zeta(M^{d-2}_{2+}) \big]
\end{align*}         
defined by
\begin{align*}                                      
T_1 &\!\longmapsto\! \left[ \begin{small} \begin{array}{c}
\Ja^\star\circ \Ja\big(M^{d-1,1}_1\big)(T_-)[T_1]\colon \sfT(M^{d-2}_{2+})
                                        \rightarrow \mathcal{A}_{\rm
                                        parity}^\zeta\big(M^{d-2}_{2+} \big) \\[4pt]
T_2 \mapsto T_2 \\[4pt]
(f \colon T_2\rightarrow T_2') \mapsto \big(\zeta^{I(T_1,T'_2)-I(T_1,T_2)} \cdot f \colon T_2 \rightarrow T'_2\big)
\end{array} \end{small} \right] \ , \\[4pt]
(g \colon T_1 \rightarrow T'_1) &\!\longmapsto\! \left[ \begin{small} \begin{array}{c}
\Ja^\star \circ \Ja\big(M^{d-1,1}_1\big)(T_-)[g] \colon \Ja^\star \circ \Ja\big(M^{d-1,1}_1\big)(T_-)[T_1]\Rightarrow \Ja^\star \circ \Ja\big(M^{d-1,1}_1\big)(T_-)[T'_1] \\[4pt]
\Ja^\star \circ \Ja\big(M^{d-1,1}_1\big)(T_-)[g]_{T_2}\colon T_2[T_1] \xlongrightarrow{\zeta^{I(T_-,T'_1)-I(T_-,T_1)+I(T'_1,T_2)-I(T_1,T_2)}} T_2[T'_1]
\end{array} \end{small} \right] .
\end{align*}
Since the
category $\mathcal{A}_{\rm parity}^\zeta(M^{d-2}_{2+})$ is complete, we can calculate the limit of this diagram objectwise. For this, we fix an object $T_2$ of $\sfT(M^{d-2}_{2+})$. A realisation for the limit is then given by the cone
\begin{equation*}
\begin{tikzcd}
\cdots \ar[r]  & T_2[T_1] \ar[rrrr, "\zeta^{ I(T_-,T'_1)-I(T_-,T_1)+I(T'_1,T_2)-I(T_1,T_2) }"] & &  & &  T_2[T'_1]\ar[r] & \cdots \\
 & & &  & & & \\
  & &  & & & & \\
& & &   \Ja\big(M^{d-1,1}_2\circ M^{d-1,1}_1\big)(T_-)[T_2]=T_2 \ar[llluuu, bend left, out=20]\ar[rrruuu, bend right, out=-20] \ar[lluuu, "{\zeta^{I(T_-,T_1)+I(T_1,T_2)-I(T_-,T_2)}}" description]  \ar[rruuu,swap, "{\zeta^{I(T-,T'_1)+I(T'_1,T_2)-I(T_-,T_2)}}" description]    & & &
\end{tikzcd}
\end{equation*}
This cone is universal, since all maps involved are isomorphisms, and
so by the universal property of limits we get an isomorphism 
\begin{align*}
\Phi^{\prime \, -1}_{T_-,T_2} \colon \Ja\big(M^{d-1,1}_2\circ
  M^{d-1,1}_1\big)(T_-)[T_2] \longrightarrow \lim_{\longleftarrow} \,
  \circ \, \Ja^\star \circ \Ja\big(M^{d-1,1}_1\big)(T_-)[T_2] \ .
\end{align*}

To show that this construction is natural in $T_2$ it is enough to observe that the diagram
\begin{equation*}
\begin{tikzcd}
\Ja\big(M^{d-1,1}_1\big)(T_-)[T_2]=T_2 \ar[dd,swap,  "{\zeta^{I(T_-,T_1)+I(T_1,T_2)-I(T_-,T_2)}}" ] \ar[rrrr, "\zeta^{I(T_-,T'_2)-I(T_-,T_2)}"]& & & &\Ja\big(M^{d-1,1}_1\big)(T_-)[T'_2]=T'_2 \ar[dd,  "{\zeta^{I(T_-,T_1)+I(T'_1,T_2)-I(T'_-,T_2)}}" ] \\
& & & & \\
T_2[T_1] \ar[rrrr, swap, "{\zeta^{I(T'_1,T_2)-I(T_1, T_2)}}"] & & &  & T'_2[T_1]
\end{tikzcd}
\end{equation*}
commutes. From the commutativity of the diagram
\begin{equation*}
\begin{tikzcd}
\Ja\big(M^{d-1,1}_1\big)(T_-)[T_2]=T_2 \ar[dd,swap,  "{\zeta^{I(T_-,T_1)+I(T_1,T_2)-I(T_-,T_2)}}" ] \ar[rrrr, "{\zeta^{I(T'_-,T_2)-I(T_-,T_2)}}"] & & & & \Ja\big(M^{d-1,1}_1\big)(T'_-)[T_2]=T_2 \ar[dd,  "{\zeta^{I(T'_-,T_1)+I(T_1,T_2)-I(T'_-,T_2)}}" ] \\
& & & & \\
T_2[T_1] \ar[rrrr, swap, "{\zeta^{I(T'_-,T_1)-I(T_-,T_1)}}"] & & & & T_2[T_1]
\end{tikzcd}
\end{equation*}
it follows that this construction is also natural in $T_-$. This completes the construction of the natural isomorphism $\Phi$. 

In order for $\Phi$ to equip $\mathcal{A}_{\rm parity}^\zeta$ with the structure of a 2-functor, we need to check naturality with respect to 2-morphisms and associativity. We start with the compatibility with 2-morphisms. Using the index theorem, this follows from the calculation
\begin{align*}
&\log_{\zeta}\mathcal{A}_{\rm parity}^\zeta\big(M_2^{d,2}\bullet M_1^{d,2}\big)_{T_-,T_+} \\
& \qquad \qquad \quad = \int_{(M_2^{d,2}\bullet M_1^{d,2})'} \, K_{\rm AS} +I_{\partial_+(M_2^{d,2}\bullet M_1^{d,2})}(T_-,T_+)-I_{\partial_-(M_2^{d,2}\bullet M_1^{d,2})}(T_-,T_+) 
\\[4pt]
& \qquad \qquad \quad = \int_{M'_1{}^{d,2}} \, K_{\rm AS}+\int_{M'_2{}^{d,2}} \, K_{\rm AS}  +I_{\partial_+(M_2^{d,2}\bullet M_1^{d,2})}(T_-,T_+)-I_{\partial_-(M_2^{d,2}\bullet M_1^{d,2})}(T_-,T_+)
\\& \qquad \qquad \qquad + \big(I_{\partial_+M_1^{d,2}}(T_-,T_1)+I_{\partial_+M_2^{d,2}}(T_1,T_+) \big) - \big(I_{\partial_+M_1^{d,2}}(T_-,T_1) +I_{\partial_+M_2^{d,2}}(T_1,T_+) \big)
\\& \qquad \qquad \qquad + \big(I_{\partial_-M_1^{d,2}}(T_-,T_1)+I_{\partial_-M_2^{d,2}}(T_1,T_+) \big)- \big(I_{\partial_-M_1^{d,2}}(T_-,T_1)+I_{\partial_-M_2^{d,2}}(T_1,T_+) \big) 
\\[4pt]
& \qquad \qquad \quad = \log_{\zeta}\mathcal{A}_{\rm parity}^\zeta\big(M_1^{d,2}\big)_{T_-,T_1}+ \log_{\zeta}\mathcal{A}_{\rm parity}^\zeta\big(M_2^{d,2}\big)_{T_1,T_+} \\ 
& \qquad \qquad \qquad +\big(I_{\partial_+(M_2^{d,2}\bullet M_1^{d,2})}(T_-,T_+)-I_{\partial_+M_1^{d,2}}(T_-,T_1) -I_{\partial_+M_2^{d,2}}(T_1,T_+)\big) \\
& \qquad \qquad \qquad -\big(I_{\partial_-(M_2^{d,2}\bullet M_1^{d,2})}(T_-,T_+)-I_{\partial_-M_1^{d,2}}(T_-,T_1) -I_{\partial_-M_2^{d,2}}(T_1,T_+)\big) \\[4pt]
& \qquad \qquad \quad = \log_{\zeta}\mathcal{A}_{\rm parity}^\zeta\big(M_1^{d,2}\big)_{T_-,T_1}+ \log_{\zeta}\mathcal{A}_{\rm parity}^\zeta\big(M_2^{d,2}\big)_{T_1,T_+} \\ 
& \qquad \qquad \qquad +\log_{\zeta}\big(\Phi_{\partial_+ M_1^{d,2},\partial_+M_2^{d,2}}\big)_{T_-,T_+}+\log_{\zeta}\big(\Phi^{-1}_{\partial_- M_1^{d,2},\partial_-M_2^{d,2}}\big)_{T_-,T_+} \ ,
\end{align*}
for all objects $T_1$ of $\sfT\big(\partial_+ \partial_+ M_1^{d,2}\big)$.

It remains to demonstrate compatibility with associativity: $\Phi \circ (\Phi \bullet\id) = \Phi \circ ( \id\bullet \Phi) $, i.e. the coherence condition \eqref{EQ1: Definition 2Functor}. For this, we fix three composable $d-1$-dimensional manifolds $M^{d-1,1}_i$, $i=1,2,3$ with incoming boundaries $M^{d-2}_{i-}$ and outgoing boundaries $M^{d-2}_{i+}$. By the naturality of all constructions, it is enough to check the relation for fixed objects $T_-$ of $ \sfT(M^{d-2}_{1-})$, $T_1$ of $ \sfT(M^{d-2}_{2-})=\sfT(M^{d-2}_{1+})$, $T_2 $ of $ \sfT(M^{d-2}_{3-})= \sfT(M^{d-2}_{2+})$, and $T_+$ of $ \sfT(M^{d-2}_{3+})$. This follows immediately from the commutative diagram
\begin{equation*}
\begin{tikzcd}
& & & & T_+ \ar[ddddrrrr, "{\zeta^{I(T_1,T_+)-I(T_2,T_1)-I(T_2,T_+)}}" description] \ar[ddddllll, "{\zeta^{I(T_-,T_2)-I(T_-,T_1)-I(T_1,T_2)}}" description] \ar[dddddddd, "{\zeta^{I(T_-,T_+)-I(T_2,T_1)-I(T_2,T_+)-I(T_-,T_1)}}" description] & & & & \\
 & & & &    & & & & \\
  & & & &    & & & & \\
   & & & &    & & & & \\
 T_+ \ar[ddddrrrr, "{\zeta^{I(T_-,T_+)-I(T_-,T_2)-I(T_2,T_+)}}" description] & & & & & & & & T_+ \ar[ddddllll, "{\zeta^{I(T_-,T_+)-I(T_-,T_1)-I(T_1,T_+)}}" description] \\
  & & & &    & & & & \\
   & & & &    & & & & \\
 & & & &    & & & & \\
 & & & & T_+ & & & &
\end{tikzcd}
\end{equation*}   

We finally have to check the coherence condition \eqref{EQ2: Definition 2Functor}.
We fix a regular 1-morphism $M^{d-1,1}\colon M^{d-2}_-\rightarrow M^{d-2}_+$. The situation can be represented diagramatically by
\begin{equation*}
\begin{small}
\begin{tikzcd}
&  & \ar[d,swap, Leftarrow, shorten >= 1, shorten <= -25, "\Phi'"] &  &  \\
 &  &\big[ \sfT(M^{d-2}_+), \mathcal{A}_{\rm parity}^\zeta(M^{d-2}_+) \big] \ar[rr, "{\lim\limits_{\longleftarrow} \, \circ \, \Ja^\star }", bend left=15] \ar[dd,swap, "{\lim\limits_{\longleftarrow} }"]  &  &\big[ \sfT(M^{d-2}_{+}), \mathcal{A}_{\rm parity}^\zeta(M^{d-2}_{+}) \big] 
 \ar[dd, "{\lim\limits_{\longleftarrow}}"] \\
 &  & &  &  \\
 \mathcal{A}_{\rm parity}^\zeta(M^{d-2}_{-})\ar[rrrruu, controls={+(1,6) and +(0,2)}, "{\Ja(M^{d-1,1})}" description,in=90, name=U] \ar[rrrr, bend right=15, in=-110, "\mathcal{A}_{\rm parity}^\zeta(M^{d-1,1})",swap] \ar[rruu, "{\Ja(M^{d-1,1})}" description] \ar[rr,swap, "{\mathcal{A}_{\rm parity}^\zeta(M^{d-1,1})}"]&  & \mathcal{A}_{\rm parity}^\zeta(M^{d-2}_+)  \ar[rruu, "{\Ja(\id)}" description, bend left=15, ""{name=ar2}]\ar[rruu, "{\Ca}" description, bend right=15, ""{name=ar1}] \ar[rr, "{\mathcal{A}_{\rm parity}^\zeta(\id)}", bend right, swap] \ar[rr,"\id" , swap] \ar[d,swap, Rightarrow,shorten <= 0, shorten >= -10, "\Phi_{M^{d-1,1},\id}" ]& \ & \mathcal{A}_{\rm parity}^\zeta(M^{d-2}_{+})   \\
  & & \ &\ar[u,swap, Leftarrow,shorten <= 0 , shorten >= 5,"\Phi_{M_+^{d-2}}"] &
 \ar[from=ar1, to=ar2, Rightarrow, shorten <= 2 , shorten >= 15,swap, "\Phi'_{M^{d-2}_+}", pos=0.2]
\end{tikzcd}
\end{small}
\end{equation*}
where here we abbreviate the identity regular 1-morphism
$M_+^{d-2}\times[0,1]$ by $\id$.
We have to show that the composition of the natural transformations in the lower part of this diagram is the identity. We can do this by showing that the composition in the upper part is the identity. We evaluate the resulting natural transformation at a fixed object $T_-$ of $\sfT(M^{d-2}_-)$. By naturality we can also fix an object $T_+$ of $\sfT(M^{d-2}_+)$. Then the composition gives
\begin{equation}
\label{EQ:Proof of C2}
\begin{tikzcd}
\big( T_+ \ar[rr, "{\zeta^{I_{\id}(T_+,T_+)}}"] & & T_+ \ar[rrrrrrr,
  "{\zeta^{I_{M^{d-1,1}}(T_-,T_+)-I_{M^{d-1,1}}(T_-,T_+)-I_{\id}(T_+,T_+)}}"]
  & & & & & & & T_+\big) = \big( T_+ \ar[r,"\id"]& T_+\big) \ ,
\end{tikzcd}
\end{equation}
which proves the condition \eqref{EQ2: Definition 2Functor}. The coherence condition for $\mathcal{A}_{\rm parity}^\zeta(\id)\circ \mathcal{A}_{\rm parity}^\zeta(M^{d-1,1})$ can be proven in the same way.

Next we come to the vertical composition of regular 2-morphisms. It is
enough to show that the composition is given by multiplication for
fixed objects $T_\pm $ of $\mathcal{A}_{\rm parity}^\zeta(M^{d-2}_\pm)$. This follows immediately from Proposition~\ref{Corollary: Index theorem for composition of regular 2-morphisms} by an argument similar to the one used in the proof of Theorem~\ref{A is field theory}.
The conditions for limit 1-morphisms and limit 2-morphisms follow now from their representations as mapping cylinders. 

Now we check compatibility with the monoidal structure. There are canonical $\C$-linear equivalences of categories given on objects by
\[
\chi^{-1}_{M^{d-2}, M'{}^{d-2}}\colon \mathcal{A}_{\rm parity}^\zeta(M^{d-2} \sqcup M'{}^{d-2}) \longrightarrow \mathcal{A}_{\rm parity}^\zeta(M^{d-2})\boxtimes \mathcal{A}_{\rm parity}^\zeta(M'{}^{d-2})
\] 
sending $(T, T')\in \text{End}_\C\big(\ker (\slashed{D}_{M^{d-2}}) \big)\oplus \text{End}_\C\big(\ker (\slashed{D}_{M'{}^{d-2}}) \big) \cong \text{End}_\C\big(\ker (\slashed{D}_{M^{d-2}\sqcup M'{}^{d-2}}) \big)$ to $T \boxtimes T'$, and 
\[
\iota^{-1} \colon \mathcal{A}_{\rm parity}^\zeta(\emptyset) \longrightarrow \fvs
\]
sending $0\in\{0\}=\text{End}_\C\big(\ker (\slashed{D}_{\emptyset}) \big)$ to $\C$. All further structures required for $\mathcal{A}_{\rm parity}^\zeta$ to be a symmetric monoidal 2-functor are trivial. It is straightforward if tedious to check that all diagrams in the definition of a symmetric monoidal 2-functor commute, but we shall not write them out explicitly.
Finally, it is straightforward to see that $\mathcal{A}_{\rm parity}^\zeta$ factors through the Picard 2-groupoid $\mathsf{Pic}_2(\Tvs)$, and hence $\mathcal{A}_{\rm parity}^\zeta$ is invertible. 
\end{proof}

\begin{remark}
The proof of Theorem \ref{Theorem: A is extended field theory} is more or less independent of the concrete form of $I_{M^{d-1,1}}(T_-,T_+)$ and the index theorem. It only uses additivity under vertical composition and the decomposition 
\[
\text{ind}\big(\slashed{D}^+_{\hat{M}'{}^{d,2}}-\widehat{T_-\oplus T_{+}}{}^+ \big) = \int_{M'{}^{d,2}} \, K_{\rm AS} +I_{\partial_+M^{d,2}}(T_-,T_+)-I_{\partial_-M^{d,2}}(T_-,T_+) \ ,
\]
into a local part and a global part depending solely on boundary conditions. Hence it should be possible to apply this or a similar construction to a large class of invariants depending on boundary conditions. 
A particularly interesting example would involve $\eta$-invariants on odd-dimensional manifolds with corners, which should be related to chiral anomalies in even dimensions and extend Dai-Freed theories~\cite{DaiFreed}. 
\end{remark}

\subsection{Projective representations and symmetry-protected topological phases\label{sec:projparity}}

A quantum field theory with parity anomaly is now regarded as a theory relative to $\Aa_{\rm parity}^\zeta$ as described in Section~\ref{sec:anomaliesextended}, i.e. a natural
symmetric monoidal 2-transformation $A_{\rm parity}^\zeta:\mbf1\Rightarrow\mathsf{tr}\Aa_{\rm parity}^\zeta$. The concrete description of the extended quantum field theory
$\mathcal{A}_{\rm parity}^\zeta$ given in the proof of Theorem~\ref{Theorem:
  A is extended field theory} allows us to calculate the corresponding
groupoid 2-cocycle along the lines discussed in
Section~\ref{sec:projreps}; this information about the parity anomaly is contained in the isomorphism \eqref{eq:Phinatural}. We choose a $\C$-linear equivalence of
categories $\chi \colon \mathcal{A}_{\rm parity}^\zeta(M^{d-2}) \rightarrow
\fvs$ sending all objects $T$ of $\sfT(M^{d-2})$ to $\C$ and all
morphisms $f$ of $\sfT(M^{d-2})$ to $\text{id}_\C$; a weak inverse is
given by picking a particular object $T_{M^{d-2}}$ in $\sfT(M^{d-2})$ and
mapping $\C$ to $T_{M^{d-2}}$. The functor $\Aa_{\rm parity}^\zeta(\phi)$ corresponding to a limit 1-morphism
$\phi \colon M^{d-2}_1 \rightarrow M^{d-2}_2$ in the symmetry groupoid
$\Sym\CobF$ is given by taking the tensor product with the complex line
\[
\begin{tikzcd}
L_{\chi,\phi} = \lim\limits_{\stackrel{\scriptstyle\longleftarrow}{\footnotesize\sfT(M^{d-2}_2)}} \
\Big(\C\,T_i \ar[rrrrrrr,"{\zeta^{I_{\mathfrak{M}(\phi)}(T_{\footnotesize
      M^{d-2}_1},T_{j})-
    I_{\mathfrak{M}(\phi)}(T_{\footnotesize M^{d-2}_1},T_{i})}}"] & & & & & &&
\C\,T_j \Big) \ ,
\end{tikzcd}
\]
where $\mathfrak{M}(\phi)$ is the mapping cylinder of $\phi $.
A choice of an object $T_{M_1^{d-2}} $ of $\sfT(M_1^{d-2})$ defines an
isomorphism $\varphi_{M_1^{d-2}} \colon L_{\chi,\phi} \rightarrow \C
$, which for simplicity we pick to be the same boundary mass
perturbation as chosen for the weak inverse above.  The groupoid
cocycle evaluated at $\phi_1 \colon M^{d-2}_1 \rightarrow M^{d-2}_2$
and $\phi_2 \colon M^{d-2}_2 \rightarrow M^{d-2}_3$ corresponding to
this choice is then given by
\begin{align*}
\alpha^{\Aa_{\rm parity}^\zeta}_{\phi_1, \phi_2}= \zeta^{I_{\mathfrak{M}(\phi_2 \circ \phi_1)}(
  T_{M^{d-2}_1}, T_{M^{d-2}_3})-I_{\mathfrak{M}(\phi_1)}(
  T_{M^{d-2}_1}, T_{M^{d-2}_2})-I_{\mathfrak{M}(\phi_2)}(
  T_{M^{d-2}_2}, T_{M^{d-2}_3})} \ .
\end{align*}
We can evaluate this expression explicitly by using \eqref{EQ: A on limit
  1-morphisms} to get
\begin{align}
\log_\zeta \alpha^{\Aa_{\rm parity}^\zeta}_{ \phi_1, \phi_2}=& -\frac{1}{2} \, \Big(\dim
                                      \big(\Lambda_{T_{M^{d-2}_1}}\cap \Lambda_{                                      
                                      \phi_1^\ast\,\phi_2^\ast T_{M^{d-2}_3}} \big) +
  \mu\big(\Lambda_{T_{M^{d-2}_1}},
  \Lambda_{\phi_1^\ast\, \phi_2^\ast T_{M^{d-2}_3}} \big) \nonumber \\ & \qquad \qquad 
  -\dim \big(\Lambda_{T_{M^{d-2}_1}
                                      }\cap \Lambda_{
                                      \phi_1^\ast T_{M^{d-2}_2}} \big)-\mu\big(\Lambda_{T_{M^{d-2}_1}}, \Lambda_{\phi_1^\ast
  T_{M^{d-2}_2}} \big) \nonumber \\ & \qquad \qquad 
  -\dim
                                      \big(\Lambda_{T_{M^{d-2}_2}}\cap
                                      \Lambda_{\phi_2^\ast
                                      T_{M^{d-2}_3}}\big)
                                       -\mu\big(\Lambda_{T_{M^{d-2}_2}},\Lambda_{ \phi_2^\ast T_{M^{d-2}_3}}
  \big) \Big) \ .
\label{eq:2cocycleexplicit}\end{align} 

To calculate the part of the 2-cocycle involving identity 1-morphisms we can use \eqref{EQ:Proof of C2} to get
\begin{align*}
\alpha^{\Aa_{\rm parity}^\zeta}_{\phi, \id_{M^{d-2}}}= \alpha^{\Aa_{\rm parity}^\zeta}_{\id_{M^{d-2}},\phi} =
  \zeta^{-I_{M^{d-2}\times [0,1]}(T_{M^{d-2}},T_{M^{d-2}})}=
  \zeta^{-\frac{1}{4} \dim\ker (\D_{M^{d-2}})} \ ,
\end{align*} 
where the last equality follows from \eqref{EQ: A on limit
  1-morphisms}. From a physical point of view it is natural to assume
this to be equal to $1$, since the identity limit morphism should
still be a non-anomalous symmetry of every quantum field theory. 
We can achieve this by normalising our anomaly quantum field theory
$\Aa_{\rm parity}^\zeta$ to the theory $\tilde\Aa_{\rm parity}^\zeta$ obtained by
redefining 
\[
\tilde I_{M^{d-1,1}}(T_-,T_+) = I_{M^{d-1,1}}(T_-,T_+)+
\tfrac{1}{8}\, \big( \dim \ker (\D_{\partial_- M^{d-1,1}}) + \dim \ker
(\D_{\partial_+ M^{d-1,1}}) \big) \ .
\]   
The proof of Theorem \ref{Theorem: A is extended field theory} then
carries through verbatum with $I_{M^{d-1,1}}$ replaced by $\tilde
  I_{M^{d-1,1}}$ everywhere.

\begin{example}
We conclude by illustrating how to extend
Example~\ref{Example:TopologicalInsolature} to the
anomaly quantum field theory $\mathcal{A}_{\rm parity}^{(-1)}$, glossing
over many technical details.
To construct the second quantized Fock space of a quantum
field theory of fermions coupled to a background gauge field on a
Riemannian manifold
$M^{d-2}$, one needs a polarization
\[
H=H^+\oplus H^-
\]  
of the one-particle Hilbert space $H$ of
wavefunctions, which we take to be the sections of the twisted spinor
bundle $S_{M^{d-2}}$. If the Dirac Hamiltonian $\D_{M^{d-2}}$ has no zero
  modes, then there exists a canonical polarization given by taking $H^+=H^{>0}$
  (resp. $H^-=H^{<0}$) to be the space spanned by the positive
  (resp. negative) energy eigenspinors. 
Given such a polarization we can define 
\[
A_{\rm parity}^{(-1)}(M^{d-2})= \mbox{$\bigwedge$} H^+\otimes \mbox{$\bigwedge$}
(H^-)^\ast \ ,
\]  
where $\bigwedge H$ denotes the exterior algebra generated by the
vector space $H$. Now time-reversal (or orientation-reversal) symmetry acts
by interchanging $H^+$ and $H^-$, and there is no problem extending
this symmetry to the Fock space $A_{\rm parity}^{(-1)}(M^{d-2})$. 

In the case that $\ker(\D_{M^{d-2}})$ is non-trivial, as is the case
for fermionic gapped quantum phases of matter, one could try to declare all zero modes to belong to $H^{>0}$ or $H^{<0}$ and use the corresponding polarization to define a Fock space. We cannot apply this method of quantization, since it breaks time-reversal symmetry. 
Therefore we are forced to use a different polarization compatible
with orientation-reversal symmetry. There is no canonical choice for
such a polarization, but rather a natural family parameterized by Lagrangian subspaces $\Lambda_T \subset \ker(\D_{M^{d-2}})$:
\[
H^+(\Lambda_T) = H^{>0}\oplus \Lambda_T \qquad \text{and} \qquad H^-(\Lambda_T) = H^{<0}\oplus \Gamma \Lambda_T \ . 
\] 
Since orientation reversion acts proportionally to the chirality
operator $\Gamma$ on spinors, these polarizations are compatible with the symmetry. 
We then get a family of Fock spaces
\[
A_{\rm parity}^{(-1)}(M^{d-2},T)= \mbox{$\bigwedge$} H^+(\Lambda_T)\otimes \mbox{$\bigwedge$}
H^-(\Lambda_T)^\ast = \mbox{$\bigwedge$} H^{>0}\otimes
\mbox{$\bigwedge$} \big(H^{<0}\big)^{ \ast } \otimes F(M^{d-2},T) \ ,
\] 
where the essential part for our discussion is encoded in the
finite-dimensional vector space
$$
F(M^{d-2},T)=\mbox{$\bigwedge$} \Lambda_T \otimes \mbox{$\bigwedge$} (
\Gamma \Lambda_T)^\ast \ .
$$ 
Fixing an ordered basis for every $\Lambda_T$, these
finite-dimensional vector spaces fit into a $\fvs $-valued pre-cosheaf
$A_{\rm parity}^{(-1)}(M^{d-2})$ on $\sfT(M^{d-2})$, where we assign to a morphism $T_1 \rightarrow T_2$ the linear map induced by sending the fixed basis of $\Lambda_{T_1}$ to the basis of $\Lambda_{T_2}$.    
By \eqref{eq:functorcat} this is an element of $\mathcal{A}^{(-1)}_{\text{parity}}(M^{d-2})$, or equivalently a $\C$-linear functor
\[
A_{\rm parity}^{(-1)}(M^{d-2}) : \fvs \longrightarrow \mathcal{A}^{(-1)}_{\text{parity}}(M^{d-2}) \ .
\]

We sketch how these pre-cosheaves fit into a natural symmetric
monoidal 2-transformation,
realising an anomalous quantum field theory $A_{\rm parity}^{(-1)}$ with parity anomaly
according to Definition~\ref{Definition Anommalous field theory}.
For a 1-morphism $M^{d-1,1}\colon M^{d-2}_- \rightarrow M^{d-2}_+$ we have to construct a natural transformation $A^{(-1)}_\text{parity} (M^{d-1,1})\colon \mathcal{A}^{(-1)}_\text{parity}(M^{d-1,1})\circ A^{(-1)}_\text{parity} (M^{d-2}_-) \Rightarrow A^{(-1)}_\text{parity} (M^{d-2}_+)$. The left-hand side is given by the pre-cosheaf 
\begin{eqnarray*}
\sfT (M^{d-2}_+)  & \longrightarrow & \fvs \ , \\
T & \longmapsto & \lim\limits_{\stackrel{\scriptstyle\longleftarrow}{\footnotesize\sfT(M^{d-2}_-)}} \ \Big(
\begin{tikzcd}
\dots \rightarrow F(M^{d-2}_-,T_-)\ar[rrr,"{(-1)^{I(T_-',T)-I(T_-,T)}}"] & & & F(M^{d-2}_-,T_-') \rightarrow  \dots
\end{tikzcd}
\Big) \ .
\end{eqnarray*}
This implies that constructing a natural transformation is the same as
defining a family of compatible linear maps\footnote{This requires replacing the limit by a colimit, which is possible since it is taken over a groupoid inside $\fvs$.} $A^{(-1)}_\text{parity}(M^{d-1,1})_{T_-,T_+}\colon F(M_-^{d-2},T_-)\rightarrow F(M_+^{d-2},T_+)$. These should again be given by an appropriate regularization of path integrals. As before we assume that these maps are well-defined up to a sign. 
To fix the sign we have to consistently fix reference background
fields on all 1-morphisms. This is possible, for example, by using a connection on the universal bundle and pullbacks along classifying maps. 
Again we can fix the sign at these reference fields to be positive. Using a spectral flow similar to \eqref{eq:spectralflow} with boundary conditions $T_-$ and $T_+$, we can fix the sign for all other field configurations. 
Assuming that this spectral flow can be calculated by the index with
appropriate boundary conditions, we see that these sign ambiguities
satisfy the coherence conditions encoded by
$\mathcal{A}^{(-1)}_{\text{parity}}$, i.e. they define a natural symmetric
monoidal 2-transformation. 
This demonstrates in which sense a field theory with parity anomaly takes values in $\mathcal{A}^{(-1)}_{\text{parity}}$. 
\end{example}

\subsection*{Acknowledgments}

We thank Severin Bunk and Lukas Woike for helpful discussions.
This work was completed while
R.J.S. was visiting the National Center for Theoretical Sciences in
Hsunchin, Taiwan during August/September 2017, whom he warmly thanks
for support and hospitality during his stay there. This work was supported by the COST Action MP1405 QSPACE, funded by the
European Cooperation in Science and Technology (COST). The work of L.M. was
supported by the Doctoral Training Grant ST/N509099/1 from the UK Science and Technology
Facilities Council (STFC).
The work of
R.J.S. was supported in part by the STFC
Consolidated Grant ST/L000334/1. 

\appendix

\section{Manifolds with corners\label{Appendix corners}}

In this appendix we collect information about manifolds with corners necessary for our constructions, following \cite[Chapter~3.1]{schommer2011classification} for the most part. We also give a short introduction to the concepts of geometry on such manifolds which are used in the main text, following~\cite{AnalyticSurgery,MelrosebGeo,ThesisLoya}. 

\subsection{Basic definitions}\label{Appendix Manifolds with corners..}

Roughly speaking, a manifold of dimension $d$ is a topological space which locally looks like open subsets of $\R^d$. The idea behind manifolds 
with corners of codimension $2$ is to replace $\R^d$ by $\R^{d-2}\times \R_{\geq0}^2$; we denote by $\mathrm{pr}_{\R_{\geq0}^2}:\R^{d-2}\times \R_{\geq0}^2\to \R_{\geq0}^2$ the projection. 
A chart for a subset $U$ of a topological space $X$ is then a homeomorphism $\varphi \colon U \rightarrow V \subset \R^{d-2}\times \R_{\geq0}^2$. 
Two charts $\varphi_1 \colon U_1 \rightarrow V_1$ and $\varphi_2 \colon U_2 \rightarrow V_2$ are compatible if $\varphi_2 \circ \varphi_1^{-1}\colon \varphi_1(U_1\cap U_2)\rightarrow \varphi_2(U_1\cap U_2)$ is a diffeomorphism. 
A map between subsets of $\R^{d-2}\times \R_{\geq0}^2$ is smooth if there exists an extension to open subsets of $\R^d$ which is smooth. 
As for manifolds, a collection of charts covering $X$ is called an atlas. 
An atlas is maximal if it contains all compatible charts. 

\begin{definition}
A \underline{manifold with corners} of codimension $2$ is a second countable Hausdorff space $M$ together with a maximal atlas.
\end{definition}    

\begin{remark}
Closed manifolds and manifolds with boundary are in particular manifolds with corners.
\end{remark}

We define the tangent space $T_x M$ at a point $x\in M$ as the space of derivatives on the real-valued functions $C^\infty(M)$ at $x$. We define embeddings in the same way as for manifolds without corners. We are now able to introduce an essential concept used throughout the main text.

\begin{definition}\label{def:collar}
A \underline{collar} for a submanifold $Y\subset \partial M$ is a diffeomorphism $\varphi \colon U_Y \rightarrow Y\times [0,\epsilon)$ for some fixed $\epsilon>0$ and a neighbourhood $U_Y$ of $Y$.  
\end{definition}

Note that there are situations in which no collars exist.

Given $x\in M$ we define the index of $x$ to be the number of coordinates of $\big(\mathrm{pr}_{\R_{\geq0}^2}\circ \varphi \big) (x)$ equal to $0$ for a chart $\varphi$. Clearly $\mathrm{index}(x)\in\{ 0,1,2\}$, and this definition does not depend on the choice of chart. 
The corners of $M$ are the collection of all points of index~$2$. A connected face of $M$ is the closure of a maximal connected subset of points of index~$1$. A manifold with corners is a manifold with faces if each $x\in M$ belongs to exactly $\mathrm{index}(x)$ connected faces.
In this case we define a face of $M$ to be a disjoint union of connected faces,  which is a manifold with  boundary.  
A boundary defining function for a face $H_i$ is a function $\rho_i \in C^\infty(M)$ such that
$\rho_i(x)\geq 0$ and $\rho_i(x)=0$ if and only if $x\in H_i$.

\begin{definition}
A \underline{$\langle 2 \rangle$-manifold} is a manifold $M$ with faces together with two faces $\partial_0 M$ and $\partial_1 M$ such that $\partial M = \partial_0 M \cup \partial_1 M$ and $\partial_0 M \cap \partial_1 M$ are the corners of $M$. 
\end{definition}  
Denote by $\mathsf{[1]}$ the category corresponding to the ordered set $\{0,1\}$. A $\langle 2 \rangle$-manifold $M$ then defines a diagram $M \colon \mathsf{[1]}^2 \rightarrow \mathsf{Man}_{\mathrm{c}} $ of shape $\mathsf{[1]}^2$ in the category $\mathsf{Man}_{\mathrm{c}} $ of manifolds with corners and smooth embeddings:
\[
\begin{tikzcd}
& M & \\
\partial_0 M \ar[ru] & & \partial_1 M \ar[lu] \\
& \partial_0 M\cap \partial_1 M \ar[ru] \ar[lu] & 
\end{tikzcd}
\]

\subsection{Gluing principal bundles}\label{Section: Gluing}

Given principal bundles with connections on manifolds $M_1$ and $M_2$, gluing them along a common boundary $\Sigma$ requires some care. Naively one could try to cover the glued manifold $M$ by $\{M_1, M_2\}$ and use the descent property for the stack of principal bundles with connection. However, this does not work as $M_1$ and $M_2$ are not open subsets of $M$. A way out is to deform $M_1$ and $M_2$ into open subsets of $M$ by cutting out a collar near their common boundary. But there is no canonical choice for such a collar. We need in this case a third open set interpolating between the two manifolds. In general there is no canonical choice for such an interpolation, whence we should consider it together with the collar as part of the gluing data.\footnote{For topological stacks all choices are equivalent.}

We give an explicit construction for this rather complicated gluing procedure of principal bundles with connections over smooth manifolds and specify the additional information needed. We fix principal bundles with connection $\pi_1 \colon P_1 \rightarrow M_1$ and $\pi_2 \colon P_2 \rightarrow M_2$ over oriented smooth $d$-dimensional compact manifolds $M_1$ and $M_2$. We assume that there are neighbourhoods $M_{1,+}\subset M_1$ and $M_{2,-}\subset M_2$ of parts of the boundaries, and require orientation-preserving diffeomorphisms $\varphi_1 \colon M_{1,+}\rightarrow \Sigma \times (-\epsilon_1,0]$ and $\varphi_2 \colon M_{2,-}\rightarrow \Sigma \times [0,\epsilon_2)$ for fixed $\epsilon_i>0$. Then $\varphi_i$ for $i=1,2$ induce projections $p_i \colon M_{i,\pm}\rightarrow \varphi_i^{-1}(\Sigma\times \{0\})$. We further assume that $P_i$ is of product structure over $M_{i,\pm}$, i.e.  $P_i|_{M_{i,\pm}}=p_i^\ast P_i|_{\varphi_i^{-1}(\Sigma\times \{0\})}$. We fix a third bundle with connection $\pi'_3\colon P'_3 \rightarrow \Sigma$, defining a bundle $\pi_3 \colon P_3 \rightarrow \Sigma \times (-\epsilon_1,\epsilon_2)$ via pullback along the projection onto $\Sigma \times \{0\}$. We choose connection-preserving gauge transformations $\psi'_i: P'_3\rightarrow (\varphi_i^ {-1})^{\ast}(P_i|_{\varphi_i^{-1}(\Sigma\times \{0\})})$. These gauge transformations induce gauge transformations over that part of $\Sigma \times (-\epsilon_1,\epsilon_2)$ where both bundles are defined, which are constant along the fibres. 

We now glue $M_1$ and $M_2$ along $\Sigma$ as usual to get a manifold 
\begin{align*}
 M=M_1 \sqcup_{\varphi_2^{-1}|_{\Sigma}\circ \varphi_1|_{\varphi_1^{-1}(\Sigma\times \{0\})}}M_2 \ , 
\end{align*}
where the collars $M_{1,+}$ and $M_{2,-}$ are needed to define a unique smooth structure on $M$. We cover $M$ with the three open sets
\begin{align}
U_1 &= M_1 \setminus \varphi_1^{-1}\big(\Sigma\times [-\tfrac{\epsilon_1}{4},0] \big)  \ , \nonumber \\[4pt]
U_2 &= M_2 \setminus \varphi_2^{-1}\big(\Sigma\times [0,\tfrac{\epsilon_2}{4}] \big) \ , \nonumber \\[4pt]
U_3 &= \big(\varphi_1^{-1}\sqcup \varphi_2^{-1} \big)\big( \Sigma \times (-\tfrac{\epsilon_1}{2},\tfrac{\epsilon_2}{2})\big) \ .
\label{Definition Cover Gluing}\end{align} 
The structures fixed so far induce connection-preserving gauge transformations 
$\psi_1 \colon P_1|_{U_1\cap U_3}\rightarrow (\varphi_1)^\ast (P_3)|_{U_1\cap U_3} $ and $\psi_2 \colon P_2|_{U_2\cap U_3}\rightarrow (\varphi_2)^\ast (P_3)|_{U_2\cap U_3} $. We use the descent property of the stack of prinicipal bundles to define a bundle over $M$. More concretely we define 
\begin{align*}
P_M= P_1\big|_{U_1} \sqcup P_2\big|_{U_2} \sqcup (\varphi_1\sqcup \varphi_2)^\ast (P_3) \ \big/ \sim \ , 
\end{align*}
where $p\sim \varphi_1(p)$ for all $p\in U_1\cap U_3$ and $p\sim \varphi_2(p)$ for all $p\in U_2\cap U_3$. Using local trivialisations for all bundles involved, or the fact that principal $G$-bundles with connection form a stack $\mathsf{Bun}^\nabla _G$, it is easy to see that $P_M$ is a principal bundle with connection over $M$. This gluing construction depends on the choice of collars $\varphi_i$, the definition of $U_i$ and the trivialisations $\psi_i$. Different choices for the collars and the open cover lead to isomorphic bundles with connection, since we assume that the bundles are of product form on the collars. 

We give another point of view on this construction using `mapping cylinders'. Given two principal bundles with connection $\pi_1 \colon P_1 \rightarrow \Sigma_1$ and $\pi_2 \colon P_2 \rightarrow \Sigma_2$, a diffeomorphism $f\colon \Sigma_1 \rightarrow \Sigma_2$, and a connection-preserving gauge transformation $\psi \colon P_1 \rightarrow f^\ast P_2$, the mapping cylinder of length $\epsilon $ is defined as 
\begin{align*}
\mathfrak{M}(f,\psi):= \Sigma_1\times \big[0, \tfrac{3\epsilon}{4}\big) \sqcup \Sigma_2 \times \big(\tfrac{\epsilon}{4},\epsilon\big] \ \big/ \sim \ ,
\end{align*} 
where $(x,t)\sim (f(x),t)$ for all $(x,t) \in \Sigma_1\times (\tfrac{\epsilon}{4}, \tfrac{3\epsilon}{4})$, together with the principal bundle with connection over $\mathfrak{M}(f,\psi)$ given by
\begin{align*}
P_{\mathfrak{M}(f,\psi)}:= P_1\times \big[0, \tfrac{3\epsilon}{4}\big) \sqcup P_2 \times \big(\tfrac{\epsilon}{4},\epsilon\big] \ \big/ \sim \ ,
\end{align*}  
where $(p,t)\sim \big((f^{-1})^\ast \psi(p),t\big)$ for all $(p,t)\in P_1\times (\tfrac{\epsilon}{4}, \tfrac{3\epsilon}{4})$. Now we see that the gluing happens by removing half of the collars $M_{1,+}$ and $M_{2,-}$, and attaching mapping cylinders $\mathfrak{M}(\varphi_1,\psi'_1)$ and $\mathfrak{M}(\varphi^{-1}_2,\psi_2'{}^{-1})$ of appropriate length. This point of view is crucial in the proof of Theorem~\ref{A is field theory}. 

The important properties we used in this construction are the stack property of $\mathsf{Bun}^\nabla _G$ and the notion of a bundle of product structure. In Section~\ref{Appendix Geometric bicategories} we use a similar construction to build a bicategory of cobordisms equipped with elements of an arbitrary stack $\mathscr{F}$.

We also need to glue metrics, which can be done using the open cover \eqref{Definition Cover Gluing}. For this, we assume that there are metrics $g_1\in \Gamma (\text{Sym}^ 2(T^\ast M_1))$, $g_2\in \Gamma (\text{Sym}^ 2(T^\ast M_2))$ and $g'_3\in \Gamma (\text{Sym}^ 2(T^\ast \Sigma))$. We equip $\Sigma \times (-\epsilon_1 ,\epsilon_2)$ with the metric $g_3=g_3'+\mathrm{d} t\otimes \mathrm{d} t$. Now it is sensible to assume that $\varphi_1$ and $\varphi_2$ are isometries. We can define a metric over $U_3$ as $(\varphi_1\sqcup \varphi_2)^\ast (g_3)|_{U_3}$. Then all metrics agree on the intersections, since we assumed that $\varphi_1$ and $ \varphi_2$ are isometries. This defines a metric on $M$, since sections of $\text{Sym}^ 2(T^\ast M)$ form a sheaf over $M$. 

\subsection{Dirac operators on spin manifolds with boundary}

Given a spin structure on a $d$-dimensional oriented Riemannian manifold $M$ with boundary $\partial M$, i.e. a double cover of the frame bundle $P_{SO(d)}(M)$ by a principal $Spin(d)$-bundle $P_{Spin(d)}(M) \rightarrow M$, we can include the frame bundle $P_{SO(d-1)}(\partial M )$ into $P_{SO(d)}(M)$ by adding the inward pointing normal vector to an orthonormal frame of $\partial M$. The pullback of the double cover $P_{Spin(d)}(M)$ along this inclusion gives a spin structure on $\partial M$. 

Assume from now on that all structures are of product form near the boundary. 
To describe the relation between the Dirac operator on the boundary and on the bulk manifold 
we use the embedding of Clifford bundles 
\begin{align*}
\text{C}\ell_{d-1}(\partial M)\longrightarrow \text{C}\ell_{d}(M) \ , \quad 
T_x (\partial M) \ni v \longmapsto  v \, n_x \ ,  
\end{align*}         
where $n$ is the inward pointing normal vector field corresponding to the boundary. This gives $S_M\big|_{\partial M}$ the structure of a Clifford bundle over $\partial M$.
For the relation to the spinor bundle over the boundary we need to distinguish between even and odd dimensions.

If the dimension $d$ of $M$ is odd then we can identify $S_M\big|_{\partial M}$ with the spinor bundle over $\partial M$.
In this case the Dirac operator can be described in a neighbourhood of $\partial M$ by
\begin{align}
\label{Equation: Dirac operator on a product}
\D_M = n\cdot \big(\D_{\partial M} +\partial_n\big) \ .
\end{align} 
 
On the other hand, if the dimension $d$ of $M$ is even then the spinor bundle $S_M=S^+_M\oplus S^-_M$ decomposes into spinors of positive and negative chirality. The Clifford action of $\text{C}\ell_{d-1}(\partial M)$ leaves this decomposition invariant and we can identify the spinor bundle over $\partial M$ with the pullback of the positive spinor bundle $S^+_M\big|_{\partial M}$. As the Clifford action of $\text{C}\ell_{d-1}(\partial M)$ commutes with the chirality operator $\Gamma$, an identification with the negative spinor bundle is possible as well. Near the boundary the Dirac operator is given by
\begin{align*}
\D_M = n\cdot\begin{pmatrix}
\D_{\partial M} +\partial_n & 0 \\ 
0 &  \Gamma|_{S^+_M} \, \D_{\partial M} \, \Gamma|_{S^-_M} +\partial_n
\end{pmatrix} \ .
\end{align*}

\subsection{b-geometry}\label{Appendix b-geometry}

b-geometry (for `boundary geometry') is concerned with the study of geometric structures on manifolds with corners which can be singular at the boundary. 
We fix a $d$-dimensional $\langle 2\rangle$-manifold $M$ and an ordering of its hypersurfaces $\{H_1, \dots, H_k\}$.
The central objects in b-geometry are b-vector fields. These are vector fields which are tangent to all boundary hypersurfaces. 
We denote by $\mathrm{Vect}_{\mathrm{b}}(M)$ the projective $C^\infty (M)$-module of b-vector fields. Then $\mathrm{Vect}_{\mathrm{b}}(M)$ is closed under the Lie bracket of vector fields.  
By the Serre-Swan theorem, the b-vector fields are naturally sections of the b-tangent bundle with fibres
\begin{align*}
^{\mathrm{b}}T_x M := \mathrm{Vect}_{\mathrm{b}}(M) \setminus \mathcal{I}_x(M) \cdot \mathrm{Vect}_{\mathrm{b}}(M) \ , 
\end{align*}
where $\mathcal{I}_x(M)= \{f\in {C}^\infty(M) \mid f(x)=0 \}$ is the ideal of functions vanishing at $x\in M$. This allows us to define arbitrary b-tensors as in classical differential geometry. The inclusion $\mathrm{Vect}_{\mathrm{b}}(M) \hookrightarrow \mathrm{Vect}(M)$ induces a natural vector bundle map $\alpha_{\mathrm{b}} \colon ^{\mathrm{b}}TM \rightarrow TM$.

The structures introduced so far can be summarized by saying that $(^{\mathrm{b}}TM,\alpha_{\mathrm{b}})$ is a boundary tangential Lie algebroid. The b-tangent bundle is isomorphic to the tangent bundle in the interior of $M$, and $(M, \mathrm{Vect}_{\mathrm{b}}(M))$ is an example of a manifold with Lie structure at infinity~\cite{LieManifolds}.
   
Using a set of boundary defining functions $x_i$, the Lie algebra $\mathrm{Vect}_{\mathrm{b}}(M)$ is locally spanned near a point $x\in H_i$ of index~$1$ by $\{ x_i\, \partial_{x_i}, \partial_{h_1}, \dots , \partial_{h_{d-1}}\}$, where $\{h_l\}_{l=1}^{d-1}$ is a local coordinate system for $H_i$. 
In a neighbourhood of $x\in H_i\cap H_j$, $i\neq j$, we can form a basis given by $\{x_i\, \partial_{x_i},x_j\, \partial_{x_j},\partial_{y_1}, \dots , \partial_{y_{d-2}} \}$, where $\{y_l\}_{l=1}^{d-2}$ is a local coordinate system on $Y_{ij}=H_i\cap H_j$. 
The dual basis for the b-cotangent bundle $^{\mathrm{b}}T^*M$ is denoted by $\big\{ \tfrac{\diff x_i}{x_i},\tfrac{\diff x_j}{x_j}, \diff y_1 ,\dots , \diff y_{d-2} \big\}$. 

A b-metric $g$ is now simply a metric on the vector bundle $^{\mathrm{b}} TM$ over $M$. This defines an ordinary metric in the interior of $M$. The general expression in local coordinates near a corner point is
\begin{align*}
g= \sum_{i,j=0,1}\, a_{ij} \ \frac{\diff x_i}{x_i} \otimes \frac{\diff x_j}{x_j} + 2 \, \sum_{i=0,1} \ \sum_{j=1}^{d-2}\, b_{ij} \ \frac{\diff x_i}{x_i} \otimes \diff y_j + \sum_{i,j=1}^{d-2}\, c_{ij} \ \diff y_i \otimes \diff y_j \ .
\end{align*}
A b-metric $g$ is exact if there exists a set of boundary defining functions $x_i$ such that it takes the form
\begin{align*}
g= \begin{cases}
\displaystyle \frac{\diff x_i}{x_i} \otimes \frac{\diff x_i}{x_i} +h_{H_i} & \text{ near } H_i \ , \\[10pt]
\displaystyle  \frac{\diff x_i}{x_i} \otimes \frac{\diff x_i}{x_i}+ \frac{\diff x_j}{x_j} \otimes \frac{\diff x_j}{x_j} + h_{H_i\cap H_j} & \text{ near } H_i\cap H_j \ ,
\end{cases}
\end{align*}
where $h_Y$ denotes a metric on $Y$.

We will now describe the relation between $\langle 2\rangle$-manifolds with exact b-metrics and the index theory on manifolds with corners considered in the main text. To define the index we attach infinite cylindrical ends $H_i\times (-\infty,0]$ to the boundary hypersurfaces and $Y_{ij}\times (-\infty,0]^2$ to the corners. The coordinate transformation $x_i= \e^{t_i}$ for $t_i\in(-\infty,0]$ maps this non-compact manifold to the interior of a manifold $X_i$ with corners. The product metric on the cylindrical ends induces a b-metric on $X_i$, since $\diff t_i\otimes\diff t_i = \tfrac{\diff x_i}{x_i} \otimes \tfrac{\diff x_i}{x_i}$. For this reason one can view the study of manifolds with exact b-metrics as the study of manifolds with cylindrical ends. 

A b-differential operator is an element of the universal enveloping algebra of $\mathrm{Vect}_{\mathrm{b}}(M)$, the collection of which act naturally on $C^\infty (M)$. 
A b-differential operator $D\in \text{Diff}^{\,k}_{\mathrm{b}}(M,E_1,E_2)$ of order $k$ between two vector bundles $E_1$ and $E_2$ over $M$ is a smooth fibre-preserving map, which in any local trivialisations of $E_1$ and $E_2$ is given by a matrix of linear combinations of products of up to $k$ b-vector fields. 
Most concepts from differential geometry such as connections, symbols and characteristic classes can be generalized to the b-geometry setting. 

Since exact b-metrics are singular at the boundary it is necessary to define a renormalised b-integral. Heuristically, the problem stems from the fact that the integral $\int_0^1\, \tfrac{\diff x}{x}$ is divergent. The cure for this is to multiply with $x^z$ for $\text{Re}(z)> 0$.

\begin{lemma}(\cite[Lemma 4.1]{Loyaindex}) \ 
Let $M$ be a manifold with corners and an exact b-metric~$g$. Then for all $f\in {C}^\infty(M)$ and $z\in \C$ with $\text{Re}(z)> 0$, the integral
\begin{align*}
F(f,z):=\int_M\, x^z\, f\ \diff g
\end{align*}
exists and extends to a meromorphic function $F(f,z)$ of $z\in\C$.
\end{lemma}

\begin{definition}
Let $M$ be a manifold with corners and an exact b-metric $g$. The \underline{b-integral} of a function $f\in C^\infty(M)$ is 
\begin{align}
\label{Definition b-integral}
{{}^{\mathrm{b}}}\!\!\int_M\, f\ \diff g = \text{Reg}_{z=0} \ F(f,z)\ .
\end{align}
\end{definition}

This allows us to define the b-trace of a pseudo-differential operator $D$ in terms of its kernel $D(x,y)$ as
\begin{align*}
^{\mathrm{b}}\text{Tr}(D)= {{}^{\mathrm{b}}}\!\!\int_M\, \text{tr} \big(D(x,x)\big)\ \diff g(x) \ , 
\end{align*}
where the trace $\mathrm{tr}$ is over the fibres of the vector bundle on which $D$ acts.

\section{Bicategories}\label{Appendix bicategories}

As it is central to the treatment of this paper, in this appendix we provide a fairly detailed account of symmetric monoidal bicategories, following~\cite{leinster:1998,schommer2011classification} for the most part.

\subsection{Basic definitions}

We introduce the basic concepts from the theory of bicategories following \cite{leinster:1998}.

\begin{definition}
\label{DefinitionBicategory}
A \underline{bicategory} $\mathscr{B}$ consists of the following data:
\begin{itemize}
\item[(a)] A class $\text{Obj}(\mathscr{B})$ of objects.

\item[(b)] A category $\mathsf{Hom}_\Bscr(A,B)$ for all $A,B \in \text{Obj}(\mathscr{B})$, whose objects $f:A\to B$ we call 1-morphisms and whose morphisms $f\Rightarrow g$ we call 2-morphisms.

\item[(c)] Composition functors \[\circ_{ABC} \colon \mathsf{Hom}_\Bscr (B,C)\times \mathsf{Hom}_\Bscr(A,B)\longrightarrow \mathsf{Hom}_\Bscr(A,C)\] for all $A,B,C\in \text{Obj}(\mathscr{B})$.

\item[(d)] Identity functors \[\mathsf{Id}_A \colon \mathsf{1} = \star \, \big/\!\!\big/ \, \{\id_\star\}\longrightarrow \mathsf{Hom}_\Bscr(A,A)\] for all $A\in \text{Obj}(\mathscr{B})$.

\item[(e)] Natural associator isomorphisms \[\mathsf{a}_{A,B,C,D} \colon \circ_{ACD} \circ \big(\id_{\mathsf{Hom}_\Bscr(C,D)} \times \circ_{ABC} \big) \Longrightarrow \circ_{ABD} \circ \big(\circ_{BCD} \times \id_{\mathsf{Hom}_\Bscr(A,B)}\big)\] for all $A,B,C,D\in \text{Obj}(\mathscr{B})$, expressing associativity of the composition. 

\item[(f)] Natural right and left unitor isomorphisms 
\[ \mathsf{r}_A \colon \circ_{AAB} \circ \big(\id_{\mathsf{Hom}_\Bscr(A,B)}\times \mathsf{Id}_A\big)\Longrightarrow \id_{\mathsf{Hom}_\Bscr(A,B)}\] and 
\[\mathsf{l}_A \colon \circ_{AAB} \circ \big(\mathsf{Id}_B \times \id_{\mathsf{Hom}_\Bscr(A,B)}\big)\Longrightarrow \id_{\mathsf{Hom}_\Bscr(A,B)} \]
for all $A,B\in \text{Obj}(\mathscr{B})$.
\end{itemize}  
These data are required to satisfy the following coherence axioms:
\begin{itemize}
\item[(C1)] The pentagon diagram
\small
\[
\begin{tikzcd}
& \big((k\!\circ\! h)\!\circ\! g\big)\!\circ\! f\arrow[dl, swap, Rightarrow, "\mathsf{a}"] \arrow[rr,Rightarrow,"\mathsf{a}\bullet \id"]& &\big(k\!\circ\!(h\!\circ\! g)\big)\!\circ\! f \arrow[rd,Rightarrow,"\mathsf{a}"] & \\
(k\!\circ\! h)\!\circ\!(g\!\circ\! f)\arrow[rrd,swap,Rightarrow,"\mathsf{a}"]& & & &\arrow[lld,Rightarrow,"\id\bullet \mathsf{a}"]k\!\circ\!\big((h\!\circ\! g)\!\circ\! f\big) \\
& &k\!\circ\!\big(h\!\circ\!(g\!\circ\! f)\big) & &
\end{tikzcd} \]
\normalsize
commutes for all composable 1-morphisms $k$, $h$, $g$ and $f$, where $\bullet$ denotes the horizontal composition of natural transformations.

\item[(C2)] The triangle diagram 
\[
\begin{tikzcd}
(g\circ \mathsf{Id})\circ f\arrow[rr,Rightarrow,"\mathsf{a}"] \arrow[rd,swap,Rightarrow,"\mathsf{r}\bullet\id"]& &g\circ (\mathsf{Id}\circ f)\arrow[dl,Rightarrow,"\id\bullet\mathsf{l}"]\\
& g\circ f &
\end{tikzcd} 
\]
commutes for all composable 1-morphisms $f$ and $g$.
\end{itemize}
\end{definition}  

There are different definitions for functors between bicategories corresponding to different levels of strictness. For our purposes the following definition is suitable.
\begin{definition}
\label{Definition Morphism Bicategory}
A \underline{2-functor} $\mathcal{F} \colon \mathscr{B}\rightarrow \mathscr{B}'$ between two bicategories $\Bscr$ and $\Bscr'$ consists of the following data:
\begin{itemize}
\item[(a)] A map $\mathcal{F} \colon \text{Obj}(\mathscr{B})\rightarrow \text{Obj}(\mathscr{B}'\,)$.

\item[(b)] A functor $\mathcal{F}_{AB} \colon \mathsf{Hom}_\Bscr(A,B)\rightarrow \mathsf{Hom}_{\Bscr'}\big(\Fa(A),\Fa(B)\big)$ for all $A,B\in \text{Obj}(\mathscr{B})$.

\item[(c)] A natural isomorphism $\Phi_{ABC}$ given by
\[ \begin{tikzcd} 
\mathsf{Hom}_\Bscr(B,C)\times \mathsf{Hom}_\Bscr(A,B) \arrow{r}{\circ }\arrow[d,swap,"{\mathcal{F}_{BC}\times \mathcal{F}_{AB}}"] & \mathsf{Hom}_\Bscr(A,C) \arrow{d}{\mathcal{F}_{AC}}\\
\mathsf{Hom}_{\Bscr'}\big(\mathcal{F}(B),\mathcal{F}(C)\big)\times \mathsf{Hom}_{\Bscr'}\big(\mathcal{F}(A),\mathcal{F}(B)\big)\ar[ru, Rightarrow, shorten <= 2ex, shorten >= 2ex, "\Phi_{ABC}"] \arrow[r,swap,"\circ'"] & \mathsf{Hom}_{\Bscr'}\big(\mathcal{F}(A),\mathcal{F}(C)\big)
\end{tikzcd} \]
for all $A,B,C\in \text{Obj}(\mathscr{B})$.

\item[(d)] A natural isomorphism $\Phi_A$ given by
\[\begin{tikzcd} \mathsf{1} \arrow[d,swap,"\id"] \arrow{rr}{\mathsf{Id}_A} & & \mathsf{Hom}_{\Bscr}(A,A)\arrow{d}{\mathcal{F}_{AA}} \\ \mathsf{1} \arrow[rr,swap, "\mathsf{Id}'_{\mathcal{F}(A)}"] \ar[rru, shorten <= 2ex, shorten >= 2ex, Rightarrow, "\Phi_A", pos=0.7] & & \mathsf{Hom}_{\Bscr'}\big(\mathcal{F}(A),\mathcal{F}(A)\big)
\end{tikzcd}\] for all $A\in \text{Obj}(\mathscr{B})$.
\end{itemize} 
These data are required to satisfy the following coherence axioms:
\begin{itemize}
\item[(C1)] The diagram
\begin{equation}
\label{EQ1: Definition 2Functor}
\begin{tikzcd} 
\big(\mathcal{F}(h) \circ' \mathcal{F}(g)\big)\circ' \mathcal{F}(f) \arrow[d,swap,Rightarrow,"\mathsf{a}'"] \arrow[r,Rightarrow,"\Phi \bullet' \id"] & \mathcal{F}(h\circ g) \circ' \mathcal{F}(f) \arrow[r,Rightarrow,"\Phi"] & \mathcal{F}\big((h\circ g)\circ f\big)\arrow[d,Rightarrow,"\mathcal{F}(\mathsf{a})"] \\ \mathcal{F}(h)\circ' \big(\Fa(g)\circ' \Fa(f)\big) \arrow[r,swap,Rightarrow,"\id \bullet' \Phi"] & \Fa(h) \circ' \Fa(g\circ f) \arrow[r,,swap, Rightarrow,"\Phi"] & \Fa\big(h\circ(g\circ f)\big)
\end{tikzcd}
\end{equation}
commutes for all composable 1-morphisms.

\item[(C2)] The diagram
\begin{equation}
\label{EQ2: Definition 2Functor}
\begin{tikzcd}
\mathcal{F}(f) \circ' \mathsf{Id}'_{\mathcal{F}(A)} \arrow[r,Rightarrow,"\id\bullet' \Phi"] \arrow[dr,swap,Rightarrow,"\mathsf{r}'"] & \mathcal{F}(f) \circ' \mathcal{F}(\mathsf{Id}_A)\arrow[r,Rightarrow,"\Phi"] & \mathcal{F}(f\circ \mathsf{Id}_A) \arrow[dl,Rightarrow,"\mathcal{F} (\mathsf{r})"] \\ & \mathcal{F}(f) &
\end{tikzcd} 
\end{equation} 
commutes for all composable 1-morphisms.

\item[(C3)] A diagram analogous to \eqref{EQ2: Definition 2Functor} for the left unitors $\mathsf{l}$ and $\mathsf{l}'$ commutes.
\end{itemize}
\end{definition} 

Again there are different ways to define natural transformations between 2-functors. The following definition is suitable for our purposes.
\begin{definition}
\label{Definition transformation Bicategory}
Given two 2-functors $\mathcal{F},\mathcal{G}  \colon \mathscr{B}\rightarrow \mathscr{B'} $, a \underline{natural 2-transformation} $\sigma \colon \mathcal{F}\Rightarrow \mathcal{G}$ consists of the following data: 
\begin{itemize}
\item[(a)] A 1-morphism $\sigma_A \colon \mathcal{F}(A)\rightarrow \mathcal{G}(A)$ for all $A \in \mathrm{Obj}(\mathscr{B})$.

\item[(b)] A natural transformation $\sigma_{AB}$ given by\footnote{Here we use $\ast$ to denote pullbacks and pushforwards in the usual way.}
\[
 \begin{tikzcd} 
\Hom_{\mathscr{B}}(A,B) \arrow{r}{\mathcal{F}_{AB} }\arrow[d,swap,"\mathcal{G}_{AB}"] & \Hom_{\mathscr{B}'}\big(\mathcal{F}(A),\mathcal{F}(B)\big)\arrow{d}{\sigma_{B\ast}}\\
\Hom_{\mathscr{B}'}\big(\mathcal{G}(A),\mathcal{G}(B)\big)\arrow[Rightarrow]{ru}{\sigma_{AB}} \arrow[r,swap,"\sigma_A^\ast"] & \Hom_{\mathscr{B}'}\big(\mathcal{F}(A),\mathcal{G}(B)\big)
\end{tikzcd}
\]
for all $A,B\in \mathrm{Obj}(\mathscr{B})$. In particular, these natural transformations comprise families of 2-morphisms $\sigma_f \colon \mathcal{G}_{AB}(f) \circ' \sigma_A \Rightarrow \sigma_{B}\circ' \mathcal{F}_{AB}(f)$ for all 1-morphisms $f:A\to B$ in $\Bscr$.
\end{itemize}
These data are required to satisfy the following coherence axioms:
\begin{itemize}
\item[(C1)] The diagram 
\begin{equation}
\label{Equation1: Definition Transformation}
\begin{tikzcd}
\big(\mathcal{G}(g) \circ' \mathcal{G}(f)\big)\circ' \sigma_A \arrow[r,Rightarrow,"\mathsf{a}'"] \arrow[d,swap,Rightarrow,"\Phi_\Ga\bullet'\id"] & \mathcal{G}(g) \circ' \big(\mathcal{G}(f)\circ' \sigma_A\big)\arrow[r,Rightarrow, "\id\bullet' \sigma_f"] & \mathcal{G}(g)\circ' \big(\sigma_{B}\circ' \mathcal{F}(f)\big)\arrow[d,Rightarrow,"\mathsf{a}'"] \\
\mathcal{G}(g\circ f)\circ' \sigma_A \arrow[d,swap,Rightarrow,"\sigma_{g\circ f}"] & & \big(\mathcal{G}(g)\circ' \sigma_{B}\big)\circ' \mathcal{F}(f) \arrow[d,Rightarrow,"\sigma_g\bullet'\id"] \\
\sigma_{C}\circ' \mathcal{F}(g\circ f) \arrow[r,swap,Leftarrow,"\id\bullet'\Phi_\Fa"] & \sigma_{C}\circ' \big(\mathcal{F}(g)\circ' \mathcal{F}(f)\big) \arrow[r,swap,Leftarrow,"\mathsf{a}'"] & \big(\sigma_{C}\circ' \mathcal{F}(g)\big)\circ' \mathcal{F}(f)
\end{tikzcd}
\end{equation}
commutes for all 1-morphisms $f \colon A\rightarrow B$ and $g \colon  B\rightarrow C$ in $\mathscr{B}$.

\item[(C2)] The diagram
\begin{equation}
\label{Equation2: Definition Transformation}
\begin{tikzcd}
\mathsf{Id}'_{\mathcal{G}(A)}\circ' \sigma_A\arrow[r,Rightarrow,"\mathsf{l}'"] \arrow[d,swap,Rightarrow,"\Phi_\Ga\bullet'\id"]&\sigma_A\arrow[r,Rightarrow,
"\mathsf{r}'{}^{-1}"] & \sigma_A \circ' \mathsf{Id}'_{\Fa(A)}\arrow[d,Rightarrow,"\id\bullet'\Phi_\Fa"] \\
\mathcal{G}(\mathsf{Id}_A)\circ' \sigma_A \arrow[rr,swap,Rightarrow,"\sigma_{\mathsf{Id}_A}"] & & \sigma_A \circ' \mathcal{F}(\mathsf{Id}_A) 
\end{tikzcd}
\end{equation}
commutes for all $A\in \mathrm{Obj}(\mathscr{B})$.
\end{itemize}
\end{definition}
Note that we do not require the natural transformation $\sigma_{AB}$ to be invertible, whence its direction matters. There is an alternative definition using the opposite direction. 
  
Since there is an additional layer of structure for bicategories, we are able to relate two natural 2-transformations to each other. There is only one way to do this, since there are no higher morphisms.
\begin{definition}
\label{Definition Modification Bicategory}
Given two natural 2-transformations $\sigma,\tau \colon \mathcal{F}\Rightarrow \mathcal{G}$, a \underline{modification} $\mit\Gamma \colon \sigma \Rrightarrow \tau$ consists of a 2-morphism ${\mit\Gamma}\!_A \colon \sigma_A \Rightarrow \tau_A$ for each $A\in \mathrm{Obj}(\mathscr{B})$ such that the diagram
\begin{equation}
\label{EQ: Modification}
\begin{tikzcd}
\mathcal{G}(f)\circ' \sigma_A \arrow[rr,Rightarrow,"\id\bullet' {\mit\Gamma}\!_A"] \arrow[d,swap,Rightarrow,"\sigma_f"] & & \mathcal{G}(f)\circ' \tau_A \arrow[d,Rightarrow,"\tau_f"] \\
\sigma_B\circ' \mathcal{F}(f)\arrow[rr,swap,Rightarrow,"{\mit\Gamma}\!_{B}\bullet'\id"] & & \tau_B\circ' \mathcal{F}(f)
\end{tikzcd}
\end{equation} 
commutes for all 1-morphisms $f \colon A\rightarrow B$ in $\Bscr$.
\end{definition}

\subsection{Symmetric monoidal bicategories}

We now describe how to introduce symmetric monoidal structures on bicategories following~\cite{schommer2011classification}.

\begin{definition}
A \underline{symmetric monoidal bicategory} consists of a bicategory $\mathscr{B}$ together with the following data:
\begin{itemize}
\item[(a)] A monoidal unit $1\in \mathrm{Obj}(\mathscr{B})$.

\item[(b)] A 2-functor $\otimes \colon \mathscr{B} \times \mathscr{B} \rightarrow \mathscr{B}$.

\item[(c)] Equivalence natural 2-transformations\footnote{Here `equivalence' means the natural 2-transformations in question have weak inverses.} $\alpha \colon \otimes \circ (\id \times \otimes) \Rightarrow \otimes \circ (\otimes \times \id),$ $\lambda \colon 1\otimes \, \cdot \, \Rightarrow \id$ and $\rho \colon \id \Rightarrow \, \cdot \, \otimes 1$.
We pick adjoint inverses which are part of the data and denoted them
by $^\star$, leaving the adjunction data implicit, and for every
equivalence natural 2-transformation we pick an adjoint weak inverse without writing them out explicitly.

\item[(d)] An equivalence natural 2-transformation $\beta \colon a\otimes b \Rightarrow b\otimes a$.

\item[(e)] The four invertible modifications
\[
\begin{tikzcd}
 &  \otimes \circ (\otimes \times \otimes)\arrow[dr,Rightarrow,"\alpha"] & \\
\otimes \circ (\otimes \times \id) \circ (\otimes \times \id \times \id) \arrow[ur,Rightarrow,"\alpha"] \arrow[d,swap,Rightarrow,"\alpha\otimes\id"] &  & \otimes \circ (\id \times \otimes) \circ (\id \times \id \times \otimes)  \\
 \otimes \circ (\otimes \times \id) \circ (\id \times \otimes \times \id) \arrow[rr,swap,Rightarrow,"\alpha"] & 
\tarrow[swap,shorten <=15pt,shorten >=5pt," \ \mit\Xi "]{uu}  
  & \otimes \circ (\id \times \otimes) \circ (\id \times \otimes \times \id) \ar[u,swap,Rightarrow,"\id\otimes\alpha"]
\end{tikzcd}
\]

\[
\begin{tikzcd}
\otimes \circ \big(\id \times (1 \otimes \, \cdot \, ) \big) \ar[rr,Rightarrow,"\alpha"] & \tarrow[shorten <=2pt,shorten >=5pt, "\ \mit\Theta"]{d} &\otimes \circ \big(( \, \cdot \, \otimes 1 ) \times \id\big) \ar[d,Leftarrow,"\rho\otimes\id"] 
\\
\otimes \ar[u,Leftarrow,"\id\otimes \lambda"] \ar[rr,swap, Rightarrow, "\id"] & \ & \otimes 
\end{tikzcd}
\]

\[
\begin{tikzcd}
\otimes \circ \big((1 \otimes \, \cdot \, )\times \id\big) \ar[rr,Rightarrow,"\lambda\otimes\id"] \ar[rd,swap,Rightarrow,"\alpha"] & \tarrow["\ \mit\Lambda",shorten <=2pt,shorten >=2pt]{d} &\otimes \\
& (1 \otimes \, \cdot \, ) \circ (\id \times \otimes) \ar[ru,swap,Rightarrow,"\lambda"] & 
\end{tikzcd}
\]
and
\[
\begin{tikzcd}
\otimes \ar[rr,Rightarrow,"\id\otimes \rho"] \ar[rd,swap,Rightarrow,"\rho"] & \tarrow[ "\ \mit\Psi ",shorten <=2pt,shorten >=2pt]{d} &\otimes \circ \big(\id \times ( \, \cdot \, \otimes 1) \big) \\
& ( \, \cdot \, \otimes 1) \circ (\id \times \otimes) \ar[ru,swap,Rightarrow,"\alpha"] & 
\end{tikzcd}
\]

\item[(f)] Further invertible modifications
\[\begin{tikzcd}
 & a\otimes (b\otimes c) \ar[r, Rightarrow, "\beta" {name=ar1}] & (b \otimes c) \otimes a \arrow[rd,Rightarrow,"\alpha"] & \\ 
(a\otimes b)\otimes c \arrow[ru,Rightarrow,"\alpha"] \arrow[rd,swap, Rightarrow,"\beta\otimes\id"] & & &b\otimes(c \otimes a) \\
 & (b\otimes a) \otimes c\ar[r,swap,Rightarrow, "\alpha" {name=ar2}] &b \otimes (a \otimes c) \arrow[ru,swap,Rightarrow,"\id\otimes\beta"] & 
 \tarrow[shorten <=10pt,shorten >=10pt, from=ar1, to=ar2, "\ R"]{}
\end{tikzcd} \]
and
\[\begin{tikzcd}
 & (a\otimes b)\otimes c \ar[r,Rightarrow, "\beta"{name=ar1}] & c \otimes (a\otimes b) \arrow[rd,Rightarrow,"\alpha"] & \\ 
a\otimes (b\otimes c) \arrow[ru,Rightarrow,"\alpha"] \arrow[rd,swap,Rightarrow,"\alpha\circ(\beta\otimes\id)\circ\alpha"] & & &(c\otimes a) \otimes b  \\
 & b\otimes (a \otimes   c)\ar[r,swap,Rightarrow, "\beta"{name=ar2}] &(a\otimes c) \otimes b \arrow[ru,swap,Rightarrow,"\beta\otimes\id"] & 
\tarrow[shorten <=10pt,shorten >=10pt, from=ar1, to=ar2, "\ S"]{}
\end{tikzcd} \]

\item[(g)] An invertible modification
\[
\begin{tikzcd}
a \otimes b  \ar[rr, Rightarrow, "\id"] \ar[rd,swap,Rightarrow,"\beta"] & \tarrow["\ \mit\Sigma ",shorten <=1pt,shorten >=1pt]{d} &a \otimes b  \\
& b \otimes a \ar[ru,swap,Rightarrow,"\beta"] & 
\end{tikzcd}
\]
\end{itemize}
These data are required to satisfy a long list of coherence diagrams, see~\cite[Appendix~C]{schommer2011classification} for details.
\end{definition}

\begin{example}
\label{Def:2Vect}
A Kapranov-Voevodsky 2-vector space~\cite{KV94} is a $\C$-linear semi-simple additive category $\Vscr$
with finitely many isomorphism classes of simple objects; in particular, a 2-vector space is also an abelian category.  
There is a 2-category $\Tvs$ of 2-vector spaces, $\C$-linear functors and natural transformations. 
Given two 2-vector spaces $\Vscr_1$ and $\Vscr_2$ we can define their tensor product $\Vscr_1 \boxtimes \Vscr_2$~\cite[Definition~1.15]{bakalov2001lectures} to 
be the category with objects given by finite formal sums 
\[
\bigoplus_{i=1}^n\, V_{1i}\boxtimes V_{2i} \ , 
\]
with $V_{1i}\in \mathrm{Obj}(\Vscr_1)$ and $V_{2i}\in \mathrm{Obj}(\Vscr_2)$. The space of morphisms is given by
\[
\mathrm{Hom}_{\Vscr_1\boxtimes \Vscr_2}\Big(\, \mbox{$\bigoplus\limits_{i=1}^n\, V_{1i}\boxtimes V_{2i}\,,\,\bigoplus\limits_{j=1}^m \, V'_{1j}\boxtimes V'_{2j} $} \, \Big)= \bigoplus_{i=1}^n \ \bigoplus_{j=1}^m\, \mathrm{Hom}_{\Vscr_1}(V_{1i},V'_{1j})\otimes_\C \mathrm{Hom}_{\Vscr_2}(V_{2i},V'_{2j}) \ .
\] 
This tensor product coincides with the Deligne product of abelian categories. It furthermore satisfies the universal property with respect to bilinear functors that one would expect from a tensor product.
We can also take tensor products of $\C$-linear functors and of natural transformations.
 Then the 2-category $\Tvs$ with $\boxtimes$ is a symmetric monoidal bicategory with monoidal unit~$1$ given by the category of finite-dimensional vector spaces $\fvs$.  
\end{example}

\begin{definition}
A \underline{symmetric monoidal 2-functor} between two symmetric monoidal bicategories $\mathscr{B}$ and $\mathscr{B}'$ consists of a 2-functor $\Ha \colon \mathscr{B} \rightarrow \mathscr{B}' $ of the underlying bicategories together with the following data:
\begin{itemize}
\item[(a)] Equivalence natural 2-transformations\footnote{We fix again adjoint inverses and the adjunction data.} $\chi \colon \otimes ' \circ \big(\Ha(\, \cdot\, )\times \Ha(\, \cdot\, )\big)\Rightarrow \Ha \circ \otimes $ and $\iota \colon 1' \Rightarrow \Ha(1)$, where here we consider $1$ as a 2-functor from the bicategory with one object, one 1-morphism and one 2-morphism to $\mathscr{B}$.

\item[(b)] The three invertible modifications
\end{itemize}
\small
\[
\begin{tikzcd}
 &\Ha(a)\otimes' \big(\Ha(b)\otimes' \Ha(c)\big) \ar[r,Rightarrow,"\id\otimes'\chi\ "{name=ar1}] & \Ha(a)\otimes' \Ha(b\otimes c)\ar[rd,Rightarrow,"\chi"] & \\
\big(\Ha(a)\otimes' \Ha(b)\big)\otimes' \Ha(c) \ar[ur,Rightarrow,"\alpha'"] \ar[rd,swap,Rightarrow,"\chi\otimes'\id"]& & & \Ha\big(a\otimes (b\otimes c)\big) \\
 & \Ha(a\otimes b)\otimes' \Ha(c) \ar[r,swap,Rightarrow,"\chi"{name=ar2}] & \Ha\big((a\otimes b)\otimes c\big)\ar[ru,swap,Rightarrow,"\Ha(\alpha)"]
\tarrow[shorten <=10pt,shorten >=10pt, from=ar1, to=ar2, "\ \mit\Omega"]{}
\end{tikzcd}
\]
\normalsize

\begin{center}
$
\begin{tikzcd}
\Ha(1)\otimes' \Ha(a)\ar[r,Rightarrow,"\chi"{name=ar1}] & \Ha(1\otimes a) \ar[d,Rightarrow,"\Ha(\lambda)"]\\
1' \otimes' \Ha(a) \ar[u,Rightarrow,"\iota\otimes'\id"] \ar[r,swap,Rightarrow, "\lambda'"{name=ar2}] & \Ha(a)
\tarrow[shorten <=5pt,shorten >=5pt, from=ar1, to=ar2, "\ \mit\Gamma"]{}
\end{tikzcd}
$ \ \ \ \ \ and \ \ \ \ \ 
$
\begin{tikzcd}
\Ha(a)\otimes' 1'\ar[r,Rightarrow,"\id\otimes'\iota\ "{name=ar1}] & \Ha(a)\otimes' \Ha(1) \ar[d,Rightarrow,"\chi"]\\
\Ha(a) \ar[u,Rightarrow,"\rho'"]\ar[r,swap,Rightarrow,"\Ha(\rho)"{name=ar2}] & \Ha(a\otimes 1)
\tarrow[swap,shorten <=5pt,shorten >=5pt, from=ar2, to=ar1, "\ \mit\Delta "]{}
\end{tikzcd}
$
\end{center}

\begin{itemize}
\item[(c)] An invertible modification
\[
\begin{tikzcd}
& \Ha(b \otimes a) \ar[rd,Rightarrow,"\Ha(\beta)"] \tarrow[shorten <=5pt,shorten >=5pt, " \ \mit\Upsilon"]{dd} & \\
\Ha(b)\otimes' \Ha(a)\ar[ru,Rightarrow,"\chi"] \ar[rd,swap,Rightarrow,"\beta'"] & & \Ha(a\otimes b) \\
 &\Ha(a)\otimes' \Ha(b)\ar[ru,swap,Rightarrow,"\chi"]  &
\end{tikzcd}
\]
\end{itemize} 
These data are required to satisfy a long list of coherence conditions, see \cite{schommer2011classification} and references therein for details.
\end{definition}  

We finally come to the central concept in this paper. In contrast to the definition given in~\cite{schommer2011classification}, we require the appearing modifications to be invertible. However, the 2-morphisms corresponding to the underlying natural transformations are not invertible in our definition, so our definition is also weaker than the definition given in~\cite{schommer2011classification}.  

\begin{definition}
\label{Def:2sym transformation}
A \underline{natural symmetric monoidal 2-transformation} between symmetric monoidal 2-functors $\Ha ,\Ka \colon \mathscr{B} \rightarrow \mathscr{B}'$ consists of a natural 2-transformation $\theta \colon \Ha  \Rightarrow \Ka $ of the underlying 2-functors together with invertible modifications
\[
\begin{tikzcd}
 & \Ha (a\otimes b) \ar[rd,Rightarrow,"\theta"] & \\
\Ha (a)\otimes' \Ha (b) \ar[ru,Rightarrow,"\chi_\Ha"] \ar[d,swap,Rightarrow, "\theta\otimes'\id"] & & \Ka(a\otimes b) \\
 \Ka(a) \otimes' \Ha (b) \ar[rr,swap,Rightarrow,"\id\otimes'\theta"] & \tarrow[shorten <=5pt,shorten >=5pt,swap, " \ \mit\Pi"]{uu} & \Ka(a) \otimes' \Ka (b) \ar[u,swap,Rightarrow,"\chi_\Ka"]
\end{tikzcd}
\]
and
\[
\begin{tikzcd}
1' \ar[rr,Rightarrow,"\iota_\Ka"] \ar[dr,swap,Rightarrow,"\iota_\Ha"] & \tarrow[shorten <=2pt,shorten >=2pt, " \ M"]{d} & \Ka(1) \\
 & \Ha (1) \ar[ur,swap, Rightarrow, "\theta"] &  
\end{tikzcd}
\]
which satisfy the following coherence conditions expressed as equalities between 2-morphisms (omitting tensor product symbols on objects and 1-morphisms to streamline the notation):
\footnotesize
\begin{equation*}
\begin{tikzcd}
 &\Ka (a)\big(\Ka (b)\Ha (c)\big) \ar[r,"\theta"] & \Ka(a)\big(\Ka( b)\Ka( c)\big) \ar[r,"\alpha'"] & \big((\Ka( a)\Ka (b)\big)\Ka (c) \ar[dr,"\chi_\Ka"] & \\
\big(\Ka (a)\Ka (b)\big)\Ha (c) \ar[ur,"\alpha'"] \ar[urrr,bend right=10,"\theta"] \ar[rd,swap,"\chi_\Ka"] & \ar[u, Leftarrow,swap, "\;\; ",shorten <=2pt,shorten >=2pt] & & & \Ka (ab)\Ka (c)\ar[dd,"\chi_\Ka"]\\
 & \Ka (ab)\Ha (c) \ar[rrru, bend right=10,"\theta"] & \ar[ddd, Rightarrow, "\ \mit\Pi",shorten <=15pt,shorten >=15pt] \ar[uu, Leftarrow,swap, "\;\; ",shorten <=5pt,shorten >=30pt, pos=0.3] & & \\
\big(\Ka (a)\Ha( b)\big)\Ha (c) \ar[uu,"\theta"]  &\ar[l, Leftarrow,shorten <=10pt,shorten >=10pt, " \mit\Pi \otimes' \id"] & & & \Ka \big((ab) c\big) \ar[dd,"\Ka(\alpha)"] \\
 & \Ha (ab)\Ha (c) \ar[uu,swap,"\theta"]\ar[rd,"\chi_\Ha"] \ar[dd, Rightarrow, "\ \mit\Omega\!_\Ha",shorten <=10pt,shorten >=10pt] & &\ar[dd, Rightarrow, "\ \theta_{\alpha}",shorten <=10pt,shorten >=10pt] & \\
\big(\Ha (a) \Ha (b)\big)\Ha (c)\ar[uu,"\theta"] \ar[ur,"\chi_\Ha"] \ar[dr,swap,"\alpha'"] & & \Ha \big((ab) c\big)\ar[rd,"\Ha(\alpha)"]\ar[rruu,"\theta"] & & \Ka \big(a(bc)\big) \\
 &\Ha (a) \big(\Ha (b)\Ha (c)\big)\ar[r,swap,"\chi_\Ha"] & \Ha (c)\Ha (bc) \ar[r,swap,"\chi_\Ha"] & \Ha (a(bc))\ar[ru,swap,"\theta"] &
\end{tikzcd}
\end{equation*}
\normalsize
\LARGE
\[
\parallel
\]
\normalsize
\begin{equation}
\label{EQ:1 Definition s.m. transformation}
\begin{footnotesize}
\begin{tikzcd}
 &\Ka (a)\big(\Ka (b)\Ha (c)\big) \ar[r,"\theta"] & \Ka (a)\big(\Ka (b)\Ka (c)\big) \ar[ddddd, Rightarrow,shorten <=20pt,shorten >=20pt,"\ \id\otimes'\mit\Pi"] \ar[r,"\alpha'"] \ar[rddd, bend right=15,swap,"\chi_\Ka"] & \big(\Ka( a)\Ka (b)\big)\Ka (c) \ar[dr,"\chi_\Ka"] \ar[d,"\alpha'"] & \\
\big(\Ka (a)\Ka (b)\big)\Ha (c) \ar[ur,"\alpha'"] &  &\ar[r,Leftarrow, swap,shorten <=30pt,shorten >=5pt] &\Ka (a)\big(\Ka( b)\Ka (c)\big)  \ar[dd,"\chi_\Ka"] & \Ka (ab)\Ka (c)\ar[dd,"\chi_\Ka"]\ar[ddl, Rightarrow,shorten <=10pt,shorten >=10pt,"\ \mit\Omega\!_\Ka"]\\
 & \ar[ld, Leftarrow,shorten <=10pt,shorten >=10pt,swap, "\alpha'{}^\star"] & &  & \\
\big(\Ka (a)\Ha (b)\big)\Ha (c) \ar[uu,"\theta"] \ar[rd,"\alpha'"]  & & &\Ka (a)\Ka (bc) \ar[rdd,"\chi_\Ka"] \ar[ddd,Rightarrow,shorten <=10pt, shorten >=10pt, "\ \mit\Pi"] & \Ka \big((ab)c\big) \ar[dd,"\Ka(\alpha)"] \\
 & \Ka (a)\big(\Ha (b)\Ha (c)\big)\ar[uuuu,"\theta"] \ar[rdd, Rightarrow,shorten <=20pt,shorten >=20pt,"\Phi_{\otimes'}\ ",swap] \ar[rd,"\chi_\Ha"]& &  & \\
\big(\Ha (a)\Ha (b)\big)\Ha (c) \ar[uu,"\theta"] \ar[dr,swap,"\alpha'"] \ar[ru, Rightarrow,shorten <=10pt,shorten >=10pt, "\alpha'{}^\star",pos=0.6] & &\Ka (a)\Ha (bc)\ar[uur,"\theta"] & & \Ka \big(a(bc)\big) \\
 &\Ha (a)\big(\Ha (b)\Ha (c)\big)\ar[r,swap,"\chi_\Ha"]\ar[uu,"\theta"] & \Ha (a)\Ha (bc)\ar[u,swap,"\theta"] \ar[r,swap,"\chi_\Ha"] & \Ha \big(a(bc)\big)\ar[ru,swap,"\theta"] & 
\end{tikzcd}
\end{footnotesize}
\end{equation}

\begin{equation*}
\begin{tikzcd}
& \Ka (1)\Ha (a) \ar[rr,"\theta"] & \ar[d, Rightarrow, "\ \mit\Pi"] &\Ka (1)\Ka (a) \ar[rd,"\chi_\Ka"]  & \\
\Ha (1)\Ha (a)\ar[ur,"\theta"]\ar[rr, "\chi_\Ha"{name=ar1}]& &\Ha (1a)\ar[rr, "\theta"{name=ar2}] \ar[d,"\Ha(\lambda)"]  & & \Ka (1a)\ar[d,"\Ka(\lambda)"] \\
1'\Ha (a) \ar[u,"\iota_\Ha"]\ar[rr,swap, "\lambda'"{name=ar3}] \ar[rrd,swap,"\theta"]& & \Ha (a) \ar[rr, swap, "\theta"{name=ar4}] \ar[d, Rightarrow, "\ \lambda'_{\theta_a}"]  & & \Ka( a) \\
 & &1'\Ka (a)\ar[rru,swap, "\lambda'"] & &
\ar[from=ar1, to=ar3, Rightarrow, "\ \mit\Gamma\!_\Ha",shorten <=2pt,shorten >=5pt]
\ar[from=ar2, to=ar4, Rightarrow, "\ \theta",shorten <=2pt,shorten >=5pt]
\end{tikzcd}
\end{equation*}
\LARGE
\[
\parallel
\]
\normalsize
\begin{equation}
\label{EQ:2 Definition s.m. transformation}
\begin{tikzcd}
& \Ka (1)\Ha (a) \ar[rr,"\theta"] \ar[rddd, Rightarrow, shorten <=10pt,shorten >=10pt, "\ \Phi_{\otimes'}"] & &\Ka (1)\Ka(a) \ar[rd,"\chi_\Ka"]  & \\
\Ha (1)\Ha (a)\ar[ur,"\theta"]&\ar[l,Leftarrow, "{\footnotesize \ \ \  M^{-1}\otimes'\id}",shorten <=13pt, pos=0.8] & & & \Ka (1a)\ar[d,"\Ka(\lambda)"] \\
1'\Ha (a) \ar[u,"\iota_\Ha"]\ar[ruu, swap,bend right,"\iota_\Ka"] \ar[rrd,swap,"\theta"]& &  &\ar[uu, Leftarrow,shorten >=10pt,swap, "\ \mit\Gamma\!_\Ka",pos=0.2] & \Ka (a) \\
 & &1'\Ka (a)\ar[rru,swap,"\lambda'"]\ar[ruuu,"\iota_\Ka"] & &
\end{tikzcd}
\end{equation}

\begin{equation*}
\begin{tikzcd}
& \Ka (a)1' \ar[r,"\iota_\Ha"{name=ar1}] & \Ka (a)\Ha (1)\ar[dr,"\theta"] & \\
\Ka (a) \ar[ru,"\rho'"] \ar[r,Rightarrow,shorten <=5pt,shorten >=5pt, "\rho'_\theta"] & \Ha (a) 1'\ar[u,"\theta"] \ar[r,swap,"\iota_\Ha"{name=ar2}] & \Ha (a)\Ha (1)\ar[rd, Rightarrow,shorten <=5pt,shorten >=5pt,"\ \mit\Pi"]  \ar[d,"\chi_\Ha"] \ar[u,swap,"\theta"]  & \Ka (a)\Ka (1) \ar[d,"\chi_\Ka"] \\
\Ha (a)\ar[u,"\theta"]\ar[rd,swap,"\theta"] \ar[ru,"\rho'"] \ar[rr,"\Ha(\rho)"]& \ar[u,Leftarrow,shorten <=10pt,shorten >=2pt,swap, "\ {\mit\Delta}_\Ha^{-1}",pos=0.6]\ar[d,Rightarrow,shorten <=2pt,shorten >=2pt,"\ \theta_\rho"] &\Ha (a1)\ar[r,"\theta"]  & \Ka (a1) \\
& \Ka (a) \ar[rru,swap,"\Ka(\rho)"]&  & 
\ar[from=ar1, to=ar2, Rightarrow,shorten <=5pt,shorten >=5pt,"\ \iota_\theta"]
\end{tikzcd}
\end{equation*}
\LARGE
\[
\parallel
\]
\normalsize
\begin{equation}
\label{EQ:3 Definition s.m. transformation}
\begin{tikzcd}
& \Ka (a)1' \ar[rrd, swap, bend right=15,"\iota_\Ka"{name=ar1}] \ar[r,"\iota_\Ha"] & \Ka (a)\Ha (1)\ar[dr,"\theta"]  & \\
\Ka (a) \ar[ru,"\rho'"]\ar[rrrd,swap,"\Ka(\rho)"{name=ar2}] & \ar[dd,Rightarrow,shorten <=10pt,shorten >=10pt,"\ \id", pos=0.6] & \ar[u,Leftarrow, swap,shorten <=5pt,shorten >=2pt,"\ \id\otimes'M^{-1}"] & \Ka (a)\Ka (1) \ar[d,"\chi_\Ka"] \\
\Ha (a)\ar[u,"\theta"]\ar[rd,swap,"\theta"]&  &  & \Ka (a1) \\
& \Ka (a) \ar[rru,swap,"\Ka(\rho)"]&  & 
\ar[from=ar1, to=ar2, Rightarrow,shorten <=2pt,shorten >=10pt," \!{\mit\Delta}_\Ka^{-1}", pos=0.2]
\end{tikzcd}
\end{equation}
and
\begin{equation}
\label{EQ:4 Definition s.m. transformation}
\begin{tikzcd}
 \Ha (b) \Ha (a) \ar[rr,"\chi_\Ha"] & \ar[d, Leftarrow, swap,shorten <=2pt,shorten >=2pt, "{\mit\Upsilon}_\Ha^{-1}\ " ] & \Ha (ba)\ar[d,"\theta"]  \\
\Ha (b)\Ha (a) \ar[d,swap,"\theta\circ'\beta'"]\ar[u,"\id"]\ar[r,"\chi_\Ha\circ'\beta'"] & \Ha (ab)\ar[ur,"\Ha(\beta)"]\ar[dr,"\theta"] \ar[r, Rightarrow,shorten <=2pt,shorten >=2pt,"\theta_\beta"]  & \Ka (ba) \\
\Ka (a)\Ka (b) \ar[rr,swap,"\chi_\Ka"]& \ar[u, Rightarrow, shorten <=2pt,shorten >=2pt, "\mit\Pi\ "] & \Ka (ab)\ar[u,swap,"\Ka(\beta)"] 
\end{tikzcd}
 \ \ = \ \ 
\begin{tikzcd}
 \Ha (a)\Ha (b) \ar[rr,"\chi_\Ha\circ'\beta'"]\ar[dr,"\theta\circ'\beta'"] & \ar[d, Leftarrow,shorten <=3pt,shorten >=3pt, swap,"\mit\Pi\ "] & \Ha (ba)\ar[d,"\theta"]  \\
\Ha (a)\Ha (b) \ar[d,swap,"\theta"]\ar[u,"\id"] \ar[r, Rightarrow,shorten <=2pt,shorten >=2pt, "\beta'_{\theta\otimes' \theta} "] & \Ka (b)\Ka (a) \ar[r,"\chi_\Ka"] & \Ka (ba)   \\
\Ka (a)\Ka (b) \ar[rr,swap,"\chi_\Ka"]\ar[ur,"\beta'"]& \ar[u, Rightarrow,shorten <=2pt,shorten >=2pt, "{\mit\Upsilon}_\Ka^{-1}\ "] & \Ka (ab)\ar[u,swap,"\Ka(\beta)"] 
\end{tikzcd}
\end{equation}
In \eqref{EQ:1 Definition s.m. transformation}, the unlabelled 2-morphisms in the first diagram are constructed from naturality of $\alpha^\star$ and 2-functoriality of $\otimes$, while the unlabelled 2-morphism in the second diagram is induced by the equivalence $\alpha^\star \circ\alpha\Rightarrow \id$. 
\end{definition}

\begin{definition}
A \underline{symmetric monoidal modification} between two symmetric
monoidal 2-trans{-}formations $\theta, \theta'\colon
\mathcal{H}\Rightarrow \mathcal{K}$ consists of a modification $m
\colon \theta \Rrightarrow \theta'$ of the underlying natural 2-transformations satisfying
\begin{equation*}
\begin{footnotesize}
\begin{tikzcd}
\mathcal{H}(a)\otimes' \mathcal{H}(b)\ar[dd, swap, "{\theta \otimes' \theta}"] \ar[rr,"{\chi_\mathcal{H}}"] & & \mathcal{H}(a\otimes b)\ar[dd, "{\theta}"] \ar[dd, bend left, "{\theta'}",out=90, in=90] & \\ 
\; \ar[rr, Rightarrow, "{\mit\Pi}",shorten >=0.5cm,shorten <=0.5cm] & \; & \; \ar[r, Rightarrow, "m",shorten >=0.3cm,shorten <=0.2cm, pos=0.45] & \; \\
\mathcal{K}(a)\otimes' \mathcal{K}(b) \ar[rr,swap,"{\chi_\mathcal{K}}"] & & \mathcal{K}(a\otimes b) &
\end{tikzcd}
 \ \ = \ \ 
\begin{tikzcd}
& \mathcal{H}(a)\otimes' \mathcal{H}(b)\ar[dd, "{\theta' \otimes' \theta'}"] \ar[dd, bend right ,swap,"{\theta \otimes' \theta}", out=-90, in=-90] \ar[rr,"{\chi_\mathcal{H}}"] & & \mathcal{H}(a\otimes b)\ar[dd, "{\theta'}"]   \\ 
\; \ar[r, Rightarrow, "m\otimes m",shorten <=0.4cm, pos=0.7] & \; \ar[rr, Rightarrow, "{\mit\Pi'}",shorten >=0.5cm,shorten <=0.8cm] & \; & \;  \\
& \mathcal{K}(a)\otimes' \mathcal{K}(b) \ar[rr,swap,"{\chi_\mathcal{K}}"] & & \mathcal{K}(a\otimes b) 
\end{tikzcd}
\end{footnotesize}
\end{equation*}
and
\begin{equation*}
\begin{tikzcd}
 & & \mathcal{H}(1) \ar[dd, "\theta"] \ar[dd, bend left, "{\theta'}",out=90, in=90] & \\
 1' \ar[rru, "{\iota_\mathcal{H}}"] \ar[rrd, swap,"{\iota_\mathcal{K}}"] \ar[rr, Rightarrow, "M",shorten >=0.2cm,shorten <=0.8cm, pos=0.65]& &\; \ar[r, Rightarrow, "m",shorten >=0.3cm,shorten <=0.2cm, pos=0.45] &\; \\
  & & \mathcal{K}(1) &
\end{tikzcd}
 \ \ = \ \ 
\begin{tikzcd}
 & & \mathcal{H}(1) \ar[dd, "\theta'"]  & \\
 1' \ar[rru, "{\iota_\mathcal{H}}"] \ar[rrd, swap,"{\iota_\mathcal{K}}"] \ar[rr, Rightarrow, "M'",shorten >=0.2cm,shorten <=0.8cm, pos=0.65]& &\;  &\; \\
  & & \mathcal{K}(1) &
\end{tikzcd}
\end{equation*}
\end{definition}


\end{document}